\documentclass[sigconf]{acmart}

\def\fullversion{0}

\usepackage{caption}

\PassOptionsToPackage{table}{xcolor}
\usepackage{colortbl}
\usepackage{style}
\usepackage{bigstrut}
\usepackage{microtype}
\usepackage{tabularx}
\usepackage{multirow}
\usepackage{multicol}
\usepackage{rotating}
\usepackage{lipsum}
\usepackage{balance}
  \newcommand{\BDPQ}{{\textsc{LaB-PQ}}\xspace}
  \newcommand{\BDPQfull}{\textbf{La}zy-\textbf{B}atched \textbf{P}riority \textbf{Q}ueue\xspace}
  \newcommand{\PQ}{\mathbb{PQ}\xspace}

  \newcommand{\tourt}{tournament tree\xspace}
  \newcommand{\wtupdate}{\mf{Sync}\xspace}
  \newcommand{\updated}{\mb{renew}}
  \newcommand{\wtleft}{\mb{left}}
  \newcommand{\wtright}{\mb{right}}
  \newcommand{\wtparent}{\mb{parent}}
  
  \newcommand{\lazyinsert}{\mf{Update}\xspace}
  \newcommand{\lazyinsertbold}{\mf{\textbf{Update}}}
  \newcommand{\boldlazyinsert}{\textbf{\textsc{Update}}}
  \newcommand{\id}{{\mathit{id}}\xspace}
  \newcommand{\record}{record\xspace}
  \newcommand{\records}{records\xspace}
  \newcommand{\bool}{\mathit{Bool}\xspace}
  \newcommand{\mapping}{mapping function\xspace}

  \newcommand{\krhograph}{$(k,\rho)$-graph\xspace}
  \newcommand{\deltastepping}{$\Delta$-Stepping\xspace}
  \newcommand{\deltastarstepping}{$\Delta^{*}$-Stepping\xspace}
  
  \newcommand{\SSSPalgo}{$\rho$-Stepping\xspace}
  \newcommand{\ssspours}{$\rho$-Stepping\xspace}
  \newcommand{\deltas}{$\Delta$-Stepping\xspace}
  \newcommand{\deltass}{$\Delta^*$-Stepping\xspace}
  \newcommand{\radiuss}{Radius-Stepping\xspace}
  \newcommand{\ourbf}{\textit{PQ-BF}\xspace}
  \newcommand{\ourdelta}{\textit{PQ-$\Delta^{*}$}\xspace}
  \newcommand{\ourrho}{\textit{PQ-$\rho$}\xspace}

  \newcommand{\addmark}{\mf{Mark}\xspace}
  \newcommand{\addmarkbold}{\mf{\textbf{Mark}}\xspace}
  \newcommand{\extract}{\mf{Extract}\xspace}
  \newcommand{\boldextract}{\textbf{\textsc{Extract}}\xspace}
  \newcommand{\traverse}{\mf{ExtractFrom}\xspace}
  \newcommand{\collect}{\mf{Reduce}\xspace}
  \newcommand{\extcond}{\mf{ExtDist}\xspace}
  \newcommand{\finishcond}{\mf{FinishCheck}\xspace}
  \newcommand{\boldextcond}{\mf{\bf ExtDist}\xspace}
  \newcommand{\boldfinishcond}{\mf{\bf FinishCheck}\xspace}
  \newcommand{\seq}{\mathit{seq}\xspace}
  \newcommand{\inQ}{\mathit{inQ}\xspace}

  \newcommand{\codeskip}{{\vspace{.05in}}}
  \newcommand{\para}[1]{{\bf \emph{#1}}\,}
  \newcommand{\mathtext}[1]{{\mathit{#1}}}
  \newcommand{\emp}[1]{\emph{\textbf{#1}}}

  \newcommand{\variablename}[1]{{\texttt{#1}}}

  \newcommand{\forkins}{\texttt{fork}}

  \newcommand{\thread}{thread}

  \newcommand{\depth}{span}
  \newcommand{\TAS}[0]{$\mf{TestAndSet}$}
  
  \newcommand{\tas}{{\texttt{test\_and\_set}}}
  
  \newcommand{\cas}{{\texttt{compare\_and\_swap}}}
  
  \newcommand{\WriteMin}{\mf{WriteMin}\xspace}

  \newcommand{\R}{\mathbb{R}}
  \newcommand{\N}{\mathbb{N}}
  \newcommand{\Z}{\mathbb{Z}}

  \newcommand{\true}{\emph{true}}
  \newcommand{\false}{\emph{false}}

  \newcommand{\ifconference}{{{\ifx\fullversion\undefined}}}

\makeatletter
\def\dfnt@space@setup{%
  \dfnt@preskip=\parskip
    \dfnt@postskip=0pt}
\makeatother

\newtheoremstyle{exampstyle}
  {.05in} 
  {.05in} 
  {} 
  {.5em} 
  {\sc \bfseries} 
  {.} 
  {.5em} 
  {} 
\theoremstyle{exampstyle} \newtheorem{compactdef}{Definition}
\theoremstyle{exampstyle} 

\makeatletter
\renewenvironment{proof}[1][\proofname]{\par
  \vspace{-\topsep}
  \pushQED{\qed}%
  \normalfont
  \topsep0pt \partopsep0pt 
  \trivlist
  \item[\hskip\labelsep
        \itshape
    #1\@addpunct{.}]\ignorespaces
}{%
  \popQED\endtrivlist\@endpefalse
  \addvspace{3pt plus 3pt} 
}

\newcommand{\smallsmallskip}{\vspace{.05in}}

\crefname{section}{Sec.}{Sec.}
\crefname{theorem}{Thm.}{Thm.}
\crefname{lemma}{Lem.}{Lem.}
\crefname{corollary}{Col.}{Col.}
\crefname{table}{Tab.}{Tab.}

\copyrightyear{2021}
\acmYear{2021}
\setcopyright{acmcopyright}\acmConference[SPAA '21]{Proceedings of the 33nd ACM Symposium on Parallelism in Algorithms and Architectures}{July 6--8, 2021}{Virtual Event, USA}
\acmBooktitle{Proceedings of the 33nd ACM Symposium on Parallelism in Algorithms and Architectures (SPAA '21), July 6--8, 2021, Virtual Event, USA}
\acmPrice{15.00}
\acmDOI{10.1145/3409964.3461782}
\acmISBN{978-1-4503-6935-0/20/07}
\acmConference[SPAA 2021]{ACM Symposium on Parallelism in Algorithms and Architectures}{July 6--8, 2021}{Virtual Event, USA}

\usepackage{array}
\begin{document}
\title[Efficient Parallel Shortest-Path Algorithms]{Efficient Stepping Algorithms and Implementations for \\Parallel Shortest Paths}
  \author{Xiaojun Dong}
  \affiliation{\institution{UC Riverside}}
  \email{xdong038@cs.ucr.edu}
  \author{Yan Gu}
  \affiliation{\institution{UC Riverside}}
  \email{ygu@cs.ucr.edu}
  \author{Yihan Sun}
  \affiliation{\institution{UC Riverside}}
  \email{yihans@cs.ucr.edu}
  \author{Yunming Zhang}
  \affiliation{\institution{MIT}}
  \email{yunming@mit.edu}

\begin{abstract}
  The single-source shortest-path (SSSP) problem is a notoriously hard problem in the parallel context.
  In practice, the $\Delta$-stepping algorithm of Meyer and Sanders has been widely adopted.
  However, $\Delta$-stepping has no known worst-case bounds for general graphs, and
  the performance highly relies on the parameter $\Delta$, which requires exhaustive tuning.
  The parallel SSSP algorithms with provable bounds, such as Radius-stepping, either have no implementations available or are much slower than $\Delta$-stepping in practice.

  We propose the \emph{stepping algorithm framework} that generalizes existing algorithms such as $\Delta$-stepping and Radius-stepping. The framework allows for similar analysis and implementations for all stepping algorithms. We also propose a new abstract data type, lazy-batched priority queue (\textsc{LaB-PQ}) that abstracts the semantics of the priority queue needed by the stepping algorithms.
  We provide two data structures for \textsc{LaB-PQ}, focusing on theoretical and practical efficiency, respectively. 
  Based on the new framework and \textsc{LaB-PQ}, we show two new stepping algorithms, $\rho$-stepping and $\Delta^*$-stepping, that are simple, with non-trivial worst-case bounds, and fast in practice.
  We also show improved bounds for a list of existing algorithms such as Radius-Stepping.

  Based on our framework, we implement three algorithms: Bellman-Ford, $\Delta^*$-stepping, and $\rho$-stepping.
  We compare the performance with four state-of-the-art implementations.
  On five social and web graphs, $\rho$-stepping is 1.3--2.6x faster than all the existing implementations.
  On two road graphs, our $\Delta^*$-stepping is at least 14\% faster than existing ones, while $\rho$-stepping is also competitive.
  The almost identical implementations for stepping algorithms also allow for in-depth analyses among the stepping algorithms in practice.

\hide{
  Some potential titles:

efficient stepping algorithms and implementations for parallel shortest path

simple and efficient algorithms and implementations for parallel shortest path

Theory and implementation of stepping algorithms for parallel shortest path problem

Provable efficient parallel shortest path problems with fast implementation}
\end{abstract} 



\maketitle


\makeatletter
\newcommand{\removelatexerror}{\let\@latex@error\@gobble}
\makeatother


\section{Introduction}
\label{sec:intro}

Given a weighted graph $G=(V,E,w)$ with $n=|V|$ vertices, $m=|E|$ edges, edge weight function $w:E\to \R^+$, and a source $s\in V$, the single-source shortest-path (SSSP) problem is to find the shortest paths from $s$ to all other vertices in the graph.
In this paper, we consider general positive edge weights.
Sequentially, the best known bound for SSSP is $O(m+n\log n)$ using Dijkstra's algorithm~\cite{dijkstra1959} with Fibonacci heap~\cite{fredman1987fibonacci}.
However, SSSP is notoriously hard in parallel.
Despite dozens of papers and implementations over the past decades, all existing solutions have some unsatisfactory aspects. 

Practically, most existing parallel SSSP implementations ~\cite{zhang2020optimizing,dhulipala2017,nguyen2013lightweight,beamer2015gap} are based on \deltas~\cite{meyer2003delta},
which is a hybrid of Dijkstra's algorithm \cite{dijkstra1959} and the Bellman-Ford algorithm~\cite{bellman1958routing,ford1956network}.
It determines the correct shortest distances in increments of $\Delta$, in step~$i$ settling down all the vertices with distances in $[i\Delta, (i+1)\Delta]$.
Within each step, the algorithm runs Bellman-Ford as substeps.

Although \deltas{} is the state-of-the-art practical parallel SSSP algorithm, two challenges still remain. Theoretically, \deltas{} has been analyzed on random graphs~\cite{meyer2001single,crauser1998parallelization}, but no bounds has been shown for the general case.
Practically, the parameter~$\Delta$ can largely affect the algorithm's performance.
The best choice of $\Delta$ depends on the graph structure, weight distribution, and the implementation itself.
\cref{fig:intro:delta} shows the running time of three state-of-the-art \deltas{} implementations \cite{zhang2018graphit,dhulipala2017,nguyen2013lightweight} and our own \deltass{} (a variant of \deltas{}, see \cref{sec:framework}) with different $\Delta$ values, on real-world graphs (more details in \cref{sec:exp}).
A badly-chosen $\Delta$ can greatly affect the performance, and the best choices of $\Delta$ are very inconsistent for different graphs (even with the same weight distribution) and implementations.
Hence, in practice, one needs exhausting searches for $\Delta$ in the parameter space as preprocessing.

Theoretically, there has been a rich literature of parallel SSSP algorithms~\cite{ullman1991high,klein1997randomized,cohen1997using,Shi1999,cohen2000polylog,spencer1997time,blelloch2016parallel} with $o(nm)$ work and $o(n)$ span (critical path length). Most of these algorithms rely on adding shortcuts to achieve the bounds.
While these algorithms are inspiring, none of them have implementations or show practical advantages over \deltas{} on real-world graphs.
We believe that one potential reason is the use of shortcuts and hopsets.
To achieve $O(n^{1-\epsilon})$ span, these algorithms need to add $\Omega(n^{1+\epsilon})$ shortcuts.
More shortcut edges contribute to more work, memory usage and footprint, hiding the advantages in the span improvement.

\ifx\fullversion\undefined
There has also been prior work on parallelizing the priority queue in Dijkstra's algorithm~\cite{brodal1998parallel,deo1992parallel,alistarh2015spraylist,sundell2005fast,linden2013skiplist,shavit2000skiplist,liu2012lock,henzinger2013quantitative,zhou2019practical,calciu2014adaptive,bingmann2015bulk,sanders2019sequential,sanders1998randomized,sanders2000fast}. 
However, they do not provide interesting worst-case work and span bounds, or better performance than \deltas{} in practice.

We summarize existing work on parallel SSSP in \cref{sec:related}.
\else
There has also been prior work on parallelizing the priority queue in Dijkstra's algorithm.
Early work on PRAM~\cite{brodal1998parallel,deo1992parallel} parallelizes priority queue operations, but the worst-case span bound is still $\Theta(n)$.
Other previous papers consider concurrent~\cite{alistarh2015spraylist,sundell2005fast,linden2013skiplist,shavit2000skiplist,liu2012lock,henzinger2013quantitative,zhou2019practical,calciu2014adaptive}, external-memory or other settings~\cite{bingmann2015bulk,sanders2019sequential,sanders1998randomized,sanders2000fast}.
They do not provide interesting worst-case work and span bounds, or better performance than \deltas{} in practice.

We summarize existing work on parallel SSSP in \cref{sec:related}.
\fi

\begin{figure*}
  \centering
  \vspace{-1em}
  \begin{tabular}{cccc}
    \multicolumn{4}{c}{\includegraphics[width=\columnwidth]{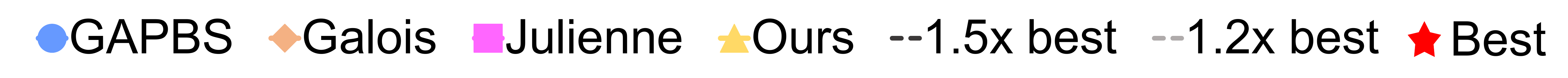}}\\
    \includegraphics[width=0.5\columnwidth]{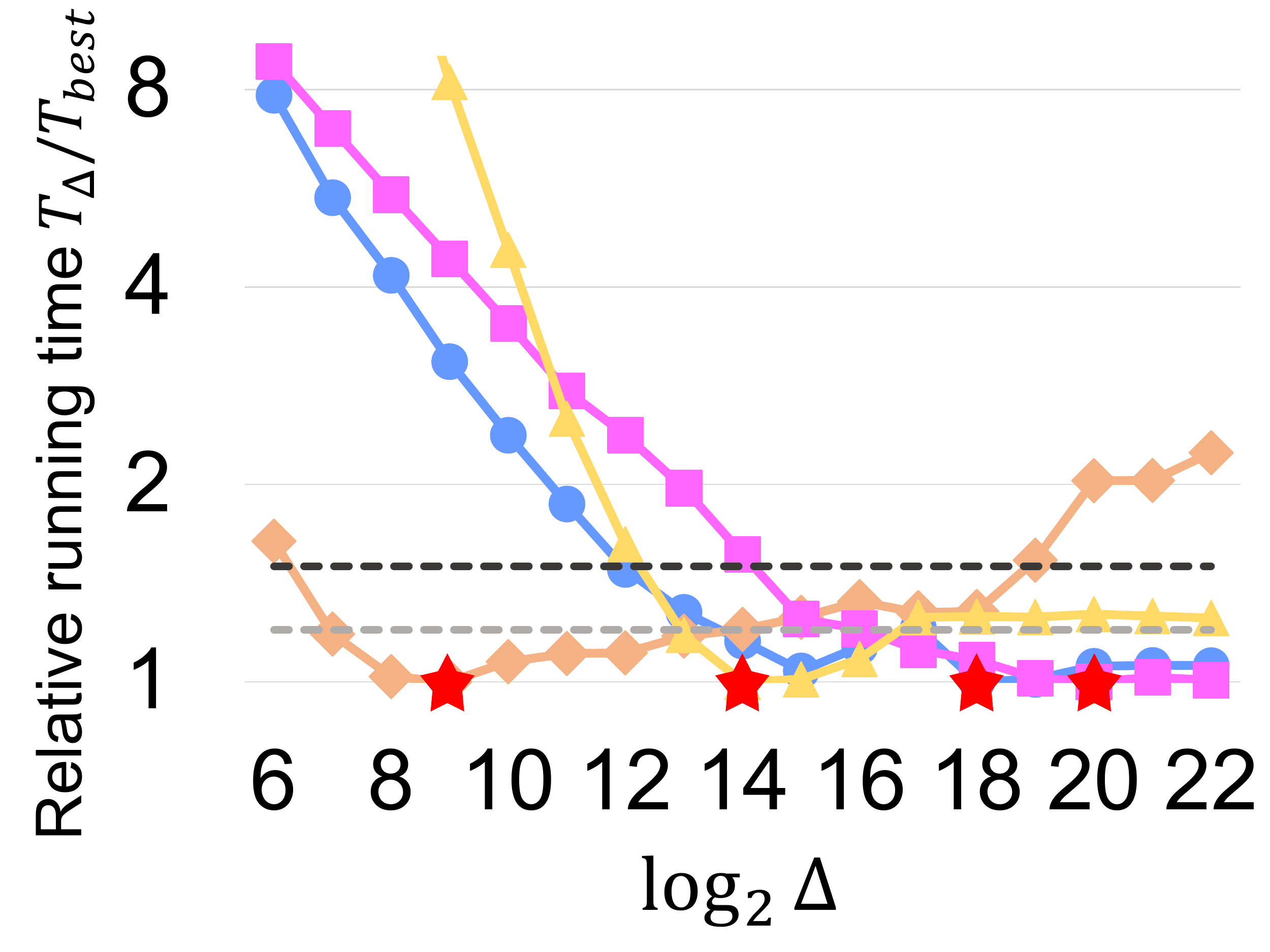} & \includegraphics[width=0.5\columnwidth]{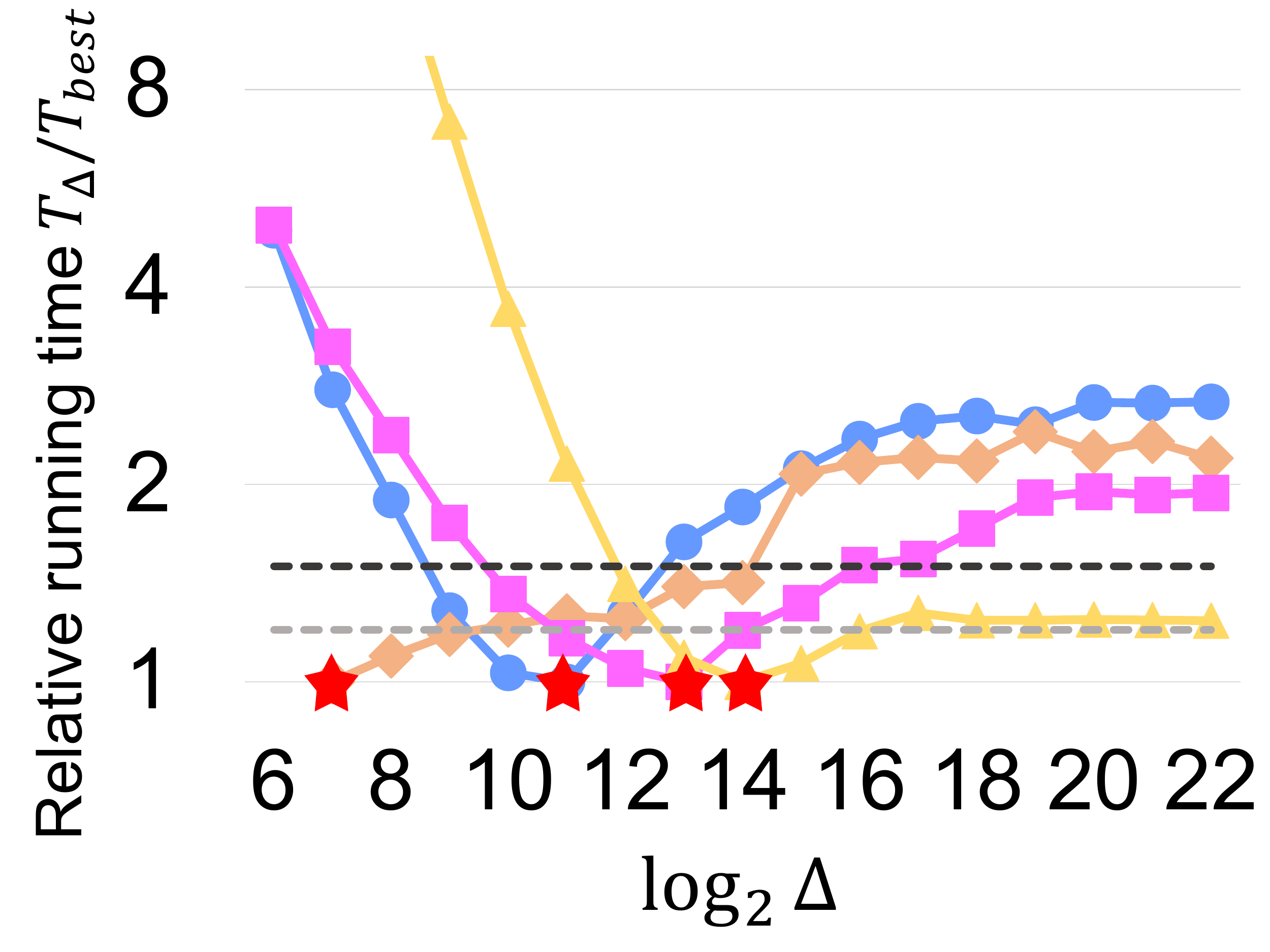} & \includegraphics[width=0.5\columnwidth]{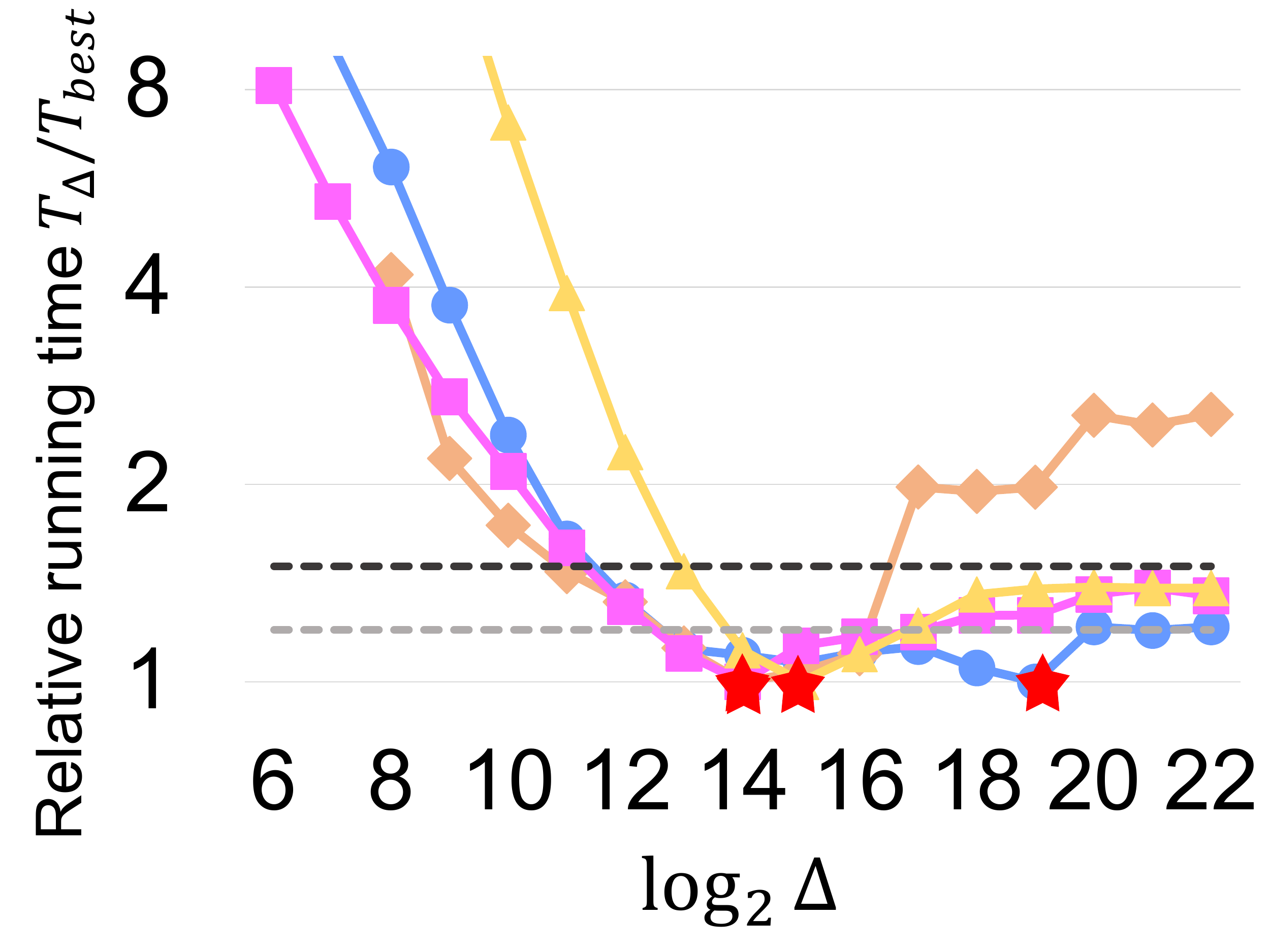} &
    \includegraphics[width=0.5\columnwidth]{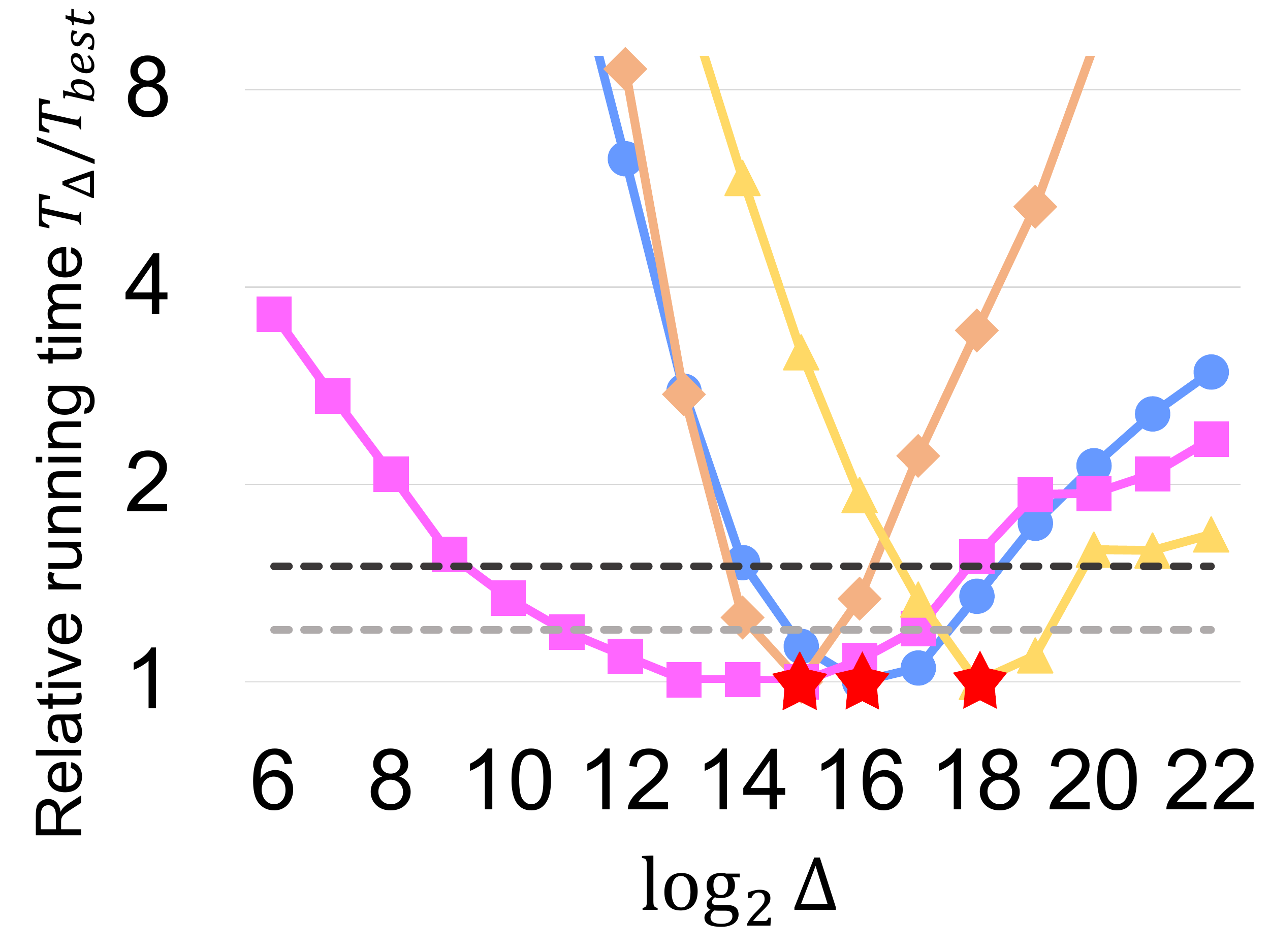}\\
    \small (a). Twitter (TW) &\small (b). Friendster (FT) &\small (c). WebGraph (WB) &\small (d). Road USA (USA) \\
  \end{tabular}
  \caption{\small \textbf{$\Delta$-stepping relative running time with varying $\Delta$},  \mdseries including social networks (Twitter and Friendster), web graph (WebGraph), and road network (Road USA). A complete version with seven graphs is presented in \ifx\fullversion\undefined the full version of this paper \cite{ssspfull}.
\else
\cref{fig:alldelta}.
\fi
  We use 96 cores (192 hyperthreads).
  We vary $\Delta$ and report the running time divided by the best running time across all $\Delta$ values.  The best choice of $\Delta$ for each implementation is marked as a red star.  
  \ifx\fullversion\undefined \else We have the following interesting findings. (1). On the same graph, the best delta can be very different for different implementations (e.g., on Twitter, Julienne's best $\Delta$ is $2^{12}$ times larger than Galois's). The best value of delta for one algorithm can make another implementation much slower (e.g., Galois's best $\Delta$ on Friendster makes all other implementations more than $4\times$ slower). The selection of $\Delta$ for one \deltas implementation does not generalize to other \deltas implementations. (2) For each implementation, the best choices of $\Delta$ vary a lot on different graphs ($2^8$ for GAPBS), although they have similar edge weight range and distribution. (3). On the same graph, the performance is very sensitive to the value of $\Delta$. Usually, $2$--$4\times$ off may lead to a 20\% slowdown, and $4$--$8\times$ off may lead to a 50\% slowdown. A badly-chosen parameter delta can largely affect the performance. (4). For the same implementation, on different graphs, the performance variance changing with $\Delta$ can be unstable. For example, Galois has very stable performance across $\Delta$ values on com-orkut (see \cref{fig:alldelta}) and Twitter, but is very unstable on other graphs. Thus, the stable performance on one graph does not guarantee that we can avoid searching the full parameter space for another graph.\fi
  \label{fig:intro:delta}}
\vspace{-2.5em}
\end{figure*}
\begin{figure}
  \includegraphics[width=\columnwidth]{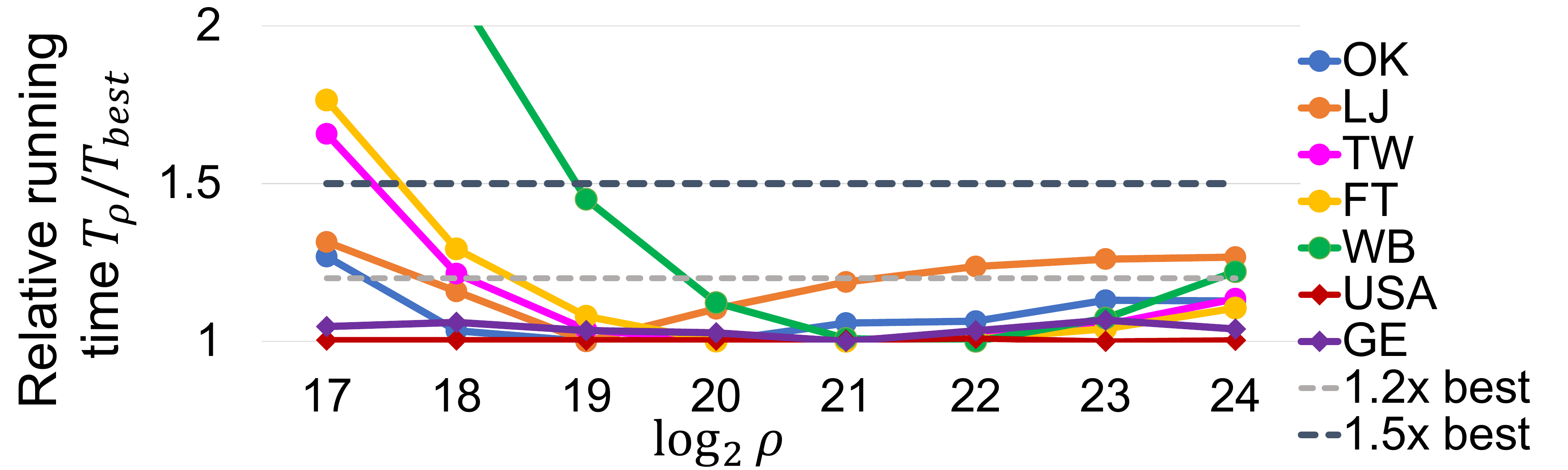}
  \caption{\small\textbf{Relative running time of \ssspours{} with varied $\rho$.} \mdseries We use 96 cores (192 hyperthreads). We vary $\rho$ and tested the average running time on 100 random sources, and divided by the time with the best $\rho$.
  \ifx\fullversion\undefined \else
  We can see that: (1) the trends are pretty consistent among all graphs; (2) when $\rho$ is between $2^{20}$ to $2^{24}$, the performance is almost always within $1.2\times$ the best performance (except for two LJ's data points); (3) the best choices of $\rho$ are within $2^{19}$ to $2^{22}$, although the graph sizes vary by almost three orders of magnitudes; and (4) If we pick $\rho$ to be $2^{21}$, all runtimes are within 10ms off from the best cases (numbers in Table \ref{tab:alltime}). \fi
  \label{fig:rho-time}}
  \vspace{-1.8em}
\end{figure}
\begin{figure}
  \includegraphics[width=\columnwidth]{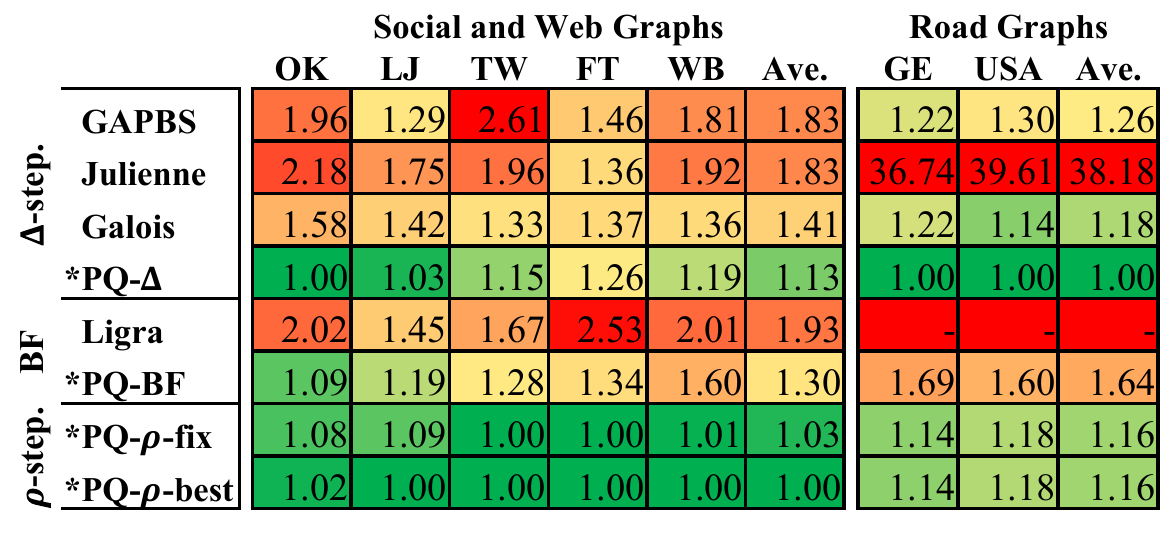}
  \vspace{-2em}
  \caption{\small \label{fig:exp:heat} \textbf{The heat map of the parallel running time relative to the fastest on each graph.} \mdseries We use 96 cores (192 hyperthreads). 
  Each column is a graph instance. ``Ave.'' gives the average numbers over five social/web graphs and two road graphs, respectively.  ``*'' denotes our implementations. \ourrho{}-fix means to use a fixed parameter $\rho$ across all graphs in \ssspours{}, and \ourrho{}-best denotes the best running time using all values of $\rho$.}
  \vspace{-1.8em}
\end{figure}

\hide{

\caption{\textbf{$\Delta$-stepping relative running time with varying $\Delta$ on different graph types},  \mdseries including social networks (Twitter and Friendster), web graph (WebGraph), and road network (Road USA). A complete version with seven graphs is presented in \cref{fig:alldelta}.
  We use 96 cores (192 hyperthreads).
  We vary $\Delta$ and report the running time divided by the best running time across all $\Delta$ values.  The best choice of $\delta$ for each implementation is marked as a red star.  We have the following interesting findings. (1). On the same graph, the best delta can be very different for different implementations (e.g., on Twitter, Julienne's best $\Delta$ is $2^{12}$ times larger than Galois's). The best value of delta for one algorithm can make another implementation much slower (e.g., Galois's best $\Delta$ on Friendster makes all other implementations more than $4\times$ slower). The selection of $\Delta$ for one \deltas implementation does not generalize to other \deltas implementations. (2) For each implementation, the best choices of $\Delta$ vary a lot on different graphs ($2^8$ for GAPBS), although they have similar edge weight range and distribution. (3). On the same graph, the performance is very sensitive to the value of $\Delta$. Usually, $2$--$4\times$ off may lead to a 20\% slowdown, and $4$--$8\times$ off may lead to a 50\% slowdown. A badly-chosen parameter delta can largely affect the performance. (4). For the same implementation, on different graphs, the performance variance changing with $\Delta$ can be unstable. For example, Galois has very stable performance across $\Delta$ values on com-orkut (see \cref{fig:alldelta}) and Twitter, but is very unstable on other graphs. Thus, the stable performance on one graph does not guarantee that we can avoid searching the full parameter space for another graph. We show more analysis in \cref{sec:exp}.\label{fig:intro:delta}}

  \caption{\textbf{Relative running time of \ssspours{} with varied $\rho$.} \mdseries We use 96 cores (192 hyperthreads). We vary $\rho$ and tested the average running time on 100 random sources, and divided by the time with the best $\rho$.  \mdseries  We can see that: (1) the trends are pretty consistent among all graphs; (2) when $\rho$ is between $2^{20}$ to $2^{24}$, the performance is always within $1.2\times$ the best performance (except for two LJ's data points); (3) the best choices of $\rho$ are within $2^{19}$ to $2^{22}$, although the graph sizes vary by almost three orders of magnitudes; and (4) If we pick $\rho$ to be $2^{21}$, all runtimes are within 10ms off from the best cases (numbers in Table \ref{tab:alltime}). \label{fig:rho-time}}

  \caption{ \label{fig:heat} \textbf{The heat map of the parallel running time relative to the fastest on each graph.} \mdseries We use 96 cores (192 hyperthreads). Table \ref{tab:alltime} gives the exact running time.  \ourdelta{}, \ourbf{} and \ourrho{} (noted with $*$) are our implementations.  ``Ave.'' gives the average numbers on social/web graphs (average over the five graphs) and road graphs (average over the two graphs), respectively.  The other columns are graph instances.  \ourrho{} is the new algorithm proposed in this paper. We use \ourrho{}-fix to denote the running time of a fixed parameter $\rho$ across all graphs, and \ourrho{}-best to denote the best running time using all values of $\rho$.}
}



\myparagraph{Our approach}.
The three previous research directions on parallel SSSP (practical implementations, theoretical bounds, parallel priority queues) are mostly studied independently.
We aim to design parallel SSSP algorithms combining the advantages---as simple as those using parallel priority queues, achieving worst-case guarantees that match the existing bounds, and as fast as (or faster than)
\deltas{} in practice.
Our key algorithmic insights include three components: a \emp{stepping algorithm framework}, which abstracts general ideas in some existing parallel SSSP algorithms, an abstract data type (ADT) \emp{\BDPQfull (\BDPQ{})} with efficient implementations, which extracts the semantics of the priority queue needed by stepping algorithms, and two new stepping algorithms \emp{\ssspours{} and \deltastarstepping}, which are efficient both in theory and practice.

Our stepping algorithm framework (\cref{alg:framework}) abstracts the common idea in some existing ``stepping'' algorithms (e.g., \radiuss{} \cite{blelloch2016parallel} and \deltas{}~\cite{meyer2003delta}): in each \emph{step}, the algorithm relaxes all vertices with tentative distances within
a certain threshold, as a batch and in parallel.
The two extreme cases are the two textbook algorithms: Dijkstra's algorithm with batch size 1, and Bellman-Ford algorithm with batch size $n$.
We formalize several algorithms in this framework (\cref{tab:framework}).
Interestingly, some variants of parallel Dijkstra~\cite{bingmann2015bulk,zhou2019practical,alistarh2015spraylist} also use a similar high-level idea.

The proposed ADT \BDPQ{} abstracts the priority queue needed by the stepping algorithms.
It supports \lazyinsert{} to commit an update to the data structure, which can be lazily batched and executed in parallel. It also supports \extract{} to return all \record{s} with keys within a certain threshold in parallel.
The \BDPQ{} is inspired by the recent work on \emph{batch-dynamic data structures}~~\cite{shun2014phase,acar2019parallel,anderson2020work,tseng2019batch,Blelloch2016justjoin,sun2018pam}, where multiple updates or queries are applied to the data structure in batches in parallel. 
One advantage of \BDPQ{} is that we do not explicitly generate the batches, but do it \emph{lazily}.
On top of the ADT, all stepping algorithms can easily use \BDPQ's interface as a black box. 
Underneath it, we provide efficient data structures to support \BDPQ{}. We show a theoretically efficient implementation of \BDPQ{} based on the tournament tree (\cref{sec:tour-tree}). It improves the cost bounds for existing parallel SSSP algorithms such as Radius-Stepping~\cite{blelloch2016parallel} and Shi-Spencer~\cite{Shi1999}.
In practice, we show simple implementations based on flat arrays, which makes our stepping algorithms outperform state-of-the-art software~\cite{zhang2020optimizing,dhulipala2017,nguyen2013lightweight,beamer2015gap}. 

Based on the stepping algorithm framework and \BDPQ{}, we also propose a new parallel SSSP algorithm, referred to as \emph{\SSSPalgo}, which is simple and efficient both in theory and in practice.
The high-level idea of \SSSPalgo is to relax a fixed number of unsettled vertices with small tentative distances in each step.
While a similar (but not the same) idea have been used in some parallel Dijkstra's algorithms~\cite{bingmann2015bulk,zhou2019practical,alistarh2015spraylist}, none of them have interesting bounds or practical performance comparable to \deltas.
In this paper, we formally analyze \SSSPalgo and show work and span bounds.
\ssspours{} achieves a better span bound than \radiuss with a slightly higher work bound (\cref{thm:rho-steps}).
The work bound also applies to directed graphs (the bounds for \radiuss only holds for undirected graphs).
Practically, our \SSSPalgo is 1.3-2.6$\times$ faster than previous implementations on social and web graphs, and is competitive on road graphs (\cref{fig:exp:heat}).

In addition to theoretical guarantees and practical performance,
another advantage of \ssspours{} is that, it needs no preprocessing (e.g., adding shortcuts in \radiuss{}) or time-consuming parameter searching (e.g., finding best $\Delta$ in \deltas{}).
Our experiments (\cref{fig:rho-time}) show that, the best choice of $\rho$ is consistent and insensitive across the real-world graphs we tested.

Inspired by the stepping algorithms and \BDPQ{}, we also show \deltastarstepping{}, a variant of \deltastepping, which is simple, has non-trivial worst-case bounds (\cref{tab:bounds}), and fast in practice (\cref{fig:exp:heat}).

\hide{
\yan{The following paragraph needs revision}
\SSSPalgo achieves comparable or better performance as \deltas{}, and similar work-span tradeoff in theory compared to \radiuss{}, while also avoids expensive preprocessing or parameter searching in \deltas or \radiuss.
Compared to \radiuss, \SSSPalgo has the same span bound and is only off by a little in the work bound, and is more practical.
Compared to \deltas, \SSSPalgo is preprocessing-free, faster in almost all tested graphs, and has better theoretical guarantees on undirected graphs.
}
\myparagraph{Our Contributions.} 
Combining our \BDPQ{} with existing algorithms and our new algorithms, we achieve new bounds and efficient implementations for parallel SSSP.
These results are due to the abstraction of stepping algorithm framework and \BDPQ{}, which greatly simplifies algorithm design, analysis, and implementation.

In theory, we show new bounds for \radiuss{}~\cite{blelloch2016parallel}, Shi-Spencer~\cite{Shi1999}, \deltass{}, and \ssspours{}.
We note that, with no shortcuts or hopsets, it seems unlikely to show $o(n)$ worst-case span (consider a chain).
However, tighter bounds can depend on certain graph parameters, which may exhibit a good property on real-world graphs.
For example, although parallel Bellman-Ford has worst-case span of $\tilde{O}(n)$, a more precise bound is $\tilde{O}(d)$, where $d$ is the shortest-path tree depth.
Indeed, on social networks with small $d$, parallel Bellman-Ford is reasonably fast (Table \ref{tab:alltime}).
To capture this, Blelloch et al.~\cite{blelloch2016parallel} proposed a graph invariant \krhograph{} that indicates how ``parallel'' a graph is.
Intuitively, a graph is a \krhograph{} if every vertex reaches $\rho$ nearest vertices in $k$ hops.
We extend this concept to analyze multiple stepping algorithms.
Our experiments show that the real-world social or web graphs we tested are $(\log n,O(\sqrt{n}))$-graphs, and road graphs we tested are $(\sqrt{n},O(n))$-graphs (\cref{fig:k-rho}).
Under our framework, the stepping algorithms share common subroutines in analyses, such as the extraction lemma (\cref{lem:num-ext}) and the distribution lemma (\cref{lem:distribute}).

In practice, our framework and array-based \BDPQ{} give unified implementations for Bellman-Ford, \deltass{} and \ssspours{}. 
Our implementations achieve the best performance on all graphs (see \cref{fig:exp:heat}).
On the social and web graphs, \ssspours{} is 1.3-2.6$\times$ faster than existing implementations.
On road graphs, our \deltass is consistently the fastest and \ssspours is competitive to previous ones.
This indicates the effectiveness of our framework since all optimizations are easily applicable to all algorithms.
We also provide an in-depth experimental study based on our framework, especially to understand the tradeoff between work and parallelism. We show how different stepping algorithms explore the frontier in steps (\cref{fig:visted-per-step,fig:enqueue}), the parameter space (\cref{fig:intro:delta,fig:rho-time}), and eventually draw interesting conclusions in \cref{sec:exp-sum}.


\hide{

The reason that we can achieve simplicity and both theoretical and practical efficiency for \SSSPalgo is due to the algorithmic framework and the abstraction of \BDPQ{}.
Unlike the parallel priority queue in previous work, \BDPQ is an ADT, not a specific data structure.
Theoretically, we can totally ignore how \BDPQ is implemented, which significantly simplifies the \SSSPalgo algorithm.
Hence, we can focus on the similarities of the algorithms in the stepping algorithm framework, borrow the analysis in the recent \radiuss paper~\cite{blelloch2016parallel} and apply to \SSSPalgo.
The practical efficiency of \SSSPalgo relies on the implementation of \BDPQ, and we can borrow the highly-tuned implementation for \deltas~\cite{dhulipala2017,zhang2020optimizing} and apply to our implementation.
As a proof-of-concept, the \BDPQ data type isolates and bridges the theoretical work and practical optimizations for parallel SSSP.
}


\ifx\fullversion\undefined
Due to page limit, we postpone some analysis and full experimental results to the full version of the paper \cite{ssspfull}.
\else
\fi
We summarize our contributions of this paper as follows.

\begin{itemize}[topsep=1.5pt, partopsep=0pt,leftmargin=*]
\setlength{\itemsep}{0pt}
  \item A stepping algorithm framework, which unifies multiple parallel SSSP algorithms. 
  \item A new ADT \BDPQ{} and two implementations, which are used in our analysis and implementations, respectively.
  \item A new parallel SSSP algorithm \ssspours{}, which is preprocessing-free, simple and efficient both in theory and in practice.
  \item A new variant of \deltastepping (\deltastarstepping), which is simple, with theoretical guarantee, and fast in practice.
  \item New analyses for stepping algorithms based on \krhograph{}, which include parameterized work and span bounds for \ssspours{} (\cref{thm:rho-steps}) and \deltass{} (\cref{thm:deltass}), and improved work bounds for \radiuss{} (\cref{cor:radius}) and Shi-Spencer (\cref{cor:shispencer}). 
  \item Efficient parallel implementations of Bellman-Ford, \deltass{} and \ssspours{}, which outperform existing ones (\cref{tab:alltime}).
  \item In-depth experimental study of parallel SSSP algorithms.
\end{itemize}

\hide{
The contribution is two-fold.
The first contribution is the new abstraction of \emph{batch-dynamic priority queues (BDPQ)} that bridges the different priority queue implementations and supports a variety of parallel SSSP algorithms.
This abstraction and algorithm framework are introduced in Section~\ref{sec:framework}.
Using a theoretically efficient algorithm for BDPQ introduced in Section~\ref{sec:winning-tree}, we can get improved bounds for existing algorithms.
Meanwhile, a practically efficient implementation for BDPQ gives unified and efficient implementations for \deltas and Bellman-Ford.

Based on BDPQ, the second and main contribution of this paper is a new parallel SSSP algorithm \emph{\SSSPalgo}.
\SSSPalgo has several advantages: it is simple, preprocessing-free, practically efficient, and has theoretical guarantees.
\SSSPalgo is very simple on top of BDPQ (which can be treated as a black box).
\SSSPalgo is similar to practically efficient \deltas and theoretically efficient \radiuss, but \SSSPalgo does not require the expensive preprocessing step in \deltas and \radiuss.
Compared to \radiuss, \SSSPalgo has the same span bound and is only off by a little in the work bound, but instead, \SSSPalgo is practical.
Compared to \deltas, \SSSPalgo is preprocessing-free, almost faster in all graph instances, and has better theoretical guarantees on undirected graphs.}

\section{Preliminaries}\label{sec:prelim}

\myparagraph{Computational Model}.
\ifconference
We use the
work-\depth{} model for fork-join parallelism with binary forking to analyze parallel
algorithms~\cite{CLRS,blelloch2020optimal}, which is recently used in many papers on parallel algorithms~\cite{agrawal2014batching,blelloch2010low,BCGRCK08,BG04,Blelloch1998,blelloch1999pipelining,BlellochFiGi11,BST12,Cole17,CRSB13,BBFGGMS16,dinh2016extending,chowdhury2017provably,blelloch2018geometry,dhulipala2020semi,BBFGGMS18,Dhulipala2018,blelloch2020randomized,gu2021parallel}.
\else
This paper uses the work-\depth{} model for fork-join parallelism with binary forking to analyze parallel
algorithms~\cite{CLRS,blelloch2020optimal}, which is used in many recent papers on parallel algorithms~\cite{agrawal2014batching,Acar02,blelloch2010low,BCGRCK08,BG04,Blelloch1998,blelloch1999pipelining,BlellochFiGi11,BST12,Cole17,CRSB13,BBFGGMS16,dinh2016extending,chowdhury2017provably,blelloch2018geometry,dhulipala2020semi,BBFGGMS18,Dhulipala2018,blelloch2020randomized,gu2021parallel}.
\begin{figure*}
\centering
  \includegraphics[width=1.5\columnwidth]{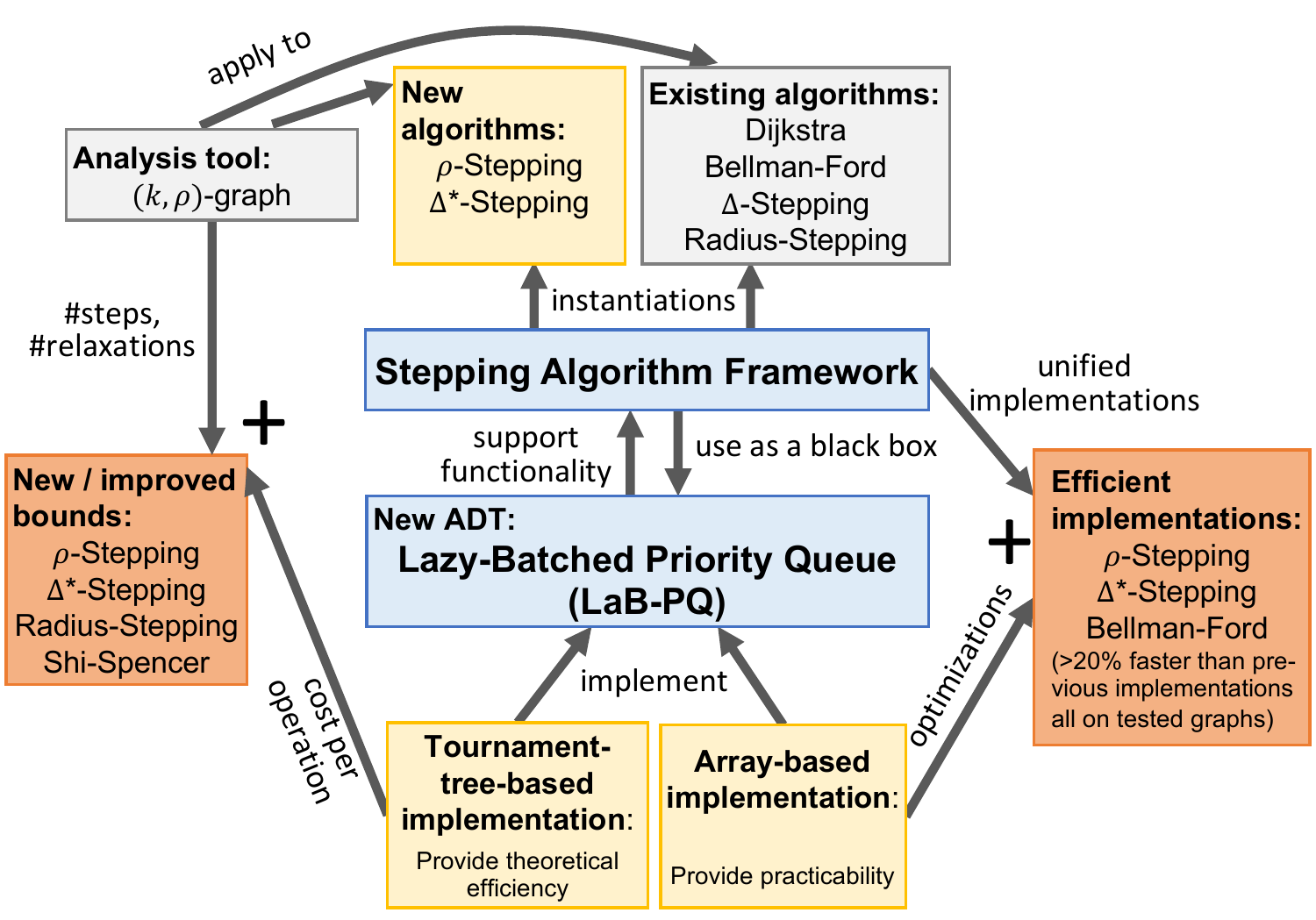}
  \caption{\textbf{An overview of all components in this paper and how they are put together.}  \mdseries The two blue boxes are the abstractions, one for the algorithms and one as an ADT.  The yellow boxes are new algorithms and data structures in this paper.  Grey boxes are existing results we use in this paper.  The two orange boxes are the outcomes of all the techniques, including new work and span bounds for parallel SSSP, and faster implementations as compared to state-of-the-art software. \label{fig:overview}}
  \includegraphics[width=0.8\columnwidth]{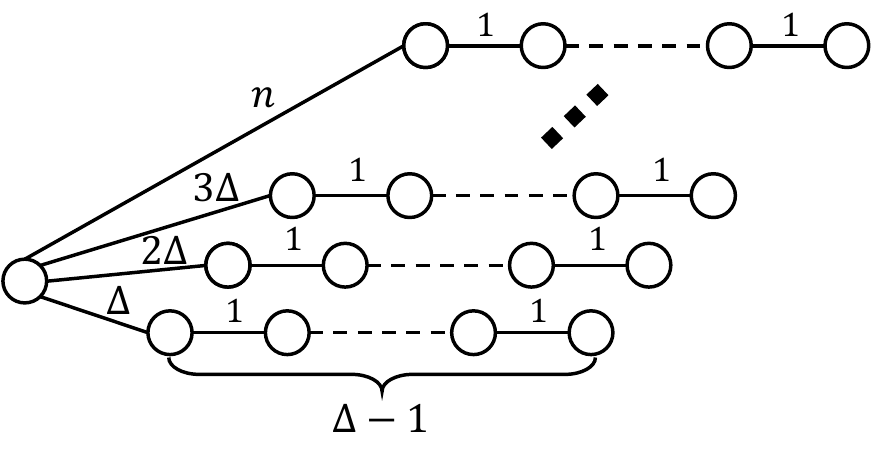}
  \caption{\textbf{An example that incurs $O(n)$ span for \deltas{} on a graph with $\Theta(\Delta)$ shortest path tree depth.} \mdseries On this graph, \deltas{} algorithm needs to run $O(n/\Delta)$ rounds (steps) since the longest distance is $O(n)$. Within each step, we need to run $O(\Delta)$ Bellman-Ford substeps, which will settle the chains respectively.  The \deltass variant proposed in this paper only requires $O(n/\Delta+\Delta)$ rounds (steps) for this graph instance. \label{fig:delta-stepping-span-example}}
  \vspace{-1em}
\end{figure*} 
\fi
We assume a set of \thread{}s that share a common memory.  Each \thread{} supports standard RAM instructions, and a \forkins{}
instruction that forks two new child \thread{}s.
When a \thread{} performs a \forkins{}, the two child \thread{}s all start by running the next instruction, and the original \thread{} is suspended until all children terminate.
A computation starts with a single {root} \thread{} and finishes when that root \thread{} finishes.
An algorithm's \defn{work} is the total number of instructions and
the \defn{span} (depth) is the length of the longest sequence of dependent instructions in the computation.
We can execute the computation using a randomized work-stealing scheduler in practice.
We assume unit-cost atomic operation \WriteMin{}$(p,\mathit{v})$\footnote{a more practical assumption is to charge $O(t)$ work and $O(\log t)$ span when $t$ operations priority update to a memory location.  It does not change the overall bound since forking $t$ parallel tasks requires $\Omega(\log t)$ span, which is already captured.}, which reads the memory location pointed to by $p$, and write value $v$ to it if $v$ is smaller than the current value.
We also use atomic operation \TAS{}$(p)$, which reads and attempts to set the boolean value pointed to by $p$ to \true{}.
It returns \true{} if successful and \false{} otherwise.

\myparagraph{Graph Notations}.
We consider a weighted graph $G=(V, E, w)$. WLOG, we assume $G$ is a connected, simple graph, with minimum edge weight $\min_{e\in E} w(e)=1$, and 
no parallel edges)
We use $L=\max_{e\in E} w(e)$.
For $v \in V$, define $N(v)=\{u\,|\,(v,u)\in E\}$ as the \defn{neighbor set} of $v$.
We use $d(u,v)$ as the shortest-path distance in $G$ between two vertices $u$ and
$v$.  
A \defn{shortest-path tree} rooted at vertex $u$ is a spanning tree $T$ of $G$ such that the path distance in $T$ from $u$ to any other $v\in V$ is $d(u,v)$. 

\myparagraph{\krhograph}. We use the concept of \krhograph{} in \cite{blelloch2016parallel} to analyze stepping algorithms.
\krhograph is a graph invariant highly related to the analysis of parallel SSSP algorithms.
Intuitively, a graph is a \krhograph{} if any vertex can reach its $\rho$ nearest neighbors in $k$ hops.
More formally, we define the \emp{hop distance} $\hat{d}(u,v)$ from a vertex $v$ to $u$ as the number of edges on the shortest (weighted) path from $v$ to $u$ using the fewest edges.
Let $r_{\rho}(v)$ be the $\rho$-th closest distance from $v$, and $\bar{r}_k(v)$ the shortest distance from $v$ to another vertex more than $k$-hops away.

\begin{compactdef}[~\krhograph~\cite{blelloch2016parallel}~]
We say a graph $G=(V,E,w)$ is a \krhograph{} if for all $v\in V$, $r_{\rho}(v)\le \bar{r}_k(v)$.
\end{compactdef}

For a given graph $G=(V,E)$, we denote $k^{G}_\rho$ to be the smallest value for $k$ to make $G$ a \krhograph. With clear context, we omit the superscription. $k_n$ is the shortest-path tree depth.

\myparagraph{Others.}
We use $\log n$ as a short form of $1+\log_2(n+1)$. We say $O(f(n))$ \defn{with high probability} (\whp{}) to indicate $O(cf(n))$ with probability at least $1-n^{-c}$ for $c \geq 1$, where $n$ is the input size.

\section{Frameworks}\label{sec:framework}

\subsection{The \BDPQ{} Abstraction}
An abstract data type \emph{\BDPQfull{}}, or \emp{\BDPQ{}}, denoted as $\PQ$, maintains a universe of \records{} $(\id{},k)$, where $\id{}\in I$ is the unique \emp{identifier} for this \record{} and $k\in K$ is the \emp{key}.
In some applications, each \record{} also has a \emp{value} $v\in V$.
In this paper, if not specified, we assume an empty value type for simplicity.
\ifconference
\else
In many applications, the identifier type $I$ is the natural number set $\N$.
\fi
In all SSSP algorithms in this paper, the identifiers are vertex labels from 1 to $n$.
The total ordering of all keys is determined by a comparison function $<_K:K\times K\mapsto \bool{}$.
A \BDPQ{} $Q\in \PQ$ is associated with a \emp{\mapping{}} $\delta_{Q}:I\mapsto K$, which maps an identifier to its corresponding key (or key-value) that can change dynamically over time.
With clear context, we omit the subscription $Q$, and use $\delta[\id]$ to denote the mapping from $\id$ to key.
In the SSSP algorithms of this paper, this \mapping{} maps each vertex label to its (tentative) distance.
In our implementation, this \mapping{} is passed to \BDPQ{} by a pointer to the tentative distance array. More formally, a \BDPQ{} $\PQ$ is parameterized on the following:
\smallsmallskip

\begin{center}
\noindent {\small\begin{tabular}{@{}ll@{}}
  \hline
  $\boldsymbol{I}$ & Unique identifier type \\
  $\boldsymbol{K}$ & Key type \\
  $\boldsymbol{V}$ & (Optional) Value type \\
  $\boldsymbol{<_K}:K\times K\mapsto \bool$ &  Comparison function on $K$\\
  $\boldsymbol{\delta_Q}:I\mapsto K\times V$ & A mapping from an id to its key (or key-value)\\
  \midrule
\end{tabular}}
\end{center}

A \BDPQ{} maintains a subset of identifiers in the universe.
It can extract \records{} with (relatively) small keys in parallel based on $\delta[\cdot]$.
The interface of the \BDPQ{} includes two functions: \lazyinsert{} and \extract{} (see Table \ref{tab:interface}).
We note that these two functions are sufficient for SSSP application. We discuss more functionalities of \BDPQ{} in
\ifx\fullversion\undefined
the full version of this paper~\cite{ssspfull}.
\else
\cref{sec:fully-dynamic}.
\fi

\boldlazyinsert{}$(\id)$ function commits an update to $Q$ regarding the record with identifier $\id$.
  It ``notifies'' $Q$ that the new key for this \record{} is now in $\delta[\id]$.
  If $\id$ is not in $Q$ yet, \lazyinsert{} inserts it to $Q$.
  Multiple \lazyinsert{} can be executed concurrently.
We note that the change of the \record{} is embodied in the change of $\delta[\id]$, and thus the data structure only needs to know the \record{}'s $\id$ to address the modification.
An important observation is that, we do not have to modify $Q$ immediately, but can execute them \emph{lazily}.
These changes make no difference to any other operations on $Q$ before the next \extract{}.
Compared to the classic ``batch-dynamic'' setting, our interface avoids explicitly generating the batch, which simplifies the algorithm and improves performance.

\boldextract{}$(\theta)$ returns all identifiers in $Q$ with key $\le \theta$, and then deletes them from $Q$.
Note that the result of \extract{} reflects all previous \emph{modifications} to $Q$, including \lazyinsert{} functions and deletions from the previous \extract{}.
It then extracts the corresponding \record{s} based on the latest 
view of $Q$.
An \extract{} function \emph{cannot} be executed concurrently with other functions (\lazyinsert{} or another \extract{}).
This is required for \BDPQ{}'s correctness.
\hide{
Generating the sequence as output does not increase the asymptotical cost of the algorithm, but can be expensive in practice because it increases memory footprint. In our implementations, we use some optimizations to avoid physically generating the sequence (more details in Section \ref{sec:implementation}).}

\myparagraph{Augmenting \BDPQ{}.} In some applications, we need a ``sum'' (the \emph{augmented value} of type $A$) of all \record{s} (keys and possible values) in the \BDPQ{} .
We refer to this as $Q.\collect()$.
This function first map each record in $Q$ to a value of type $A$, and use a binary commutative 
and associative operator~$\oplus$ ($(A,\oplus)$ is a commutative monoid) to compute abstract sum of all \record{s} in $Q$ using $\oplus$.

\begin{table}
  \small
  \begin{tabular}{@{}p{2.5cm}@{  }@{  }p{5.5cm}@{}}
    \hline
    $Q.\lazyinsert(\id)$: \flushright$\PQ \times I \mapsto \Box$ &
    Modify the \record{} with identifier $\id$ in $Q$ to $\delta[\id]$. If $\id\notin Q$, first add it to $Q$.\\
    $s=Q.$\extract{}$(\theta)$:\flushright $\PQ\times K\mapsto \seq$& Return identifiers in $Q$ with keys no more than $\theta$ and delete them from $Q$.\\
    \hline
  \end{tabular}\vspace{-.2in}
  \caption{\small \textbf{Interface of \BDPQ{}.} \label{tab:interface}}\vspace{-.1in}
\end{table}

\subsection{The Stepping-Algorithm Framework}\label{sec:step-algo}

\begin{table*}
  \centering\small
  \vspace{-1em}
  \begin{tabular}{lp{3.5cm}<{}p{3.6cm}<{}cc}
    \toprule
    \bf Algorithm & \boldextcond{} & \boldfinishcond & \bf Work & \bf Span\\
    \midrule
    Dijkstra~\cite{dijkstra1959} & $\displaystyle\theta\gets\min_{v\in Q}(\delta[v])$ & - & $\tilde{O}(m)$ & $\tilde{O}(n)$ \\
    Bellman-Ford~\cite{bellman1958routing,ford1956network} & $\theta\gets+\infty$ & - & $\tilde{O}(k_nm)$ & $\tilde{O}(k_n)$ \\
    \deltas~\cite{meyer2003delta} & $\theta\gets i\Delta$ & if no new $\delta[v]<i\Delta$, $i\gets i+1$  & - & - \\
    \deltass (new)& $\theta\gets i\Delta$ & - & $\tilde{O}\left(k_nm\right)$ & $\tilde{O}\left(\frac{k_n(\Delta+L)}{\Delta}\right)$ \\
    \radiuss~\cite{blelloch2016parallel} & $\displaystyle\theta\gets\min_{v\in Q}(\delta[v]+r_\rho(v))$ & if there exists $\delta[v]<\theta$, do not recompute \extcond & $\tilde{O}(k_\rho m)$ & $\tilde{O}\left(\frac{k_\rho n}{\rho}\cdot \log{L}\right)$ \\
    \SSSPalgo (new) & $\theta\gets$ $\rho$-th smallest $\delta[v]$ in $Q$ & - &  $\tilde{O}(k_n m)$ & $\tilde{O}\left(\frac{k_\rho n}{\rho}\right)$ (undirected) \\
    \bottomrule
  \end{tabular}\vspace{-.15in}
  \caption{\small \textbf{SSSP Algorithms in the stepping algorithm framework}, their \extcond{} and \finishcond{}, and the work and span bounds based on the \BDPQ implementation in \cref{sec:bdpq}. Here $L$ is the longest edge in the graph (assuming the shortest has length 1). $\rho$, $k_\rho$ and $k_n$ are related to $(k,\rho)$-graph defined in \cref{sec:prelim}.   $\tilde{O}()$ omits $\log n$ and lower-order terms for simplicity, and the full bounds are shown in \cref{tab:bounds}. \label{tab:framework}}\vspace{-.15in}
\end{table*}
{
\begin{algorithm}
\caption{ The Stepping Algorithm Framework.\label{alg:framework}}
\fontsize{9pt}{10pt}\selectfont
\SetKwFor{parForEach}{ParallelForEach}{do}{endfor}
\KwIn{A graph $G=(V,E,w)$ and a source node $s$.}
\KwOut{The graph distances $d(\cdot)$ from $s$.}
\DontPrintSemicolon
$\delta[\cdot]\leftarrow +\infty$, associate $\delta$ to a \BDPQ{} $Q$\\
$\delta[s]\leftarrow 0$, $Q.\lazyinsert{}(s)$\\
\While{$|Q|>0$\label{line:while-loop}} {
\parForEach {$u\in Q.\extract(\extcond)$\label{line:extract}} {
    \parForEach{$v\in N(u)$\label{line:innerloop2}} {
    \If {$\WriteMin(\delta[v],\delta[u]+w(u,v))$ \label{line:writemin}} {
      $Q.\lazyinsert{}(v)$
    }
  }
}
Execute \finishcond
}
\Return {$\delta[\cdot]$}\\[.15em]
\end{algorithm}
}

On top of the \BDPQ interface, we also propose a simple \emp{stepping algorithm framework},
in order to reveal the internal connection of the existing SSSP algorithms.
Recall the two sequential textbook algorithms, Dijkstra's algorithm~\cite{dijkstra1959} and Bellman-Ford algorithm~\cite{bellman1958routing,ford1956network}.
Dijkstra only visits one vertex at a time and thus is work-efficient, but it is inherently sequential.
Bellman-Ford visits all vertices in a step so it requires redundant work, but can be easily parallelized.
Many parallel SSSP algorithms integrate the idea in both algorithms, and visit a subset of unsettled vertices close to the source node.
Hence, they require less work than Bellman-Ford, and have better parallelism than Dijkstra.
These algorithms are referred to as stepping algorithms (e.g., \deltas and \radiuss) since they process a batch of vertices in a step. This is captured by \BDPQ{} in the stepping algorithm framework.

We present this stepping algorithm framework in \cref{alg:framework}.
This framework requires two user-defined functions, \extcond and \finishcond.
Many SSSP algorithms can be instantiated by plugging in different \extcond and \finishcond functions (see  \cref{tab:framework}).
\cref{alg:framework} starts with associate the distance array $\delta$ to a \BDPQ{} $Q$.
It then runs in \emp{steps}. 
In each step, we process vertices with distances within a threshold~$\theta$, which is computed by \extcond and used as the parameter of \extract{}.
The extracted vertices will relax their neighbors using \WriteMin (\cref{line:writemin}). If successful, we call \lazyinsert{} on the corresponding neighbor. Some algorithms (e.g., \deltas) contain substeps in each step.
This is captured by \finishcond{}---if the condition is not true, the threshold $\theta$ will not be recomputed.
We say a vertex $v$ is \defn{settled} the last time it is extracted from the \BDPQ{} and relaxes all its neighbors (and thus its distance does not change thereafter).
We define the \defn{frontier} as all vertices in $Q$, which are those waiting to be explored to relax their neighbors.


The stepping algorithm framework applies to various algorithms as shown in \cref{tab:framework}.
We now briefly introduce them.

\myparagraph{Dijkstra and Bellman-Ford}.
Dijkstra's algorithm visits and settles the vertex with the closest distance in the frontier.
By setting $\theta$ as $\min_{v\in Q}(\delta[v])$, \cref{alg:framework} works the same as Dijkstra's algorithm, with the exception that multiple vertices with the same distances will be processed together, which does not affect correctness and efficiency.
Finding the closest vertex can be supported using $\collect()$ and taking $\min$ on keys.
Bellman-Ford visits all vertices in the frontier in each step, so we set $\theta$ as infinity, and in each step \cref{alg:framework} relaxes the neighbors of all vertices in~$Q$.

\myparagraph{\deltas}.
As a hybrid of Dijkstra and Bellman-Ford,
\deltas{} visits and settles all the vertices with shortest-path distances between $i\Delta$ and $(i+1)\Delta$ in step $i$.
Within each step, the algorithm runs Bellman-Ford as substeps.
Hence we can set $\theta$ to $i\Delta$, and use \finishcond to check if any newly relaxed vertex still has distance within $i\Delta$.
If not, we increment $i$ and proceed to the next step.

\myparagraph{\deltass}. We note that \finishcond is not necessary for \deltas, just like other stepping algorithms.
In fact, all existing implementations~\cite{zhang2020optimizing,nguyen2013lightweight,beamer2015gap,dhulipala2017} relaxed \finishcond in different ways. In this paper, we show that removing \finishcond in \deltas (referred to as \emp{\deltass}) can lead to better bounds (\cref{thm:deltass}) and good practical performance (\cref{sec:exp}).

\myparagraph{\radiuss}.
In \radiuss, we precompute $r_\rho(v)$, the distance from each vertex $v$ to the $\rho$-th closest vertex, for all vertices.
Then in each step, \radiuss sets the threshold $\theta$ as $\min_{v\in Q}(\delta[v]+r_\rho(v))$, and then uses Bellman-Ford as substeps to compute the distances for vertices no more than the threshold.
\finishcond is needed by the theoretical analysis, which bounds the number of total substeps to be $O({(k_\rho n/\rho)}\cdot \log{\rho L})$.

To implement \radiuss in our framework, we need an augmented \BDPQ.
We set $r_\rho(u)$ of a vertex $u$ as the value of each \record{}.
We map each record to $k+v$ for a \record{} with key $k$ (distance) and value $v$ (vertex radius), and set the operator $\oplus$ as $\min$.
The threshold in \extract is $\theta=\min_{v\in Q}(\delta[v]+r_\rho(v))$, computed by $Q.\collect()$.
In \cref{sec:bdpq}, we show that maintaining the augmented values does not affect the asymptotical cost bounds.

\myparagraph{\SSSPalgo}.
In this paper, we propose a new algorithm \SSSPalgo{} in the stepping algorithm framework.
\ssspours{} extracts the $\rho$ nearest vertices in the frontier, and relaxes their neighbors.
The threshold $\theta$ is the $\rho$-th smallest element in $Q$.
We overload the notation of $\rho$ from \radiuss{} because they share high-level similarities in the theoretical analysis.
The only step for \SSSPalgo in addition to the stepping algorithm framework is finding the $\rho$-th closest distance among all vertices in the frontier (the \extcond).
In our implementation, we simply use a sampling scheme that randomly pick $s=O(n/\rho+\log n)$ elements, sort them and pick the $(\rho s/n)$-th one.
More details on how to find the $\rho$-th element is in
\ifx\fullversion\undefined
the full paper,
\else
\cref{app:cpam},
\fi
and an efficient implementation is in \cref{sec:impl-theta-sampling}.

Picking the a subset of vertices with closest distances and relaxing their neighbors is not a groundbreaking idea, and has been used in the literature (e.g.,~\cite{bingmann2015bulk,zhou2019practical,alistarh2015spraylist}).
However, the extracting process in previous work is either sequential or concurrent, so none of the existing algorithms support non-trivial work and span bounds, or practical efficiency as compared to \deltas.
In this paper, we argue that this simple solution can achieve both theoretical and practical efficiency.
Theoretically, we show that:
\begin{theorem}[Cost for \SSSPalgo]
\label{thm:rho-steps}
  On a $(k_{\rho},\rho)$-graph $G$, the \SSSPalgo{} algorithm has in $O\left({k_{n}m}\log\frac{n^2}{m\rho}\right)$ work and $O\left(\frac{k_n n\log n}{\rho}\right)$ span.
  If $G$ is undirected, the span is $O\left(\frac{k_\rho n\log n}{\rho}\right)$.
\end{theorem}

We will first show implementations of \BDPQ and the cost, and then formally prove this result in \cref{sec:discussion}.
\ssspours{} also has good practical performance, which is shown in \cref{sec:exp}.

\section{\BDPQ Implementation} \label{sec:bdpq}

We now discuss how to efficiently support \BDPQ{} in \cref{alg:framework}.
We present two data structures for \BDPQ with the goal of theoretical and practical efficiency, respectively.
The obliviousness for data structures from the algorithm's perspective is an advantage of the \BDPQ ADT.

In our analysis, 
we define a \emp{batch of modifications} as all \lazyinsert{} operations between two invocations of \extract{} functions.
The \emp{modification work} on a batch $B$ is all work paid to \lazyinsert{} all \record{s} in $B$, as well as any later work (done by a later \extract{}) to actually apply the updates.
We define a \emp{batch of extraction} as all \record{s} returned by an \extract{} function.
The \emp{extraction work} on a batch $B$ is all work paid to output the batch from the \extract{} function, as well as any later work (done by the next \extract{}) to actually remove them from $Q$.

\subsection{Related Work}

Early PRAM and BSP algorithms had explored parallel priority queues in a variety of approaches~\cite{brodal1998parallel,chen1994fast,das1996optimal,pinotti1991parallel,ranade1994parallelism,crupi1996parallel,baumker1996realistic}, and heavily rely on synchronization-based techniques such as pipelining.
These algorithms do not have better bounds than recent batch-dynamic search trees~\cite{Blelloch2016justjoin,sun2018pam,blelloch2020optimal,sun2019implementing,sun2019parallel} when mapping to the fork-join model.
Other previous papers considered the concurrent, external-memory, and other settings~\cite{alistarh2015spraylist,sundell2005fast,linden2013skiplist,shavit2000skiplist,liu2012lock,henzinger2013quantitative,zhou2019practical,calciu2014adaptive,bingmann2015bulk,sanders2019sequential,sanders1998randomized,sanders2000fast}.
These data structures also do not have better bounds than batch-dynamic search trees since they do not focus on optimizing work or span.
However, existing batch-dynamic search trees or other data structures (e.g., skiplists) maintaining the total ordering of the \record{s}, incur an $\Omega(\log(n))$ work lower bound per \record{} update (more details are in the full paper).
Our key observation is that maintaining total ordering, which incurs overhead both theoretically and practically, is not necessary for a parallel priority queue.

To the best of our knowledge, the only parallel data structure that has similar bounds to our new data structure is the batch-dynamic binary heap~\cite{wang2020parallel}.
However, it has a few disadvantages: it does not support efficient batch-extract, is very complicated (no implementation available), and the span is suboptimal ($O(\log^2 n)$ in the binary fork-join model).
Our new tournament-tree based \BDPQ{} supports full features in the \BDPQ, has $O(\log n)$ span, and is arguably much simpler.

\subsection{Tournament-Tree-Based Implementation}\label{sec:tour-tree}

We start with introducing the tournament tree (aka.\ winner tree).
It is a complete binary tree with $n$ external nodes (leaves) and $n-1$ interior nodes.
A \tourt stores the \record{s} in the leaves.
In our use case, we only need to store the \record{} $\id$ in the leaves using the \BDPQ interface.
Each interior node stores $k\in K$ ($K$ is in key type for the records) that takes the smaller key (defined by $<_K$) from its children.
\cref{fig:tour-tree}(a) illustrates a \tourt when keys are integers and $<_K$ is $<_{\Z}$.

\begin{figure}
  \centering
\vspace{-.5em}
\includegraphics[width=0.8\columnwidth]{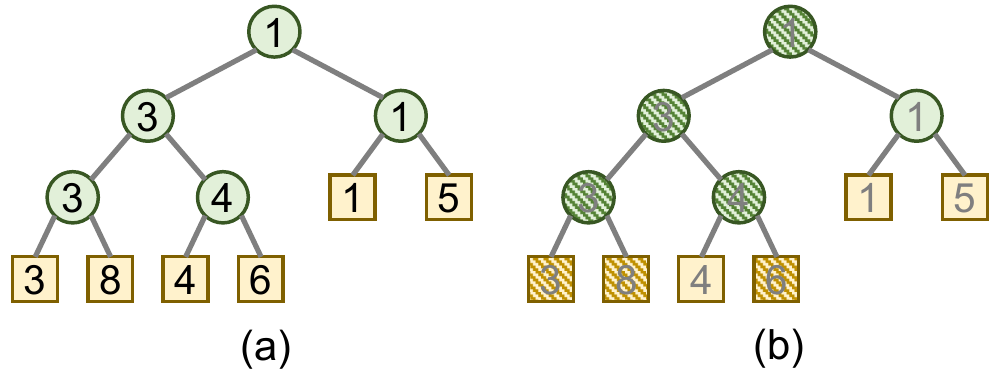}
  \caption{\small \textbf{A tournament tree.}  Square leaf nodes store the records and round interior nodes keep the smallest key in their subtrees. (a) is a \tourt{} containing 6 records $3,8,4,6,1$ and $5$. (b) shows an update on a batch of $3,8$ and $6$. The shaded nodes are marked as \emph{renewed}.\label{fig:tour-tree}}
  \vspace{-2em}
\end{figure}

We now discuss how to use a \tourt to implement \BDPQ{}.
We will use $t.\wtleft$, $t.\wtright$ and $t.\wtparent$ to denote the left child, right child and parent of a node~$t$.
For simplicity, we assume the universe of the records has a fixed size $n$ (for SSSP $n=|V|$), and the \tourt has $n$ leaf nodes each with a boolean flag $\mb{inQ}$ indicating if this record is in (has been inserted to) the \BDPQ $Q$.
We note that this is sufficient for the SSSP algorithms.
The dynamic version (where the size of the \tourt{} changes with the size of \BDPQ{}) will be described in the full paper.


\begin{algorithm}
\caption{The Tournament-tree based \BDPQ.\label{alg:tour-tree}}
\fontsize{9pt}{10pt}\selectfont
\SetKw{KwMaintain}{Maintains}
\KwMaintain{A \tourt $T$ with $n$ leaf nodes each corresponding to a record.}
    \vspace{0.5em}
\DontPrintSemicolon
\SetKwBlock{InPar}{In Parallel:}{}
\SetKwProg{myfunc}{Function}{}{}

\myfunc{\upshape $\addmark($record id $\id$, boolean flag $\mb{newflag})$} {
  Let $t$ be the tree leaf corresponding to \mb{id}\\
  $t.\inQ\gets \mb{newflag}$\\
  \While {\upshape $t\ne T.root$ and $\TAS(t.\wtparent.\updated)$} {
    $t\gets t.\wtparent$
  }
}
\codeskip{}
\myfunc{\upshape $\wtupdate($node $t) \to k\in K$\label{line:update-start}} {
  \If {$t$ is leaf} {
    \lIf{$t.\mb{inQ}$} {\Return{$\delta[t.id]$}}
    \lElse {\Return{$+\infty$}
  }
  }
  \lIf {$t.\updated=0$} {\Return{$t.k$}}
  $t.\updated \gets 0$\\
  \InPar {
    $\mb{leftKey}\gets \wtupdate(t.\wtleft)$\\
    $\mb{rightKey}\gets \wtupdate(t.\wtright)$
  }
  \Return {$t.k \gets\min(\mb{leftKey}, \mb{rightKey}$)}\label{line:update-end}
}
\codeskip{}
\myfunc{\upshape $\traverse(\text{threshold }\theta\text{, node }t) \to \mb{seq}$\label{line:traverse-start}} {
  \If {$t$ is a leaf} {
    \If {$(\delta[t.\id] \le \theta)$}{
    $\addmark(t.\id, 0)$\label{line:delete} \tcp*{Marked as not in Q}
    \Return{$\{t.\mb{id}\}$}}
    \lElse{\Return{$\{\}$}}
  }
  \lIf {$\theta<t.k$} {\Return{$\{\}$}\tcp*[f]{empty seq}}
  \InPar {
    $\mb{leftseq}\gets\traverse(\theta,t.\wtleft)$\\
    $\mb{rightseq}\gets\traverse(\theta,t.\wtright)$
  }
  \Return $\mb{leftseq}+\mb{rightseq}$\label{line:traverse-end}
}
\codeskip{}
\myfunc{\upshape $\extract(\text{threshold }\theta$)} {
  $\wtupdate(T.\mb{root})$\\
  \Return $\traverse(\theta,T.\mb{root})$
}
\codeskip{}
\myfunc{\upshape $\lazyinsert(\id$)} {
  $\addmark(\id{},1)$\\
}
\end{algorithm}

A \tourt $T$ on $n$ records contains $2n-1$ nodes in total. The first $n-1$ nodes are interior nodes.
For an interior node $t$, we use $t.k$ to denote the key stored in node $t$.
To support the \BDPQ{} interface, each interior node contains a bit flag $\updated$ indicating if any key in its subtree has been modified after the last update of $t.k$.
This flag is initially set to 0 (\false).

Constructing such a tree with given initial values simply takes linear work and $O(\log n)$ span using divide-and-conquer: construct both subtrees recursively in parallel, and update the root's key based on the two children's keys.

We next present the implementation of \BDPQ's interface using \tourt{}. Due to page limit, we only show the pseudocode (\cref{alg:tour-tree}) and a high-level overview here.
The analysis are given in the full paper. 
We first introduce a helper function \addmark{}$(\id{}, \mathtext{newflag})$.

\addmarkbold{}$(\id{}, \mathtext{newflag})$ first sets the record $\id$'s $\mathtext{inQ}$ flag to be $\mathtext{newflag}$---0 means the record should be deleted, and 1 means inserted.
Whichever value $\mathtext{newflag}$ is, this means the \record{} of $\id$ has been updated.
Then, the algorithm marks the $\updated$ flags of the nodes on the tree path from the updated leaf to the root.
This process is executed using \TAS{}.
If the \TAS{} fails, we know that another \addmark{} has marked the rest of the path, so the current \addmark{} terminates immediately. 
An example is shown in \cref{fig:tour-tree}(b).
Updating the nodes' keys is postponed to the next \extract{} function.

\para{\upshape \lazyinsertbold{}.}
The \lazyinsert{} algorithm simply calls \addmark{}$(\id{}, 1)$.

\para{\upshape \boldextract{}.} \extract{} first uses a function \wtupdate{} to update the keys for all nodes with $\updated$ flag as 1.
It then calls \traverse{} to output all \record{s} with keys no more than $\theta$.
Those output keys are also marked as deleted from $T$ (\cref{line:delete}).

The $\wtupdate{}(t)$ function recursively restores the keys in the interior nodes using a divide-and-conquer approach (\cref{line:update-start}--\ref{line:update-end}), and returns the key at the current node $t$.
The return value of a leaf node is either the record's key or infinity, depending on the $\inQ$ flag.
For an interior node, \wtupdate{} will update the key to be the smaller one of its two children and return this key.
After all interior tree nodes have been updated, we use \traverse{} to acquire all records with keys no more than $\theta$.
This step can be parallelized similarly using divide-and-conquer (\cref{line:traverse-start}--\ref{line:traverse-end}): we can traverse the left and right subtrees respectively and concatenate the two results. A subtree is skipped when the key at the subtree root (minimum key in the subtree) is larger than $\theta$.
Note that if we want the output sequence in a consecutive array, we can traverse for two rounds---the first round computes the number of extracted records, and the second round writes them to the corresponding slots. 

We can implement \collect similarly.
We keep a collective status $t.a\in A$ ($A$ is the augmented value type) for each interior node $t$, and it is updated in the \mf{Update} function in \cref{line:update-start}--\ref{line:update-end} similar to the update for $k$.
We do not need to update $t.a$ for node $t$ if the subtree rooted at $t$ remains unchanged, which is captured by $t.\updated$. 

\begin{theorem}\label{thm:main-tour-tree}
  Consider a \tourt{} on a universe of $n$ records, implemented with algorithms in \cref{alg:tour-tree}.
  The modification work on a size-$b$ batch is $O\left(b\log(n/b)\right)$.
  The extraction work on a size-$b$ batch is $O\left(b\log(n/b)\right)$.
  The span of \extract{} and \lazyinsert{} is $O(\log n)$.
\end{theorem}

\ifx\fullversion\undefined
The formal analysis of \cref{thm:main-tour-tree} is given in
the full paper.
We will see how to use \cref{thm:main-tour-tree} in \cref{sec:sssp}.
\else
To prove this theorem, we will use the following lemma.
\begin{lemma}
\label{lem:wtfunctions}
Given $b$ invocations of \addmark{} function in a batch, the total cost of these \addmark{} functions is $O\left(b\log\frac{n}{b}\right)$.
If there are $b$ invocations of \addmark{} function in the last batch, the total cost of the \wtupdate{} is $O\left(b\log\frac{n}{b}\right)$.
If the \traverse{} algorithm extracts $b$ (smallest) elements from \tourt{} $T$, the total cost of \traverse{} is $O\left(b\log\frac{n}{b}\right)$.
\end{lemma}
\begin{proof}
  We first show that, in a \tourt{}, for a subset $X$ of tree leaves, if we denote $S(X)$ as the set of all ancestors of nodes in $X$, then $|S(X)|=O\left(|X|\log\frac{n}{|X|}\right)$.
  This has been shown on more general self-balanced binary trees \cite{Blelloch2016justjoin,blelloch2018geometry}, and just a simplified case suffices as \tourt{s} are complete binary trees.

  First, \addmark{} modifies the flag $\updated$ for each tree path node all modified leaves to the root. Let $X$ be all tree leaves corresponding to the invocations \addmark{} functions in this batch. Therefore $|X|=b$.
  By definition of $S(X)$, we know that only nodes in $X\cup S(X)$ are visited.
  Because of the \tas{} operation, a node $v\in S(X)$ recursively call \addmark{} on its parent $u$ if and only if $v$ successfully set the $\updated$ mark of $u$. This means that for any interior node $u\in S(X)$, this can happen only once. Therefore, every node in $S(X)$ is visited by the \addmark{} function at most once. This proves that the cost of all \addmark{} functions $O\left(b\log\frac{n}{b}\right)$.

  For \wtupdate{}, it restores all relevant interior tree nodes top-down. Let $X$ be all tree leaves corresponding to the invocations \addmark{} functions in the last batch (so that they need to be addressed in this \wtupdate{}). Therefore $|X|=b$.
  Note that the algorithm skips a subtree when the $\updated$ flag is \false{}. Therefore, the total number of visited nodes must be a subset of all nodes in $X\cup S(X)$ and their children. Given that each node has at most two children, this number can also be asymptotically bounded by $|X\cup S(X)|$, which are those marked \true{} in the $\updated$ flags. This gives the same bound as the total cost of \addmark{} functions in a batch.

  For \extract{}, denote the leaves corresponding to all the output \record{s} as set $X$ of size $b$.
  Note that we skip a subtree if its key (minimum key in its subtree) is larger than $\theta$. This means that we visit an interior node only if at least one of its descendants will be included in the output batch. Therefore, all visited nodes are also $X\cup S(X)$ and all their children.
  For all visited leaves, \extract{} also calls \addmark{} to set $\mb{inQ}$ flag to 0 (\false). As proved above, the total cost of these \addmark{} functions is $O\left(b\log\frac{n}{b}\right)$.
  Therefore, the total cost of \traverse{} is also $O\left(b\log\frac{n}{b}\right)$ to extract a batch of size~$b$.
\end{proof}

With \cref{lem:wtfunctions} that shows the cost of each function in \cref{alg:tour-tree}, we can now formally prove \cref{thm:main-tour-tree}.

\medskip
\begin{proof}[Proof of \cref{thm:main-tour-tree}]
  Recall that the \emp{modification work} on a batch $B$ as all work paid to \lazyinsert{} all \record{s} in $B$, as well as any later work (possibly done by a later \extract{}) to actually apply the updates, and the \emp{extraction work} on a batch $B$ is all work paid to output the batch from the \extract{} function, as well as any later work (possibly done by the next \extract{}) to actually remove them from $Q$. In \tourt{}, the modification work includes all \addmark{} operations called by the \lazyinsert{} functions in the batch, as well as later cost in the \wtupdate{} operation to restore keys of all the relevant interior nodes. From \cref{lem:wtfunctions}, the total cost is $O(b\log(\frac{n}{b}))$ for a batch of size $B$.
  The extraction work on a batch $B$ includes the work done by \traverse{} to get the output sequence, as well as to restore keys of all the relevant interior nodes in the next \wtupdate{}.

  For both \extract{} and \lazyinsert{}, the span is no more than the height of the tree, which is $O(\log n)$.
\end{proof}
\fi



\subsection{Array-Based Implementation}

\cref{alg:tour-tree} uses a tree-based structure to provide tight work bounds for applying a batch of modifications or extractions.
This is asymptotically better than batch-dynamic search trees~\cite{Blelloch2016justjoin,sun2018pam,blelloch2020optimal}.
However, in practice, maintaining a tree-based data structure can be expensive because of larger memory footprint and random access.
Even though we can implement a \tourt{} in a flat array (no pointers), it still requires extra storage for interior nodes and incurs frequent random accesses (following tree path).
When the batch size $b$ approaches $n$ and $O(\log(n/b))$ becomes small, the theoretical advantage of \tourt{s} becomes insignificant, and is asymptotically the same as just loop over all \record{s}.

This is observed by the practitioners. Most (if not all) practical SSSP implementations just keep an array for all records without maintaining sophisticated structures.
This is because parallel SSSP algorithms usually use a very large value of $b$ to get sufficient parallelism.
To implement \lazyinsert{} on an array, we can just set a flag to indicate a \record{} is added to $Q$.
For \extract, we loop over the entire array and pack all records with keys within~$\theta$ in parallel, which takes linear work and $O(\log n)$ span.
While efficiently implementing the array requires many subtle details (shown in \cref{sec:impl-bdpq}), asymptotically, the following bound is easy to see.

\begin{theorem}\label{thm:array-based}
  The array-based \BDPQ requires $O(b)$ modification work on a size-$b$ batch.
  The extraction work on a size-$b$ batch is $O(n)$.
  The span of \extract{} and \lazyinsert{} is $O(\log n)$.
\end{theorem}

\section{Analysis for Stepping Algorithms}\label{sec:sssp}
\begin{table*}[!t]
\vspace{-1em}
\small
  {\centering
    \def\arraystretch{1.25}
  \newcommand\encircle[1]{\raisebox{.6pt}{\textcircled{\raisebox{-.3pt} {\scriptsize #1}}} }
  \newcommand{\undirectedbound}{\scriptsize (U)}
  \newcommand{\directedbound}{\scriptsize (D)}
  \begin{tabular}{@{}l@{}c@{}c@{  }@{  }c@{}cc}
    \toprule
    \multirow{2}[0]{*}{\bf Algorithm} & \multicolumn{2}{c}{\bf Work} & \bf \multirow{2}[0]{*}{Span} & \multicolumn{2}{c}{\bf \hfill Previous Best\hfill}\\
    & \bf Tournament-tree-based & \bf Array-based & & \bf Work & \bf Span \\
    \midrule
    Dijkstra~\cite{dijkstra1959,brodal1998parallel} & $O\left(m\log \frac{n^2}{m}\right)$ & $O(m+n^2)$ & $O(n\log n)$ & $O(m \log n)$ & same \\
    Bellman-Ford~\cite{bellman1958routing,ford1956network} & $O(k_nm)$ & $O(k_nm)$ & $O(k_n\log n)$ & same & same \\
    \deltass & $O\left(k_nm\log\frac{nL}{m\Delta}\right)$ & $O\left(k_nm+\frac{k_nn(\Delta+L)}{\Delta}\right)$ & $O\left(\left(\frac{k_n (\Delta+L)}{\Delta}\right)\log n\right)$ & - & - \\
    \radiuss{$^{\dagger}$}~\cite{blelloch2016parallel} & $O\left(k_\rho m\log\frac{n^2\log{\rho L}}{m\rho}\right)$ \undirectedbound{} & $O\left(k_\rho m+\frac{k_\rho n^2}{\rho}\cdot \log{\rho L}\right){}$ \undirectedbound & $O\left(\frac{k_\rho n}{\rho}\cdot \log{\rho L}\log n\right)$ \undirectedbound{} &$O\left(k_\rho m\log n\right){}$ \undirectedbound{} & same\\
    Shi-Spencer$^{\dagger}$~\cite{Shi1999} & $O\left((m+n\rho)\log\frac{n^2}{m+n\rho}\right)$ \undirectedbound{} & $O\left(m+n\rho+\frac{n^2}{\rho}\right)$ \undirectedbound{} & $O\left(\frac{n\log n}{\rho}\right)$  \undirectedbound{} &$O\left((m+n\rho)\log n\right)$ \undirectedbound{} & same\\
    \multirow{2}[0]{*}{\SSSPalgo} & \multirow{2}[0]{*}{$O\left(k_n m\log\frac{n^2}{m\rho}\right)$} & $O\left(k_nm+\frac{n^2k_\rho}{\rho}\right)$  \undirectedbound & $O\left(\frac{k_\rho n\log n}{\rho}\right)$ \undirectedbound & \multirow{2}[0]{*}{-} & \multirow{2}[0]{*}{-} \\
    && $O\left(k_nm+\frac{n^2k_n}{\rho}\right)$  & $O\left(\frac{k_n n\log n}{\rho}\right)$  \\
    \bottomrule
  \end{tabular}
  \vspace{-1em}
  \caption{\small\textbf{New work and span bounds for the stepping algorithms and comparison to previous results.} (U) indicates the bound only works for undirected graphs. 
  (-) indicates no non-trivial bound is known to the best of our knowledge.
  (same) indicates the previous bound matches the tournament-tree-based work or the span.
  All new work bounds for \deltass, \radiuss, Shi-Spencer, and \ssspours are based on the distribution lemma (\cref{lem:distribute}) and the \BDPQ{} bounds.
  \radiuss and Shi-Spencer (noted with $^{\dagger}$) require preprocessing. 
  \vspace{-1em}
  \label{tab:bounds}}
}
\end{table*}

With the stepping algorithm framework (\cref{alg:framework}) and \BDPQ's implementation, we can now formally analyze the cost bounds for the stepping algorithms, which are summarized in \cref{tab:bounds}.
Our new bounds are parameterized by the definition of \krhograph shown in~\cite{blelloch2016parallel}.
We first show some useful results for all stepping algorithms in \cref{sec:cost-prelim}, and use them to prove the results in \cref{sec:cost-algo}.
We later show the span for \ssspours{} on undirected graphs in \cref{sec:undirected}, and compare with existing algorithms in \cref{sec:discussion}.

\subsection{Useful Results for All Stepping Algorithms}\label{sec:cost-prelim}

We first show two useful lemmas for all stepping algorithms.

\begin{lemma}[Number of extractions]\label{lem:num-ext}
In a stepping algorithm, a vertex $v\in V$ will not be extracted from the priority queue (\cref{line:extract} in \cref{alg:framework}) more than $k_n$ times.
\end{lemma}
\begin{proof}
  Consider the shortest path $P=\{v_0=s, v_1, v_2, \ldots, v_l=v\}$ from the source $s$ to $v$ with fewest hops.
  Since we assume the edge weights are positive, we know that $d(s,v_i)<d(s,v_{j})$ for $i<j$.
  Hence, whenever $v$ is extracted from the priority queue, the earliest unsettled vertex $v_i$ in $P$ must also be extracted and settled.
  This is because $v_{i-1}$ is already settled and have relaxed $v_i$ in previous rounds, and $d(s,v_i)\le d(s,v)$.
  Based on the definition of the \krhograph, we have $l\le k_n$, which proves the lemma.
\end{proof}

\begin{lemma}[Distribution]\label{lem:distribute}
If a stepping algorithm has $S$ steps, and incurs $U$ updates (relaxations), the total work is $\displaystyle O\left(U\log(nS/U)\right)$ using tournament-tree-based \BDPQ.
\end{lemma}

\begin{proof}
  The work of a stepping algorithm consists of modification work for relaxations (updates) and extraction work applied to the \BDPQ{}.
  Each extracted vertex corresponds to a previous successful relaxation, and an update and an extraction have the same cost per vertex.
  Hence, we only need to analyze modification costs since extraction costs are asymptotically bounded.

  The $U$ updates are distributed in $S$ steps. Let $u_i$ be the number of relaxations applied in step $i$ ($\sum_i{u_i}=U$).
  The overall work across all steps is $W=O(\sum_i{u_i\log(n/u_i)})$.
  Since $u_i\log(n/u_i)$ is concave, $\sum_i{u_i\log(n/u_i)}\le S\cdot((\sum_i{u_i}/S)\log (n/(\sum_i{u_i}/S)))=U\log(nS/U)$, which proves the lemma.
\end{proof}

\subsection{Cost Bounds for Stepping Algorithms}\label{sec:cost-algo}

With \cref{lem:distribute,lem:num-ext} for the stepping algorithms and \cref{thm:main-tour-tree,thm:array-based} for \BDPQ's cost, we can now show the cost bounds shown in \cref{tab:bounds} except for one given in \cref{sec:undirected}.

Dijkstra's algorithm has $O(n)$ steps and $O(m)$ relaxations, \cref{lem:distribute} gives $O(m\log(n^2/m))$ work which is essentially better than Brodal et al.'s algorithm~\cite{brodal1998parallel} (their span is also $O(n\log n)$ on the fork-join model).
Bellman-Ford has $O(k_n)$ steps and $O(k_nm)$ relaxations, so the work is $O(k_nm)$ and the span is $O(k_n\log n)$.
The following theorem shows the number of steps for \ssspours.

\begin{theorem}[Number of steps for \ssspours]
\label{thm:rho-steps-general}
  On a $(k_{\rho},\rho)$-graph, the \SSSPalgo{} algorithm finishes in $O\left(k_{n}n/\rho\right)$ steps.
\end{theorem}
\begin{proof}
  In \SSSPalgo{}, each step can either be a \emph{full-extract}, where $|Q|\ge \rho$ so we extract $\rho$ vertices with closest tentative distances, or a \emph{partial-extract}, where $|Q|<\rho$ so we extract all but fewer than $\rho$ vertices.
  There can be at most $O\left(k_{n}n/\rho\right)$ full-extracts, since \cref{lem:num-ext} shows that each vertex can only be extracted for $k_n$ times.
  Given that we have $n$ vertices in total, there can be at most $O\left(k_{n}n/\rho\right)$ full-extracts.
  We now show that at most $k_n$ partial-extracts can occur.
  Similar to the analysis for \cref{lem:num-ext}, once a partial-extract occurs, at least one vertex on the shortest path $P$ from source $s$ to any vertex $v$ is settled.
  Based on the definition of the \krhograph, we have $|P|\le k_n$, so in total, at most $k_n$ partial-extracts can occur.
  Putting both parts together proves the theorem.
\end{proof}
Combining the result with \cref{lem:distribute} gives the work bound of \ssspours{} in \cref{tab:bounds}.

We now show that we can get better work bounds for \radiuss{} using \BDPQ.
\radiuss extracts all vertices with distance within $\min_{v\in Q}(\delta[v]+r_\rho(v))$ in each step. The original papers uses a search tree to support this operation.
We note that our \BDPQ fully captures the need in \radiuss{}.
By replacing the search tree with our \tourt and plugging in the numbers of relaxations and steps, we get the following results.

\begin{corollary}\label{cor:radius}
\radiuss~\cite{blelloch2016parallel} uses $O\left(k_\rho m\log\frac{n^2\log{\rho L}}{m\rho}\right)$ work and $O\left(\frac{k_\rho n}{\rho}\cdot \log{\rho L}\log n\right)$ span, with $O(m\log n+n\rho^2)$ work and $O(\rho\log\rho+\log n)$ span for preprocessing.
\end{corollary}

We can also improve another parallel SSSP algorithm Shi-Spencer~\cite{Shi1999} by replacing their original search-tree-based priority queue with our \tourt{} (more details in
\ifx\fullversion\undefined
the full paper).
\else
\cref{app:shi-spencer}).
\fi

\begin{corollary}\label{cor:shispencer}
Shi-Spencer algorithm~\cite{Shi1999} can be computed using $O\left((m+n\rho)\log\frac{n^2}{m+n\rho}\right)$ work and $O\left(\frac{n\log n}{\rho}\right)$ span, with $O(m+n\rho^2\log n\log\rho)$ work and $O(\log n\log\rho)$ span for preprocessing.
\end{corollary}

We also derive the bounds for \ssspours on directed graphs in \cref{thm:rho-steps}, and give the formal analysis for \deltastarstepping{}:
\begin{theorem}\label{thm:deltass}
  \deltass uses $O\left(\frac{k_n(\Delta+L)}{\Delta}\right)$ steps, and thus has $O\left(k_nm\log\frac{nL}{m\Delta}\right)$ work and $O\left(\frac{k_n (\Delta+L)}{\Delta}\log n\right)$ span based on \BDPQ{}.
\end{theorem}
\ifx\fullversion\undefined
Due to the space limit, the proof of \deltastarstepping is given in the full paper.
We note that for the original \deltas, such bounds do not hold.
An additional factor of $k_n$ will be introduced in span if each step needs to settle down all vertices in the distance threshold.
\else
\begin{proof}[Proof of \cref{thm:deltass}]
  We first show the number of steps in \deltass.
  The farthest vertices in the shortest-path tree has distance no more than $k_n L$ where $L$ is the largest edge weight.
  After $\lceil k_n L/\Delta\rceil$ steps, all vertices are included in the threshold, so the algorithm will finish in no more than another $k_n$ steps (the number of steps of Bellman-Ford).
  Hence, in total, \deltass uses $S=O\left(\frac{k_n(\Delta+L)}{\Delta}\right)$ steps.

  Based on \cref{lem:num-ext}, we know the number of total relaxations are upper bounded $U=k_nm$.
  We can now use \cref{lem:distribute} to get the work bound:
  $$\displaystyle O\left(U\log\frac{nS}{U}\right)=O\left(k_nm\log\frac{nk_n(\Delta+L)}{k_nm\Delta}\right)=O\left(k_nm\log\frac{n(\Delta+L)}{m\Delta}\right)$$
  Note that since $m\ge n$, $\frac{n(\Delta+L)}{m\Delta}$ makes a difference only when $L>\Delta$.
  Hence, the work bound can be simplified as $O\left(k_nm\log\frac{nL}{m\Delta}\right)$.
\end{proof}

We note that for the original \deltas, such bounds do not hold.
An additional factor of $k_n$ will be introduced if we need to settle down all vertices in each step, as shown in \cref{fig:delta-stepping-span-example}.
\fi

\subsection{Number of Steps for Undirected Graphs}\label{sec:undirected}

We can show tighter span bounds for \SSSPalgo on undirected graphs, which is inspired by existing results including \radiuss~\cite{blelloch2016parallel} and Shi-Spencer's algorithm~\cite{Shi1999}.
\ifx\fullversion\undefined
\else
We first give the main theorem for this section.
\fi

\begin{theorem}[Number of steps, Undirected]
\label{thm:rho-steps-un}
  On an undirected $(k_{\rho},\rho)$-graph, the \SSSPalgo{} algorithm finishes in $O\left(k_{\rho}n/\rho\right)$ steps.
\end{theorem}

In real-world graphs, we usually have $k_\rho\ll \rho$ for large $\rho$.
Hence, by picking a large $\rho$, say $n/\log n$, \SSSPalgo only requires a small number of rounds and provides ample parallelism.
As a comparison, \radiuss requires $O\left(\frac{k_{\rho}n}{\rho}\log\rho L\right)$ steps, a factor of $O(\log\rho L)$ more for the worst-case guarantee. 

\ifx\fullversion\undefined
Due to the space limit, we show the proof in the full version of this paper, and only provide our proof sketch here.
We will show that after step $(2k_\rho+3)t$ for $t\ge 1$, \SSSPalgo{} will successfully settle at least the closest $t\rho$ vertices from $s$ and relax their neighbors.
This means we need $n/\rho\cdot O(k_\rho)=O\left(k_{\rho}n/\rho\right)$ steps.
We will show this by induction.
We note that the base case trivially holds when $t=1$, since $s$ can reach $\rho$ closest vertices in $k_{\rho}$ hops.
Assume this is true for $t$, we will show that this is also true for $t+1$.

More specifically, let $\mathcal{N}_{\rho}(u)$ be the set of $\rho$-nearest vertices from vertex $u$.
For simplicity, let $\mathcal{T}_\rho=\mathcal{N}_{\rho}(s)$ where $s$ is the source vertex.
The inductive hypothesis assumes that vertices in $\mathcal{T}_{t\rho}$ are settled.
We then show that within the next $2k_\rho+2$ steps, all vertices in $\mathcal{T}_{(t+1)\rho}\setminus\mathcal{T}_{t\rho}$ are settled (updated to the exact distance).

Let $v\in \mathcal{T}_{(t+1)\rho}\setminus\mathcal{T}_{t\rho}$ and $P=\{s=v_0, v_1, v_2, \dots, v_l=v\}$ be a the shortest path from $s$ to $v$ with the fewest hops.
Assume all vertices from $v_0$ to $v_i$ are in $\mathcal{T}_{t\rho}$, and beyond $v_i$ all vertices are not in $\mathcal{T}_{t\rho}$.
We show that $v$ is within $2k+2$ hops from $v_i$.
%
More details will be given in the full version of this paper.

\else
Our proof sketch is as follows.
We will show that after step $(2k_\rho+3)t$ for $t\ge 1$, \SSSPalgo{} will successfully settle at least the closest $t\rho$ vertices from $s$ and relax their neighbors.
We will show this by induction.
We note that the base case trivially holds when $t=1$, since $s$ can reach $\rho$ closest vertices in $k$ hops.
Assume this is true for $t$, we will show that this is also true for $t+1$.

We start with some notations.
Let $\mathcal{N}_{\rho}(u)$ be the set of $\rho$-nearest vertices from vertex $u$.
For simplicity, let $\mathcal{T}_\rho=\mathcal{N}_{\rho}(s)$ where $s$ is the source vertex.
The inductive hypothesis assumes that vertices in $\mathcal{T}_{t\rho}$ are settled.
We now show that within the next $2k_\rho+2$ steps, all vertices in $\mathcal{T}_{(t+1)\rho}\setminus\mathcal{T}_{t\rho}$ are settled (updated to the exact distance).

Let $v\in \mathcal{T}_{(t+1)\rho}\setminus\mathcal{T}_{t\rho}$ and $P=\{s=v_0, v_1, v_2, \dots, v_l=v\}$ be a the shortest path from $s$ to $v$ with the fewest hops.
Assume all vertices $v_0$ through $v_i$ are in $\mathcal{T}_{t\rho}$, and beyond $v_i$ all vertices are not in $\mathcal{T}_{t\rho}$.
We will first show that $v$ is within $2k+2$ hops from $v_i$.

Throughout the analysis, we consider $v_j$ on path $P$ where $j=i+k+2$ as a special vertex, if it exists. If not, it directly means that $v$ is no more than $k+2$ hops away from $v_i$.
We use the neighbor set of $v_j$ to show $l\le i+2k_\rho+3$ and complete the proof.
To start, we show the following lemma.
  \begin{lemma}
  \label{lem:nonn1}
  $v_{i+1}$ is not in $\mathcal{N}_{\rho}(v_j)$.
  \end{lemma}
  \begin{proof}
  From the definition of \krhograph, all vertices in $\mathcal{N}_{\rho}(v_j)$ are within $k$ hops from $v_j$.
  Since the shortest path with fewest hops from $v_{i+1}$ to $v_j$ has $k+1$ hops, $v_{i+1}$ cannot be in $\mathcal{N}_{\rho}(v_j)$.
  \end{proof}

We now show the following lemma that $\mathcal{T}_{t\rho} \cap \mathcal{N}_{\rho}(v_j)=\varnothing$.
It says that no vertices in $\mathcal{T}_{t\rho}$ are close enough to be processed in the previous steps.
We use \cref{lem:nonn1} and $v_{i+1}$ as a separating vertex.

  \begin{lemma}
  \label{lem:nonn2}
  None of the vertices in $\mathcal{T}_{t\rho}$ is in $\mathcal{N}_{\rho}(v_j)$.
  \end{lemma}
  \begin{proof}
    Assume to the contrary that there exists a vertex $u\in \mathcal{T}_{t\rho}\cap \mathcal{N}_{\rho}(v_j)$.
    Then we know $d(s,u)<d(s,v_{i+1})$, since $u$ is one of the $t\rho$ closest vertices from $s$ ($u\in \mathcal{T}_{t\rho}$) but $v_{i+1}$ is not.
    From \cref{lem:nonn1}, we know $v_{i+1}\notin \mathcal{N}_{\rho}(v_j)$.
    This means that if $u\in \mathcal{N}_{\rho}(v_j)$, then $d(u,v_j)<d(v_{i+1}, v_j)$.
    Combining the above two conclusions, we know that $d(s,u)+d(u,v_j)<d(s,v_{i+1})+d(v_{i+1},v_j)$, which leads to a contradiction that $v_{i+1}$ is on the shortest path from $s$ to $v_j$.
  \end{proof}

  \cref{lem:nonn2} says none of the $\rho$ closest vertices of $v_j$ are in $\mathcal{T}_{t\rho}$.
  We next show that for all vertices in $\mathcal{T}_{(t+1)\rho}\setminus\mathcal{T}_{t\rho}$ are close to $v_j$.

  \begin{lemma}
  For any $u\in \mathcal{T}_{(t+1)\rho}\setminus\mathcal{T}_{t\rho}$, it is either $v_{i+1}$, or it is within $k$ hops from $v_j$.
  \end{lemma}
  \begin{proof}
    Again, assume to the contrary that $u$ is not within $k$ hops from $v_j$ and is not $v_{i+1}$.
    In that case, $u$ is not in $\mathcal{N}_{\rho}(v_j)$.
    In that case, any vertex $u'\in \mathcal{N}_{\rho}(v_j)$ should be closer to $s$ than $u$.
    Since $u\in \mathcal{T}_{(t+1)\rho}$, any $u'\in \mathcal{N}_{\rho}(v_j)$ should also be in $\mathcal{T}_{(t+1)\rho}$.
    Lemma \ref{lem:nonn2} shows that none of $\rho$ vertices in $\mathcal{N}_{\rho}(v_j)$ is in $\mathcal{T}_{t\rho}$. Therefore, all the $\rho$ vertices in $\mathcal{N}_{\rho}(v_j)$ and $u$ should be in $\mathcal{T}_{(t+1)\rho}\setminus\mathcal{T}_{t\rho}$, which indicates that $|\mathcal{T}_{(t+1)\rho}\setminus\mathcal{T}_{t\rho}|>\rho$, leading to a contradiction.
  \end{proof}

As a result, we know that $v\in \mathcal{T}_{(t+1)\rho}\setminus\mathcal{T}_{t\rho}$ is at most $2k+2$ hops away from a vertex in $\mathcal{T}_{t\rho}$, i.e., $l\le i+2k+2$.
Recall that $j=i+k+2$, so within $2k+2$ hops from $v_i$, we can get the shortest distance of any $v\in \mathcal{T}_{(t+1)\rho}\setminus\mathcal{T}_{t\rho}$.

  \begin{lemma}
  \label{lem:updaterho}
  After step $(2k_\rho+3)t+h+1$, vertex $v_{i+h}$ must have been settled and have relaxed all its neighbors for $h\le (l-i)$.
  \end{lemma}
  \begin{proof}
  The inductive hypothesis indicates that $v_i$ must have relaxed all its neighbors in the first $(2k_\rho+3)t$ steps. Therefore, as its neighbor, $v_{i+1}$ should be settled at step $(2k_\rho+3)t+1$.

  Next, we show that $v_{i+1}$ will be extracted from the \BDPQ{} to update all its neighbors no later than step $(2k_\rho+3)t+2$.
  We know that $v_{i+1}$ is no farther than $v$ from $s$, and $v$ is in $\mathcal{T}_{(t+1)\rho}\setminus\mathcal{T}_{t\rho}$.
  This means that there cannot be at least $\rho$ unsettled vertices in the frontier that are closer to $s$ than $v_{i+1}$, so $v_{i+1}$ will be extracted.
  Similarly, $v_{i+2}$ will be settled in step $(2k_\rho+3)t+2$. We can show this inductively that the lemma holds for all $h\le l-i$.
  \end{proof}
  Plugging in $h=l-i$, we can know that $v=v_l$ must have been settled and relaxed all its neighbors before step $(t+1)(2k+3)$.
  With \cref{lem:updaterho}, we directly get \cref{thm:rho-steps-un}.

Lastly, we note that \cref{thm:rho-steps-un} is an upper bound.
\cref{lem:updaterho} shows that $v_{i+h}$ is settled no later than step $(2k_\rho+3)t+h+1$, but it can be settled earlier.
This is shown in \cref{fig:visted-per-step} by our experiment, and on real-world graphs, the number of steps is very small.
\fi

\subsection{Comparisons and Discussions}\label{sec:discussion}

For \ssspours, the number of total steps is $O(k_\rho n/\rho)$ for undirected graphs and $O(k_n n/\rho)$ for directed graphs.
The undirected case is a factor of $O(\log \rho L)$ better than \radiuss (\radiuss does not have non-trivial span bound on direct graphs).
The work bound is off by a factor of $k_n/k_\rho$ on undirected graphs, but it applies to directed graphs.
Also, our experiments show that, on social and web graphs, $k_n/k_\rho$ is usually small (\cref{fig:k-rho}) for reasonably large values of $\rho$ (e.g., $\rho > \sqrt{n}$).

Both \ssspours and \deltass focus on practical considerations.
Since in practice we usually pick a large $\rho$, the number of steps is small. This leads to a small overhead for step-based synchronization.
\cref{thm:deltass} show that \deltass only incurs a factor of $1+L/\Delta$ more steps (recall $L=\max w(e)$) than Bellman-Ford, upper bounding the synchronization cost in practice (\cref{fig:visted-per-step}).
Regarding work, \cref{thm:rho-steps,thm:deltass} show that both \tourt{}-based and array-based versions are efficient when using proper parameters of $\rho$ and $\Delta$.
Exactly in our experiments, the best values $\rho$ and $\Delta$ match the analysis here (e.g., a large $\rho$ on social networks).
We note that Bellman-Ford has better work and span than both \ssspours and \deltass.
In fact, it seems hard to beat the work and span of Bellman-Ford (parameterized on $k_n$) if no shortcut edges are allowed. 
Our analysis provides worst-case guarantees for \ssspours and \deltass, and they seem good for the $(k_{\rho},\rho)$ parameters of many real-world graphs.
In practice, both \ssspours{} and \deltass{} exhibit better performance than Bellman-Ford because of visiting fewer vertices and edges (more efficient ``work'').
Since analyzing SSSP algorithms based on \krhograph is new, many interesting questions remain open.

The work for \radiuss and Shi-Spencer can be improved by at most a logarithmic term.

\section{Implementation Details}\label{sec:implementation}
\label{sec:impl-bdpq}

We implemented three algorithms in the stepping algorithm framework: \ssspours{}, \deltass{}, and Bellman-Ford, all using array-based \BDPQ{}.
Our implementations are simple, and are unified for the three algorithms (we only need to change \extcond and \finishcond accordingly, as shown in \cref{tab:framework}).
We present some useful optimizations we used in our implementation. Most of them apply to all the three algorithms.
Our code is available at: \url{https://github.com/ucrparlay/Parallel-SSSP}.

\myparagraph{Sparse-dense optimization.}
We use sparse-dense optimization similar to Ligra~\cite{ShunB2013}.
When the current frontier is small (sparse mode), we explicitly maintain an array of vertices as the frontier.
Otherwise (dense mode), we use an array of $n$ bit flags to indicate whether each vertex is in the current frontier, and skip those not in the frontier when processing them.
The dense mode has a more cache-friendly access pattern, and avoids explicitly maintaining the frontier array, but always needs $O(n)$ time to check all vertices.
Hence, the sparse mode is used when the frontier size is small than a certain threshold.

\hide{\xiaojun{This is not used in the last version}
Our additional optimization is that we can switch from the sparse mode to the dense mode within a step. Once we detect that the next frontier size exceeds a fixed threshold, we start to only mark the bit flags instead of scattering the vertices to the next frontier (see details in the next paragraph). This saves much work in the last sparse step before dense steps.
}

\myparagraph{Queue size estimation and scattering.}
One challenge in the sparse mode is maintaining the frontier array since the size can change dramatically during the execution.
Some existing implementations (e.g., Ligra) use a parallel pack to generate the next frontier sequence, which scans all edges incident the current frontier for two rounds (one for computing offsets and another round to pack). This can incur a large overhead.
To avoid this, we use a resizable hash table to maintain the next frontier, and scatter the vertices to the next frontier by putting them into random slots in the hash table. In the process of our algorithm, we use sampling to estimate the next frontier size in order to resize the hash table.

\myparagraph{Bidirectional relaxation for undirected graphs.}
We use a novel optimization for undirected graphs.
Before the algorithm relaxes all $v$'s neighbors (\cref{line:innerloop2} in \cref{alg:framework}), it first attempts to relax $v$ using all its neighbors.
This aims to update $v$'s distance first, so it will be more ``effective'' when $v$ relaxes other vertices later.
Another reason is that parallel SSSP implementation is usually I/O bounded.
Since in relaxations, we need to check $v$'s neighbors' distances anyway, we can load them to the cache and use them to relax $v$'s tentative distance first with small cost.
This optimization only applies to undirected graphs.

\myparagraph{Threshold estimation for \ssspours.}
\label{sec:impl-theta-sampling}
In both \deltas{} and \radiuss{} (although we did not implement \radiuss{}), the distance threshold can be directly computed.
In \ssspours{}, we need to compute the threshold (the $\rho$-th smallest element in the frontier) in each step.
We use the sampling-based idea as mentioned in \cref{sec:step-algo}.
In particular, at the beginning of \extract{}, we first sample $s=O(n/\rho+\log n)$ uniformly random samples from the current frontier.
Then we sort the samples and pick the threshold from the samples.
Since $s$ is small, this step is sequential and fast.
In \ssspours, if the frontier size is smaller than $\rho$, we pick $\theta$ as the maximum distances in the frontier.

In our experiments, we observe that in \ssspours{}, the threshold estimation in the first dense rounds is usually inaccurate. This is because in the early stage, the $\rho$-th closest distance in the frontier is usually far from the source, and during the relaxation, much more vertices go below this threshold.
Hence, we add a heuristic to adjust the threshold: using 10\% of $\rho$ at the first two dense rounds.

\myparagraph{Large neighbor sets.}
On road networks and the begin and end for all graphs, the frontier and its neighborhood are very small.
Relaxing the neighbors in a round-based manner leads to insufficient workload and thus overhead in the synchronization cost.
To optimize this case, we use a similar ``bucket fusion'' optimization proposed by Zhang et al.~\cite{zhang2020optimizing}, which is later integrated to GAPBS~\cite{beamer2015gap}.
In our \deltastarstepping and \ssspours, when processing a vertex $v$, instead of using $v$'s direct neighbors, we run a local BFS until we reach $t=4096$ vertices (or when the tentative distances reach more than $\theta$). We use these vertices as $v$'s neighborhood $\mathcal{N}(v)$, and update them all.
Note that this information is maintained and processed locally.
As such, we can extend multiple hops in one round.
We apply this optimization in sparse rounds with average edge degree fewer than 20, and thus call them \emph{super sparse} rounds.
This optimization can greatly optimize the performance for road networks, since as shown in \cref{fig:k-rho}, the values of $k_\rho$ on road graphs is large.
The impact on the performance for social and web graphs is smaller since the dense rounds spend the most time.

\section{Experiments}
\label{sec:exp}

\newcommand{\sfgraph}{scale-free network}
\hide{
\begin{table*}[htbp]
  \centering
  \caption{Add caption}
    \begin{tabular}{|c|r|rrr|r|r|r|r|r|r|r|}
    \hline
    \multirow{2}[4]{*}{\textbf{}} & \multicolumn{1}{c|}{\multirow{2}[4]{*}{\textbf{Graph}}} & \multicolumn{1}{c}{\multirow{2}[4]{*}{\textbf{n}}} & \multicolumn{1}{c}{\multirow{2}[4]{*}{\textbf{m}}} & \multicolumn{1}{c|}{\multirow{2}[4]{*}{\textbf{directed?}}} & \multicolumn{4}{c|}{\textbf{\deltas}} & \multicolumn{2}{c|}{\textbf{Bellman-Ford}} & \multicolumn{1}{c|}{\textbf{\ssspours}}\\
\cline{6-12}
& \multicolumn{1}{c|}{} & \multicolumn{1}{c}{} & \multicolumn{1}{c}{} & \multicolumn{1}{c|}{} & \multicolumn{1}{c}{\textbf{GAPBS}} & \multicolumn{1}{c}{\textbf{Julienne}} & \multicolumn{1}{c}{\textbf{\ourdelta}} & \textbf{Galois} & \multicolumn{1}{c}{\textbf{Ligra}} & \multicolumn{1}{c|}{\textbf{\ourbf}} & \multicolumn{1}{c|}{\textbf{\ourrho}} \bigstrut\\
    \hline
    \multirow{5}{*}{\begin{sideways}\textbf{social}\end{sideways}}
    & \textbf{OK} & 3M    & 234M  & undirected & \multicolumn{1}{r}{0.240 } & \multicolumn{1}{r}{0.268 } & \multicolumn{1}{r}{0.156 } & 0.178  & \multicolumn{1}{r}{0.248 } & 0.174  & \textbf{0.138}\\
          & \textbf{LJ} & 4M    & 68M   & directed & \multicolumn{1}{r}{0.062 } & \multicolumn{1}{r}{0.085 } & \multicolumn{1}{r}{} & 0.073  & \multicolumn{1}{r}{0.069 } &       & \textbf{0.057} \\
          & \textbf{TW} & 42M   & 1.47B & directed & \multicolumn{1}{r}{2.416 } & \multicolumn{1}{r}{1.820 } & \multicolumn{1}{r}{1.498 } & 1.218  & \multicolumn{1}{r}{1.548 } & 1.741  & \textbf{1.007} \\
          & \textbf{FD} & 65M   & 3.61B & undirected & \multicolumn{1}{r}{2.954 } & \multicolumn{1}{r}{2.748 } & \multicolumn{1}{r}{3.672 } & 3.120  & \multicolumn{1}{r}{5.122 } & 4.672  & \textbf{2.222} \\
          & \textbf{WB} & 89M   & 2.04B & directed & \multicolumn{1}{r}{0.964 } & \multicolumn{1}{r}{1.034 } & \multicolumn{1}{r}{} & 1.124  & \multicolumn{1}{r}{1.067 } &       & \textbf{0.754}\\
    \hline
    \multirow{2}{*}{\begin{sideways}\textbf{Road}\end{sideways}} & \textbf{GE} & 12M   & 32M   & undirected & \multicolumn{1}{r}{0.221 } & \multicolumn{1}{r}{6.624 } & \multicolumn{1}{r}{} & 0.314  & \multicolumn{1}{r}{} & \multicolumn{1}{r|}{} &  \bigstrut[t]\\
          & \textbf{USA} & 24M   & 58M   & undirected & \multicolumn{1}{r}{0.333 } & \multicolumn{1}{r}{10.155 } & \multicolumn{1}{r}{} & 0.659  & \multicolumn{1}{r}{} & \multicolumn{1}{r|}{} & \\
    \hline
    \end{tabular}%
  \label{tab:addlabel}%
\end{table*}%

}

\hide{ 
\begin{table*}[htbp]
  \centering
  \small
    \begin{tabular}{c|r|r@{}r@{}rr@{}r@{}rr@{}r@{}rr@{}r@{}r|r@{}r@{}r|r@{}r@{}rr@{}r@{}r}
    \hline
    \multicolumn{2}{c|}{\multirow{2}[3]{*}{\textbf{Graph}}} & \multicolumn{12}{c|}{\textbf{Social}}                                                         & \multicolumn{3}{c|}{\textbf{Web}} & \multicolumn{6}{c}{\textbf{Road}}  \bigstrut\\
\cline{3-23}
\multicolumn{2}{c|}{} & \multicolumn{3}{c}{\textbf{OK}} & \multicolumn{3}{c}{\textbf{LJ}} & \multicolumn{3}{c}{\textbf{TW}} & \multicolumn{3}{c|}{\textbf{FD}} & \multicolumn{3}{c|}{\textbf{WB}} & \multicolumn{3}{c}{\textbf{GE}} & \multicolumn{3}{c}{\textbf{USA}} \bigstrut[t]\\
\hline
    \multicolumn{2}{c|}{\textbf{n}} & \multicolumn{3}{c}{3M} & \multicolumn{3}{c}{4M} & \multicolumn{3}{c}{42M} & \multicolumn{3}{c|}{65M} & \multicolumn{3}{c|}{89M} & \multicolumn{3}{c}{12M} & \multicolumn{3}{c}{24M} \\
    \multicolumn{2}{c|}{\textbf{m}} & \multicolumn{3}{c}{234M} & \multicolumn{3}{c}{68M} & \multicolumn{3}{c}{1.47B} & \multicolumn{3}{c|}{3.61B} & \multicolumn{3}{c|}{2.04B} & \multicolumn{3}{c}{32M} & \multicolumn{3}{c}{58M} \\
    \multicolumn{2}{c|}{\textbf{directed?}} & \multicolumn{3}{c}{undirected} & \multicolumn{3}{c}{directed} & \multicolumn{3}{c}{directed} & \multicolumn{3}{c|}{undirected} & \multicolumn{3}{c|}{directed} & \multicolumn{3}{c}{undirected} & \multicolumn{3}{c}{undirected} \\
    \multicolumn{2}{c|}{\textbf{\#Threads}} & \multicolumn{1}{c@{}}{(1)} & \multicolumn{1}{c@{}}{(96h)} & \multicolumn{1}{c}{(SU)} & \multicolumn{1}{c@{}}{(1)} & \multicolumn{1}{c@{}}{(96h)} & \multicolumn{1}{c}{(SU)} & \multicolumn{1}{c@{}}{(1)} & \multicolumn{1}{c@{}}{(96h)} & \multicolumn{1}{c}{(SU)} & \multicolumn{1}{c@{}}{(1)} & \multicolumn{1}{c@{}}{(96h)} & \multicolumn{1}{c|}{(SU)} & \multicolumn{1}{c@{}}{(1)} & \multicolumn{1}{c@{}}{(96h)} & \multicolumn{1}{c|}{(SU)} & \multicolumn{1}{c@{}}{(1)} & \multicolumn{1}{c@{}}{(96h)} & \multicolumn{1}{c}{(SU)} & \multicolumn{1}{c@{}}{(1)} & \multicolumn{1}{c@{}}{(96h)} & \multicolumn{1}{c}{(SU)} \bigstrut[b]\\
    \hline
    \multirow{4}[1]{*}{\begin{sideways}\textbf{$\Delta$-step.}\end{sideways}} & \textbf{gapbs} &       & 0.240  &       &       & 0.062  &       &       & 2.416  &       &       & 2.954  &       &       & 0.964  &       &       & 0.221  &       &       & 0.333  &  \bigstrut[t]\\
          & \textbf{Julienne} &       & 0.268  &       &       & 0.085  &       &       & 1.820  &       &       & 2.748  &       &       & 1.034  &       &       & 6.624  &       &       & 10.155  &  \\
          & \textbf{Galois} &       & 0.178  &       &       & 0.073  &       &       & 1.218  &       &       & 3.120  &       &       & 1.124  &       &       & 0.314  &       &       & 0.659  &  \\
          & \textbf{\ourdelta} &       & 0.156  &       &       &       &       &       & 1.498  &       &       & 3.672  &       &       &       &       &       &       &       &       &       &  \\
    \hline
    \multirow{2}[0]{*}{\begin{sideways}\textbf{BF}\end{sideways}} & \textbf{Ligra} &       & 0.248  &       &       & 0.069  &       &       & 1.548  &       &       & 5.122  &       &       & 1.067  &       &       &       &       &       &       &  \\
          & \textbf{\ourbf} &       & 0.174  &       &       &       &       &       & 1.741  &       &       & 4.672  &       &       &       &       &       &       &       &       &       &  \\
    \hline
    \multirow{2}[0]{*}{\begin{sideways}\textbf{$\rho$-step.}\end{sideways}} & \textbf{\ourrho-fix} &       & 0.138  &       &       & 0.057  &       &       & 1.007  &       &       & 2.222  &       &       & 0.754  & \textbf{} &       &       &       &       &       &  \bigstrut[b]\\
    &\textbf{\ourrho-best} &       & 0.138  &       &       & 0.057  &       &       & 1.007  &       &       & 2.222  &       &       & 0.754  & \textbf{} &       &       &       &       &       &  \bigstrut[b]\\
    \hline
    \end{tabular}%
    \caption{Parallel and sequential running times for all implementations on all graphs. \ourdelta{}, \ourbf{} and \ourrho{} (-fix and -best) are our implementations. For all \deltas{} algorithms, we report the best running time across all values of parameter $\Delta$. For our \ssspours{}, we report the best running time across all values of parameter $\rho$ as \ourrho-best, and report the running time with a fixed value of $\rho$ as \ourrho-fix, which is $2^{21}$ for social and web graphs, and $2^15$ for road networks. We choose different $\rho$ values for social/web graphs and road networks because they have very different $k_\rho$-$\rho$ property. We explain more details in the main text.\label{tab:main-table}}%
\end{table*}%
}


\hide{
\begin{table*}[t]
\small
\vspace{-.5em}
\def\arraystretch{0.95}
\vspace{-.05in}
    \begin{tabular}{c|r|r@{}r@{}rr@{}r@{}rr@{}r@{}rr@{}r@{}r|r@{}r@{}r|r@{}r@{}rr@{ }r@{}r}
    \hline
    \multicolumn{2}{c|}{\multirow{2}[3]{*}{\textbf{Graph}}} & \multicolumn{12}{c|}{\textbf{Social}}                                                         & \multicolumn{3}{c|}{\textbf{Web}} & \multicolumn{6}{c}{\textbf{Road}} \bigstrut\\
\cline{3-23}    \multicolumn{2}{c|}{} & \multicolumn{3}{c}{\textbf{OK}} & \multicolumn{3}{c}{\textbf{LJ (D)}} & \multicolumn{3}{c}{\textbf{TW (D)}} & \multicolumn{3}{c|}{\textbf{FT}} & \multicolumn{3}{c|}{\textbf{WB (D)}} & \multicolumn{3}{c}{\textbf{GE}} & \multicolumn{3}{c}{\textbf{USA}} \bigstrut[t]\\
    \multicolumn{2}{c|}{\textbf{\#vertices}} & \multicolumn{3}{c}{3M} & \multicolumn{3}{c}{4M} & \multicolumn{3}{c}{42M} & \multicolumn{3}{c|}{65M} & \multicolumn{3}{c|}{89M} & \multicolumn{3}{c}{12M} & \multicolumn{3}{c}{24M} \\
    \multicolumn{2}{c|}{\textbf{\#edges}} & \multicolumn{3}{c}{234M} & \multicolumn{3}{c}{68M} & \multicolumn{3}{c}{1.47B} & \multicolumn{3}{c|}{3.61B} & \multicolumn{3}{c|}{2.04B} & \multicolumn{3}{c}{32M} & \multicolumn{3}{c}{58M} \\
    \multicolumn{2}{c|}{\textbf{\#threads}} &
    \multicolumn{1}{c@{}}{(1)} & \multicolumn{1}{@{}c@{}}{(96h)} & \multicolumn{1}{c}{(SU)} &
    \multicolumn{1}{c@{}}{(1)} & \multicolumn{1}{@{}c@{}}{(96h)} & \multicolumn{1}{c}{(SU)} &
    \multicolumn{1}{c@{}}{(1)} & \multicolumn{1}{@{}c@{}}{(96h)} & \multicolumn{1}{c}{(SU)} &
    \multicolumn{1}{c@{}}{(1)} & \multicolumn{1}{@{}c@{}}{(96h)} & \multicolumn{1}{c|}{(SU)} &
    \multicolumn{1}{c@{}}{(1)} & \multicolumn{1}{@{}c@{}}{(96h)} & \multicolumn{1}{c|}{(SU)} &
    \multicolumn{1}{c@{}}{(1)} & \multicolumn{1}{@{}c@{}}{(96h)} & \multicolumn{1}{c}{(SU)} &
    \multicolumn{1}{c@{}}{(1)} & \multicolumn{1}{@{}c@{}}{(96h)} & \multicolumn{1}{c}{(SU)} \bigstrut[b]\\
    \hline
    \multicolumn{1}{c|}{\multirow{4}[2]{*}{\begin{sideways}$\boldsymbol{\Delta}$-\textbf{step.}\end{sideways}}} & \textbf{GAPBS}
    & 3.42  & .240  & 14.2  & 1.14  & .103  & 11.0  & 58.6  & 2.42  & 24.2  & 85  & 2.95  & 28.7  & 50.8  & 1.92  & 26.5  & 2.01  & 0.22  & 9.1  & 1.83  & 0.33  & 5.5  \bigstrut[t]\\
    \multicolumn{1}{c|}{} & \textbf{Julienne}$^{[1]}$
    & 4.82  & .268  & 18.0  & 2.86  & .140  & 20.4  & 43.1  & 1.82  & 23.7  & 95  & \underline{2.75}  & 34.7  & 86.1  & 2.04  & 42.2  & 1.54  & 6.62  & 0.23  & 2.04  & 10.16  & 0.20  \\
    \multicolumn{1}{c|}{} & \textbf{Galois}
    & 3.07  & .178  & 17.3  & 1.75  & .120  & 14.6  & 29.1  & \underline{1.22}  & 23.9  & 104  & 3.12  & 33.4  & 45.4  & 2.24  & 20.2  & 2.73  & 0.31  & 8.71  & 2.81  & 0.66  & 4.26  \\
    \multicolumn{1}{c|}{} & \textbf{*\ourdelta}
    & 3.49  & \underline{.156}  & 22.4  & 1.41  & \underline{.098}  & 14.4  & 39.3  & 1.50  & 26.2  & 115  & 3.67  & 31.2  & 33.3  & \underline{1.39}  & 24.0  & 5.78  & \textbf{\underline{0.17}}  & 33.56  & 4.92  & \textbf{\underline{0.24}}  & 20.61  \bigstrut[b]\\
    \hline
    \multicolumn{1}{c|}{\multirow{2}[2]{*}{\begin{sideways}\textbf{BF}\end{sideways}}} & \textbf{Ligra}
    & 5.07  & .248  & 20.5  & 2.55  & .115  & 22.1  & 42.6  & \underline{1.55}  & 27.5  & 218.2  & 5.12  & 42.6  & 81.4  & 2.13  & 38.2  &    -   &    -   &  -     &  -     &  -     & - \bigstrut[t]\\
    \multicolumn{1}{c|}{} & \textbf{*\ourbf}
    & 3.70  & \underline{.174}  & 21.3  & 1.54  & \underline{.109}  & 14.1  & 44.7  & 1.74  & 25.7  & 144  & \underline{4.67}  & 30.8  & 48.5  & \underline{1.74}  & 27.9  & 12.98    &     \underline{0.30}  &   42.9    &  16.44     &  \underline{0.40}     &  41.56    \bigstrut[b]\\
    \hline
    \multicolumn{1}{c|}{\multirow{3}[2]{*}{\begin{sideways}$\boldsymbol{\rho}$-\textbf{step.}\end{sideways}}} & \textbf{*\ourrho-fix}
    & 3.55  & .131  & 27.0  & 1.46  & .102  & 14.3  & 37.0  & \textbf{0.97}  & 38.2  & 108  & 2.10  & 51.4  & 31.0  & 1.13  & 27.4  &    7.14   &    0.24   &    30.1   &   6.14    &  0.32     & 19.4 \bigstrut[t]\\
    \multicolumn{1}{c|}{} & \textbf{*\ourrho-best}
    & 3.41  & \textbf{.125}  & 27.2  & 1.28  & \textbf{.086}  & 14.9  & 36.9  & \textbf{0.97}  & 38.1  & 111  & \textbf{2.08}  & 53.2  & 29.5  & \textbf{1.12}  & 26.3  & 7.14  & 0.24  & 30.1  & 5.36  & 0.32  &  16.9 \\
    \multicolumn{1}{c|}{} & \textbf{
    } & \multicolumn{3}{c}{$(\rho=2^{19})$} & \multicolumn{3}{c}{$(\rho=2^{19})$} & \multicolumn{3}{c}{$(\rho=2^{21})$} & \multicolumn{3}{c|}{$(\rho=2^{20})$} & \multicolumn{3}{c|}{$(\rho=2^{22})$} &    \multicolumn{3}{c}{$(\rho=2^{15})$}    &  \multicolumn{3}{c}{$(\rho=2^{14})$} \bigstrut[b]\\
    \hline
    \end{tabular}%
    \vspace{-.15in}
  \caption{\small\textbf{Parallel and sequential running times for all implementations on all graphs.}  
  Our implementations are noted with $*$. (D) means directed graph.
  On each graph, we use bold numbers for the fastest running time, and use underline to denote the fastest parallel \deltas{} implementation and the fastest parallel Bellman-Ford implementation on each graph instance.
  For all \deltas{} algorithms, we report the best running time across all values of parameter $\Delta$. For our \ssspours{}, we report the best running time across all values of parameter $\rho$ as \ourrho-best, and report the running time with a fixed value of $\rho$ as \ourrho-fix, which is $2^{21}$ for social and web graphs, and $2^{15}$ for road networks. \newline
  [1]: Julienne does not achieve satisfactory performance on road graphs.
  We have checked this with the authors, and the reason is that Julienne was not optimized on road graphs. The reported numbers are the best among all possible values of $\Delta$.
  \label{tab:alltime}}
  \vspace{-.15in}
\end{table*}
}

\begin{table*}[t]
\small
\vspace{-1em}
\def\arraystretch{0.95}
\vspace{-.05in}
    \begin{tabular}{c|r|r@{}r@{}rr@{}r@{}rr@{}r@{}rr@{}r@{}r|r@{}r@{}r|r@{}r@{}rr@{ }r@{}r}
    \hline
    \multicolumn{2}{c|}{\multirow{2}[3]{*}{\textbf{Graph}}} & \multicolumn{12}{c|}{\textbf{Social}}                                                         & \multicolumn{3}{c|}{\textbf{Web}} & \multicolumn{6}{c}{\textbf{Road}} \bigstrut\\
\cline{3-23}    \multicolumn{2}{c|}{} & \multicolumn{3}{c}{\textbf{OK}} & \multicolumn{3}{c}{\textbf{LJ (D)}} & \multicolumn{3}{c}{\textbf{TW (D)}} & \multicolumn{3}{c|}{\textbf{FT}} & \multicolumn{3}{c|}{\textbf{WB (D)}} & \multicolumn{3}{c}{\textbf{GE}} & \multicolumn{3}{c}{\textbf{USA}} \bigstrut[t]\\
    \multicolumn{2}{c|}{\textbf{\#vertices}} & \multicolumn{3}{c}{3M} & \multicolumn{3}{c}{4M} & \multicolumn{3}{c}{42M} & \multicolumn{3}{c|}{65M} & \multicolumn{3}{c|}{89M} & \multicolumn{3}{c}{12M} & \multicolumn{3}{c}{24M} \\
    \multicolumn{2}{c|}{\textbf{\#edges}} & \multicolumn{3}{c}{234M} & \multicolumn{3}{c}{68M} & \multicolumn{3}{c}{1.47B} & \multicolumn{3}{c|}{3.61B} & \multicolumn{3}{c|}{2.04B} & \multicolumn{3}{c}{32M} & \multicolumn{3}{c}{58M} \\
    \multicolumn{2}{c|}{\textbf{\#threads}} &
    \multicolumn{1}{c@{}}{(1)} & \multicolumn{1}{@{}c@{}}{(96h)} & \multicolumn{1}{c}{(SU)} &
    \multicolumn{1}{c@{}}{(1)} & \multicolumn{1}{@{}c@{}}{(96h)} & \multicolumn{1}{c}{(SU)} &
    \multicolumn{1}{c@{}}{(1)} & \multicolumn{1}{@{}c@{}}{(96h)} & \multicolumn{1}{c}{(SU)} &
    \multicolumn{1}{c@{}}{(1)} & \multicolumn{1}{@{}c@{}}{(96h)} & \multicolumn{1}{c|}{(SU)} &
    \multicolumn{1}{c@{}}{(1)} & \multicolumn{1}{@{}c@{}}{(96h)} & \multicolumn{1}{c|}{(SU)} &
    \multicolumn{1}{c@{}}{(1)} & \multicolumn{1}{@{}c@{}}{(96h)} & \multicolumn{1}{c}{(SU)} &
    \multicolumn{1}{c@{}}{(1)} & \multicolumn{1}{@{}c@{}}{(96h)} & \multicolumn{1}{c}{(SU)} \bigstrut[b]\\
    \hline
    \multicolumn{1}{c|}{\multirow{4}[2]{*}{\begin{sideways}$\boldsymbol{\Delta}$-\textbf{step.}\end{sideways}}} & \textbf{GAPBS}
    & 3.42  & .240  & 14.2  & 1.14  & .103  & 11.0  & 58.6  & 2.42  & 24.2  & 84.7  & 2.95  & 28.7  & 50.8  & 1.92  & 26.5  & 2.01  & 0.22  & 9.1   & 1.83  & 0.33  & 5.5  \bigstrut[t]\\

    \multicolumn{1}{c|}{} & \textbf{Julienne}$^{[1]}$
    & 4.82  & .268  & 18.0  & 2.86  & .140  & 20.4  & 43.1  & 1.82  & 23.7  & 95.4  & 2.75  & 34.7  & 86.1  & 2.04  & 42.2  & 1.54  & 6.62  & 0.2   & 2.04  & 10.16  & 0.2  \\

    \multicolumn{1}{c|}{} & \textbf{Galois}
    & 3.08  & .194  & 15.9  & 1.72  & .113  & 15.1  & 29.7  & 1.23  & 24.2  & 92.2  & 2.76  & 33.4  & 45.0  & 1.45  & 31.1  & 2.80  & 0.22  & 12.8  & 2.72  & 0.29  & 9.3  \\

    \multicolumn{1}{c|}{} & \textbf{*\ourdelta}
    & 3.45  & \textbf{\underline{.123}}  & 28.1  & 2.04  & .082  & 25.0  & 39.3  & \underline{1.07}  & 36.9  & 115.4  & \underline{2.55}  & 45.3  & 62.8  & \underline{1.27}  & 49.6  & 5.54  & \textbf{\underline{0.18}}  & 30.7  & 4.81  & \textbf{\underline{0.26}}  & 18.8  \bigstrut[b]\\

    \hline
    \multicolumn{1}{c|}{\multirow{2}[2]{*}{\begin{sideways}\textbf{BF}\end{sideways}}} & \textbf{Ligra}
    & 5.07  & .248  & 20.5  & 2.55  & .115  & 22.1  & 42.6  & 1.55  & 27.5  & 218.2  & 5.12  & 42.6  & 81.4  & 2.13  & 38.2  &    -   &    -   &  -     &  -     &  -     & - \bigstrut[t]\\

    \multicolumn{1}{c|}{} & \textbf{*\ourbf}
    & 3.71  & \underline{.134}  & 27.7  & 2.58  & \underline{.095}  & 27.2  & 45.7  & \underline{1.18}  & 38.6  & 147.7  & \underline{2.72}  & 54.4  & 97.6  & \underline{1.71}  & 57.2  & 12.97  & \underline{0.30}  & 42.6  & 16.28  & \underline{0.41}  & 39.8  \bigstrut[b]\\
    \hline

    \multicolumn{1}{c|}{\multirow{3}[2]{*}{\begin{sideways}$\boldsymbol{\rho}$-\textbf{step.}\end{sideways}}} & \textbf{*\ourrho-fix}
    & 3.56  & .132  & 27.0  & 2.46  & .087  & 28.2  & 37.6  & \textbf{0.93}  & 40.6  & 112.7  & \textbf{2.02}  & 55.8  & 60.6  & 1.07  & 56.7  & 6.43  & 0.21  & 31.1  & 3.84  & 0.30  & 12.7  \bigstrut[t]\\

    \multicolumn{1}{c|}{} & \textbf{*\ourrho-best}
    & 3.42  & .125  & 27.5  & 2.07  & \textbf{{.080}}  & 28.6  & 37.6  & \textbf{0.93}  & 40.6  & 112.7  & \textbf{2.02}  & 55.8  & 57.5  & \textbf{1.06}  & 54.1  & 6.43  & 0.21  & 31.1  & 3.86  & 0.30  & 12.8  \\

    \multicolumn{1}{c|}{} & \textbf{
    } & \multicolumn{3}{c}{$(\rho=2^{19})$} & \multicolumn{3}{c}{$(\rho=2^{19})$} & \multicolumn{3}{c}{$(\rho=2^{21})$} & \multicolumn{3}{c|}{$(\rho=2^{21})$} & \multicolumn{3}{c|}{$(\rho=2^{22})$} &    \multicolumn{3}{c}{$(\rho=2^{21})$}    &  \multicolumn{3}{c}{$(\rho=2^{23})$} \bigstrut[b]\\
    \hline
    \end{tabular}%
    \vspace{-.15in}
  \caption{\small\textbf{Parallel and sequential running times for all implementations on all graphs.}  
  Our implementations are noted with $*$. (D): directed graph. (1): running time on one core. (96h): running time using 96 cores with hyperthreading (192 threads). (SU): speedup.
  On each graph, bold numbers are the fastest running time, and underline numbers denote the fastest \deltas{} implementation and the fastest Bellman-Ford implementation on each graph instance.
  For all \deltas{} algorithms, we report the best running time across all values of parameter $\Delta$. For \ssspours{}, we report the best running time across all values of parameter $\rho$ as \ourrho-best, and report the running time with a fixed value of $\rho=2^{21}$ as \ourrho-fix. \newline
  [1]: Julienne does not achieve satisfactory performance on road graphs.
  We have checked this with the authors, and the reason is that Julienne was not optimized on road graphs. The reported numbers are the best among all possible values of $\Delta$.
  \label{tab:alltime}}
  \vspace{-.15in}
\end{table*}

\myparagraph{Experimental setup.}
\label{sec:exp:setup}
We run all experiments on a quad-socket machine with Intel Xeon Gold 6252 CPUs with a total of 96 cores (192 hyperthreads). The system has 1.5TB of main memory and 36MB L3 cache on each socket. Our codes were compiled with \texttt{g++} 7.5.0 using \texttt{CilkPlus} with \texttt{-O3} flag. For all parallel implementations, we use all cores  and \texttt{numactl -i all}, which evenly spreads the memory pages across the
processors in a round-robin fashion.


We implemented three algorithms based on the framework in \cref{sec:framework}: Bellman-Ford (\ourbf), \deltass{} (\ourdelta), and \SSSPalgo{} (\ourrho).
We use array-based \BDPQ{} because we observe that when the output size of \extract{} is large, the array-based implementation has better performance than \tourt{} (see more details in the full paper).
For all graphs we use, the best running time is achieved using a reasonably large $\rho$.
We compare our implementations with state-of-the-art SSSP implementations: Bellman-Ford algorithm in Ligra~\cite{ShunB2013}, \deltas{} in Julienne~\cite{dhulipala2017}, GAPBS~\cite{beamer2015gap,zhang2020optimizing}, and Galois~\cite{nguyen2013lightweight}. Throughout the section, when we refer to ``\deltas{}'', it includes our \deltass{}, and the existing \deltas{} in Julienne, GAPBS and Galois.

We test seven graphs, including four social networks com-orkut (OK)~\cite{yang2015defining}, Live-Journal (LJ)~\cite{backstrom2006group}, Twitter (TW)~\cite{kwak2010twitter} and Friendster (FT)~\cite{yang2015defining}, one web graph WebGraph (WB) \cite{webgraph}, and two road graphs~\cite{roadgraph} RoadUSA (USA) and Germany (GE). The graph information is provided in Table \ref{tab:alltime}.
In almost all experiments, the social and web graphs show a similar trend.
This is because they follow similar power-law-like degree distribution. Throughout the section, we use ``\sfgraph{s}'' to refer to social and web graphs.
On \sfgraph{s}, we set edge weight uniformly at random in range $[1,2^{18})$. On road graphs, the edge weights are from the original dataset, which is up to $2^{25}$.

For all \deltas{} algorithms (except for \cref{fig:intro:delta} where we vary~$\Delta$), we report the \emph{best running time} across all $\Delta$ values. When we report average of multiple sources, we first find the best $\Delta$ value on one source, and use it for other sources. We do this for every graph-implementation combination.
For all \ssspours{} algorithms (except for in \cref{tab:alltime} where we explicitly report the best running time across $\rho$ values), we use a fixed value of $\rho=2^{21}$. For most experiments, we report the average of 10 sources.
When taking the average is meaningless, we use one representative source.

In this section, we will first discuss the overall performance of all implementations. We then compare some statistics to better understand the performance of \ourrho{}, \ourdelta{} and \ourbf{}. We evaluate the number of vertices visited by the algorithm as an indicator of the overall work. Since road graphs exhibit different properties from the \sfgraph{s}, we then discuss road graphs separately. We also analyze the two algorithms \ssspours{} and \deltas{} with their corresponding parameters.
Due to page limit, we will discuss the $k_\rho$ properties for each graph in the full version, and only show the $k_\rho$-$\rho$ curves in this paper (\cref{fig:k-rho}). 
\ifx\fullversion\undefined
Due to page limit, we postpone some figures, discussions and more experiments to the full paper~\cite{ssspfull} (e.g., using different machines and different types of sources).
\else \fi
Generally, we show that our \ssspours{} has \textbf{especially good performance on \sfgraph{s}}, and the performance gain of \ssspours is from three aspects: \textbf{good parallelism}, \textbf{less overall work}, and \textbf{more evenly distributed work to all steps}. We summarize conclusions and interesting findings at the end of this section.

\myparagraph{Overall Performance.} We present the running time of all
implementations in \cref{tab:alltime}.
In all cases, one of our implementations achieves the best performance, and is 1.14$\times$ to orders of magnitude faster than the previous implementations. We show a heat map of relative parallel running time in \cref{fig:exp:heat}.

On \sfgraph{s}, \ourrho{} and \ourdelta{} outperform all existing implementations.
\ourrho{} has better performance. 
On average over five graphs, \ourrho{} is $1.41\times$ faster than Galois, $1.83\times$ faster than Julienne and GAPBS, and $1.93\times$ faster than Ligra.

On road graphs, \ourdelta{} is the fastest, and \ourrho{} is also competitive.
Ligra did not finish in 30 seconds on road graphs, since Ligra uses plain Bellman-Ford that is inefficient for graph with deep shortest-path tree (more than $10^4$, see \cref{fig:k-rho}).
Our \ourbf{} with the neighbor-set optimization (see \cref{sec:implementation}) finishes on both graphs in about 0.4s.

We report the sequential running time of the corresponding parallel version and show self-speedup in Table \ref{tab:alltime}.
We note that comparing the sequential running time
of different implementations does not seem useful because both \deltas and \ssspours are parameterized.
To get the best sequential performance, one should just use a small $\Delta$ or $\rho$.
The reported time is the sequential performance using the corresponding parameter that performs best in parallel, and it makes more sense just to compare the speedup numbers.
The self-speedup of \ourrho{} is almost always the best among all implementations (\ourdelta{} is close but slightly worse). Hence, the good performance of \ourrho{}, especially on \sfgraph{s}, is partially due to good scalability.  In other words, \ourrho{} achieves the best ``work-span tradeoff'' in practice.



Among the implementations of the same algorithm, \ourbf{} outperforms Ligra on all graphs.
For all \deltas{} algorithms, \ourdelta{} is also the fastest on all graphs.
Overall, our three algorithms outperform existing implementations, indicating the efficiency of stepping algorithm framework for parallel SSSP implementations.

\label{sec:exp:analysis}

\hide{
\begin{figure}[t]
  \centering
  \includegraphics[width=\columnwidth]{figures/heat.pdf}
  \vspace{-.2in}
  \caption{ \small \label{fig:exp:heat} \textbf{Parallel running time relative to the fastest on each graph.} \mdseries Implementations noted with $*$ are ours. ``Ave.'' means average on \sfgraph{s} (over the five graphs) and road graphs (over the two graphs), respectively. (96h) means parallel running time (96 cores, 192 hyperthreads). (1) means the sequential running time of the corresponding parallel version. (SU) means the speedup number.
  \ourrho{}-fix denotes uses a fixed parameter $\rho$ ($2^{21}$ for \sfgraph{s} and $2^{15}$ for road graphs). \ourrho{}-best is the best running time using all values of $\rho$.
  }
  \vspace{-.2in}
\end{figure}
}

\begin{figure*}[t]
\vspace{-.5em}
\begin{tabular}{cccc}
  \includegraphics[width=0.45\columnwidth]{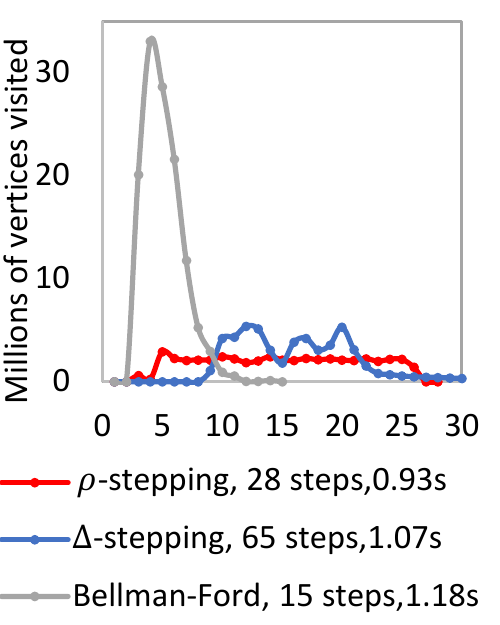}&
  \includegraphics[width=0.45\columnwidth]{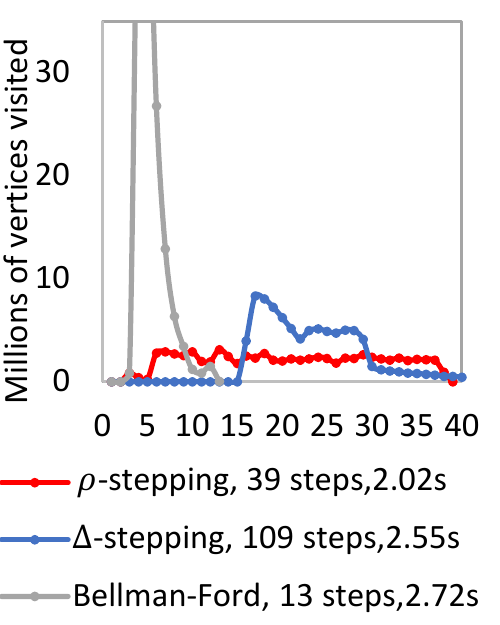}&
  \includegraphics[width=0.45\columnwidth]{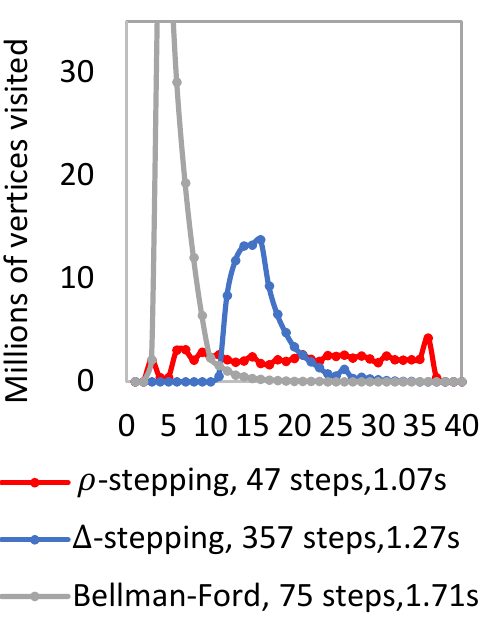}&
  \includegraphics[width=0.45\columnwidth]{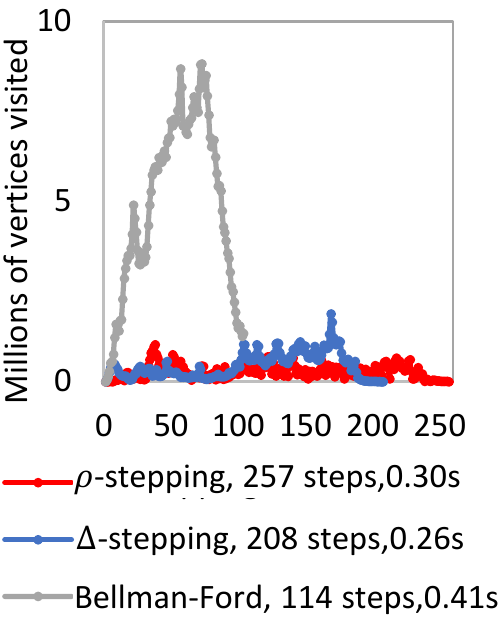}\\
\bf (a). TW& \bf (b). FT & \bf (c). WB& \bf (d). USA \\
\end{tabular}
\hide{
\begin{tabular}{ccc}

  \includegraphics[width=0.5\columnwidth]{figures/per-round/GE.pdf}&

  \includegraphics[width=0.5\columnwidth]{figures/per-round/WB-e.pdf}&
  \includegraphics[width=0.5\columnwidth]{figures/per-round/GE-e.pdf}&
  \includegraphics[width=0.5\columnwidth]{figures/per-round/USA-e.pdf}\\
\bf (e). WB&\bf (f). GE&\bf (g). USA\\
\end{tabular}}
  \caption{\textbf{Number of visited vertices in each step in \ourrho, \ourdelta and \ourbf.} \mdseries Here we only run on one source vertex, since it has unclear meaning to compute the average of multiple runs on each step.  Hence, the runtimes can be different from Table \ref{tab:alltime} (average on 100 runs from 10 source vertices), and some curves are bumpy. We use 96 cores (192 hyperthreads).\label{fig:visted-per-step}}
\vspace{-1em}
\begin{minipage}{1.15\columnwidth}
  \includegraphics[width=\columnwidth]{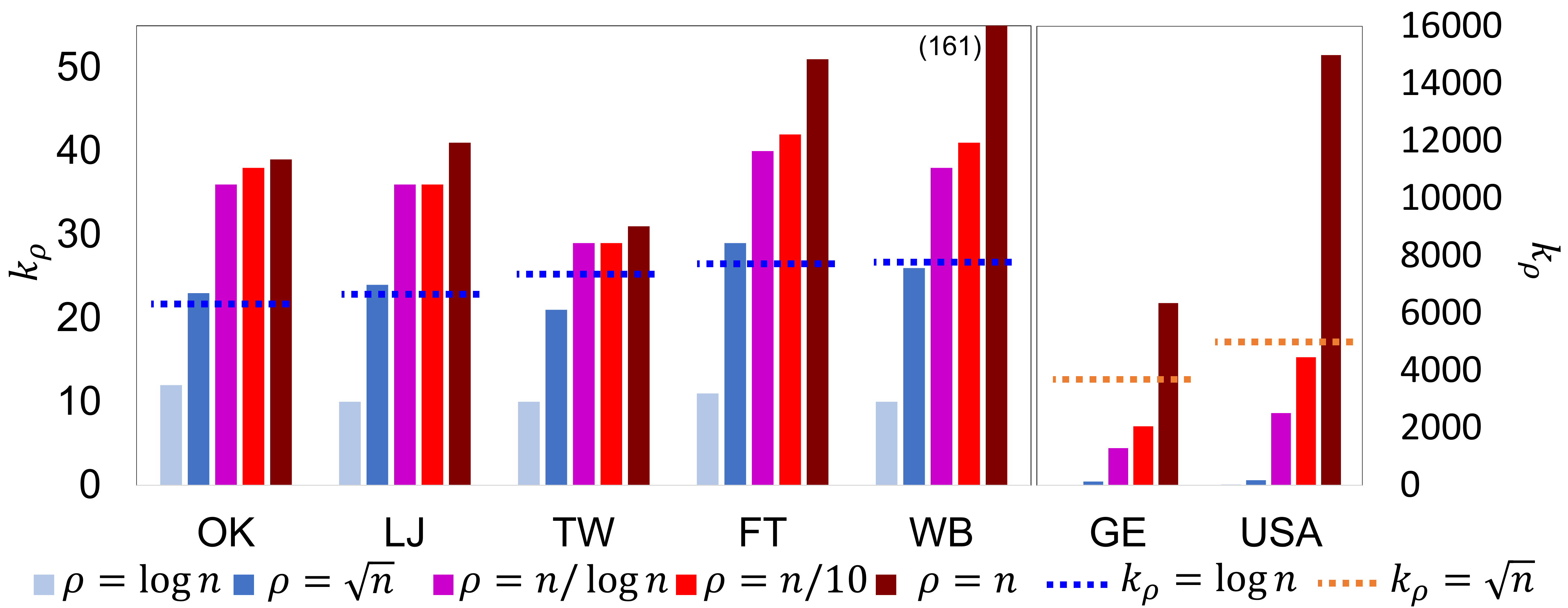}
  \caption{\small\textbf{The values of $k_{\rho}$ with different values of $\rho$ for different graphs.}}\label{fig:k-rho}
\end{minipage}\hfill
\begin{minipage}{0.85\columnwidth}
  \includegraphics[width=.95\columnwidth]{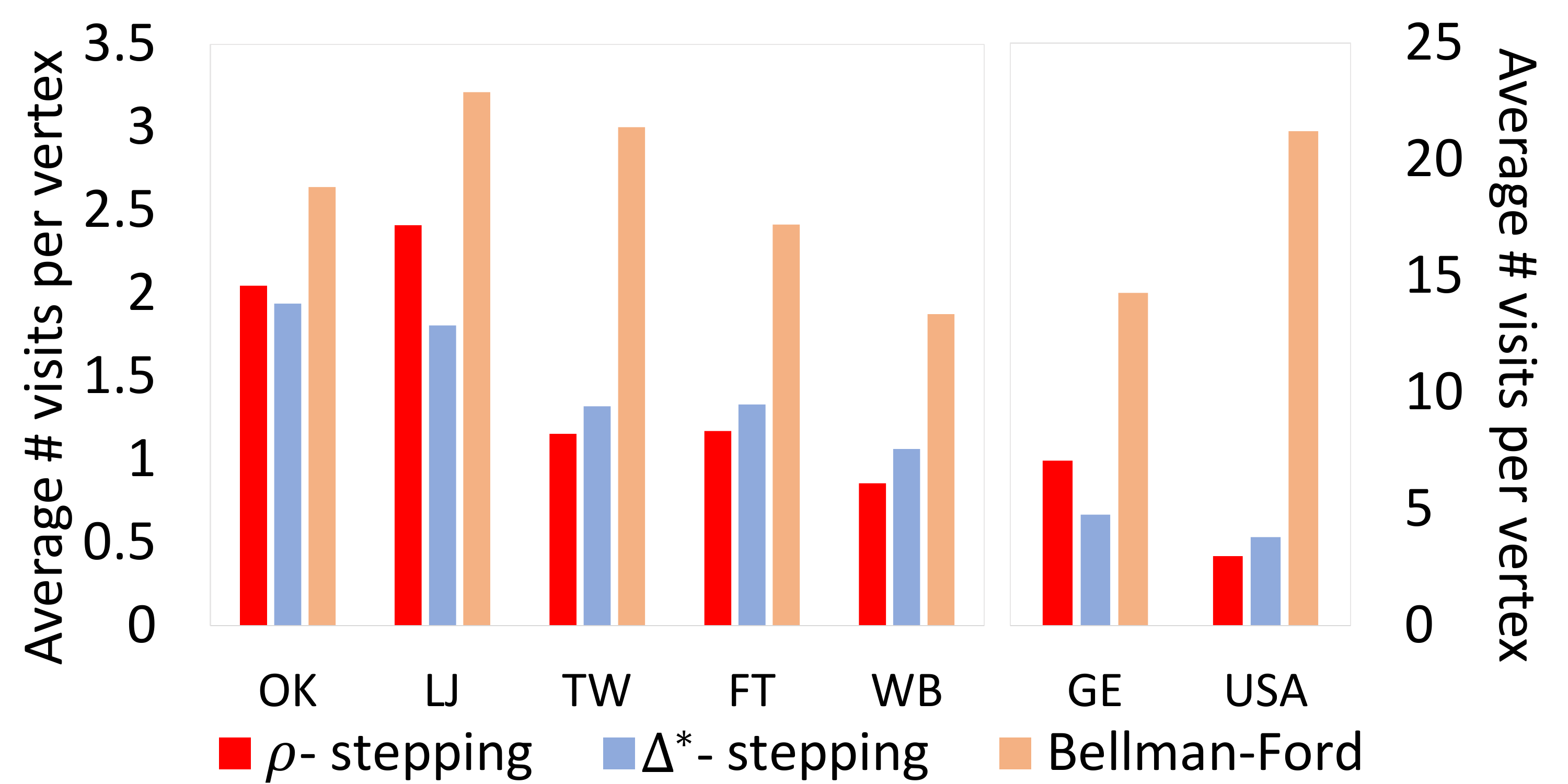}
  \caption{\small\textbf{Number of visits per vertex and per edge, respectively, for \ourrho, \ourdelta and \ourbf on all graphs.} }\label{fig:enqueue}
\end{minipage}
\vspace{-2.5em}
\end{figure*}

\hide{
\begin{figure*}
\centering
  \includegraphics[width=1.3\columnwidth]{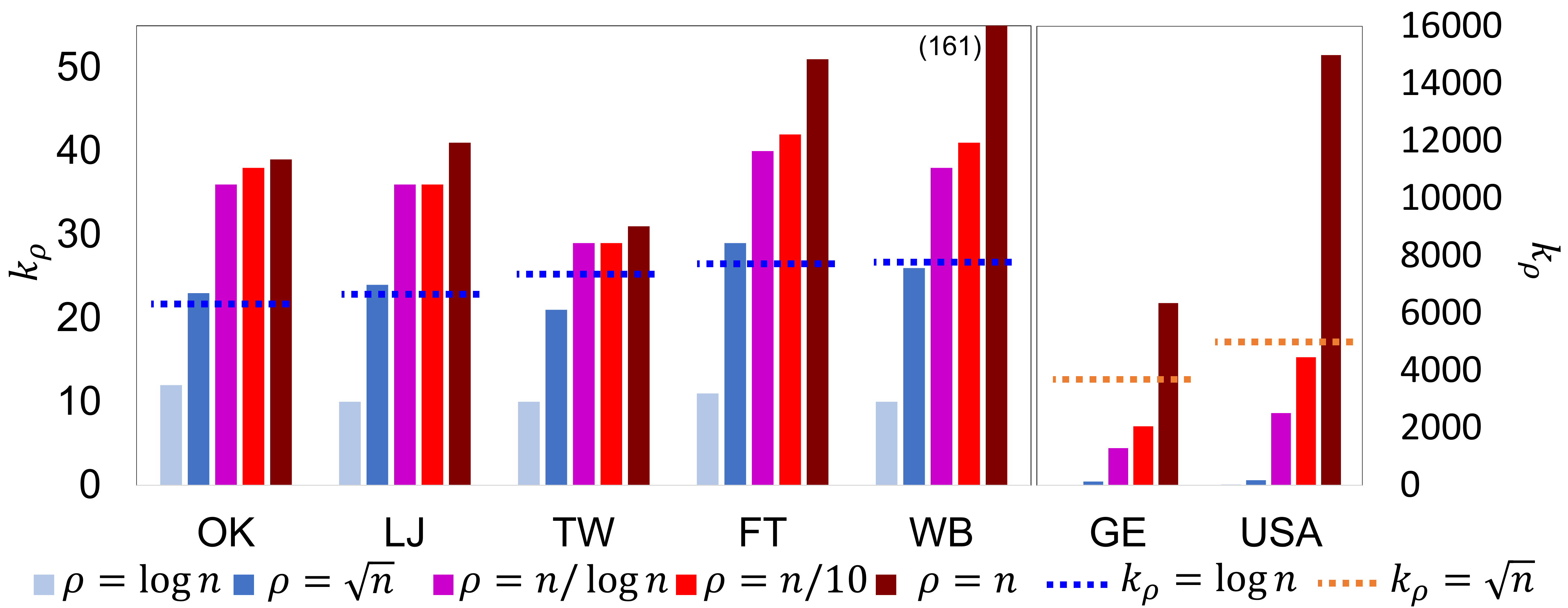}
  \caption{\textbf{The values of $k_{\rho}$ with different values of $\rho$ for different graphs.}}\label{fig:k-rho}
  \includegraphics[width=.9\columnwidth]{figures/enqueue/num-enqueue-vertex.pdf}
    \hspace{1em}
\includegraphics[width=.9\columnwidth]{figures/enqueue/num-enqueue-edge.pdf}
  \caption{\textbf{Number of visits per vertex and per edge, respectively, for all our algorithms (\ourrho, \ourdelta and \ourbf) on all graphs.} }\label{fig:enqueue}
  \vspace{2em}
\end{figure*}
}

\hide{
\begin{figure*}
  \begin{tabular}{cccc}
    \multicolumn{4}{c}{\includegraphics[width=1.3\columnwidth]{figures/delta/caption.pdf}}\\
    \includegraphics[width=0.5\columnwidth]{figures/delta/delta_orkut.pdf} & \includegraphics[width=0.5\columnwidth]{figures/delta/delta_LJ.pdf} & \includegraphics[width=0.5\columnwidth]{figures/delta/delta_twitter.pdf} &
    \includegraphics[width=0.5\columnwidth]{figures/delta/delta_friendster.pdf} \\
\bf (a). OK& \bf (b). LJ & \bf (c). TW & \bf (d). FT \\
  \end{tabular}
  \begin{tabular}{ccc}
\includegraphics[width=0.5\columnwidth]{figures/delta/delta_web.pdf} & \includegraphics[width=0.5\columnwidth]{figures/delta/delta_GE.pdf} &
\includegraphics[width=0.5\columnwidth]{figures/delta/delta_US.pdf}\\
\bf (e). WB&\bf (f). GE&\bf (g). USA\\
  \end{tabular}
  \caption{\textbf{$\Delta$-stepping relative running time with varying $\Delta$.} We use 96 cores (192 hyperthreads).}\label{fig:alldelta}

  \includegraphics[width=\columnwidth]{figures/rho-time/rho.pdf}
  \includegraphics[width=\columnwidth]{figures/rho-time/road.pdf}
  \caption{\textbf{\ssspours{} relative running time with varying $\rho$.} We use 96 cores (192 hyperthreads).}\label{fig:allrho}
\end{figure*}
}

\myparagraph{Number of visits to vertices.}
Unlike Dijkstra, other parallel SSSP algorithms can visit each vertex or edge more than once.
While this allows for parallelism, the total work is also increased.
To show how much ``redundant'' work is done for the stepping algorithms, we measure the average number of visits per vertex (\cref{fig:enqueue}), and the number of visited vertices in each step on four representative graphs (\cref{fig:visted-per-step})
\ifx\fullversion\undefined
\footnote{We also measured the number of visited edges, which show very similar trend to the vertices. Due to page limit, we report the numbers in the full version of the paper, which also shows the results for all seven graphs.}.
\else
\footnote{We also measured the number of visited edges, which show very similar trend to the vertices. We present the results in \cref{fig:visted-per-step}}.
We show results for four representative graphs, and full results are in \cref{fig:visted-per-step}.
\fi
We note that the other systems vary a lot in implementation details, and it is hard to directly measure these quantities from their code.
Hence, we compare among our implementations. 
For the same reason, \ourdelta{} may not precisely reflect the numbers of other \deltas{} implementations. In this paragraph, we first focus on the \sfgraph{s}, and discuss road graphs later.

\Cref{fig:enqueue} shows the average number of visits per vertex.
\ifconference
On the two small graphs (OK and LJ), since the work cannot saturate all 192 threads, \ourrho act similar to Bellman-Ford to maximize parallelism and uses visits more vertices than \deltas{}.
\else
On the two small graphs (OK and LJ), since the work cannot saturate all 192 threads, both \ourdelta and \ourrho act similar to Bellman-Ford to maximize parallelism.
\fi
\ifconference
For the larger graphs (TW, FT, and WB), \ourrho{} always triggers the smallest average visit to vertices.
\else
For the three larger graphs (TW, FT, and WB), \ourrho{} always triggers the smallest average visit to vertices and edges.
\fi
The trend showed in \cref{fig:enqueue} exactly matches the sequential time of each implementation.
Hence, one advantage of \ourrho{} over \ourbf{} and \ourdelta{} on \sfgraph{s} is less total work.


\Cref{fig:visted-per-step} shows the number of visited vertices per step.
In \ourbf{}, the numbers always grow quickly to a large value, stay for a few steps, and finish quickly.
Although usually using the fewest steps, \ourbf{} is the slowest, since the dense steps cause many redundant relaxations.
\ourdelta{} usually uses more steps than both \ourbf and \ourrho{}.
In most of the steps (at the beginning and the end), \ourdelta{} visits only a small number of vertices, but the peak values are much higher than \ourrho{}.
\ourrho{} shows a more even pattern across the steps: in most of the steps, it processes a moderate number of vertices, and the peak value is much smaller than \ourdelta{} or \ourbf{}. 

These patterns in \cref{fig:visted-per-step} reflect the nature of the three algorithms.
Bellman-Ford always visits all vertices in the frontier in each step.
This created significant redundant work.
\deltass{} controls work-span tradeoff based on distances.
On \sfgraph{s}, it reaches the peak work in some middle steps, which is significantly higher than other steps.
\ssspours{} controls the work-span tradeoff using the number of vertices processed per step.
We believe on \sfgraph{s}, this quantity is a closer indicator to the actual ``work'' in each step than the distance gap is.
In other words, \ssspours{} controls the work in each step that is minimal to saturate all processors, so it explores sufficient parallelism with minimized redundant work.

\ifx\fullversion\undefined
\else
\myparagraph{$k$-$\rho$ property for different graphs.} To better understand the properties of different graphs, we show the $\rho$-$k_\rho$ curves in \cref{fig:k-rho}. We note that computing the exact value of $k_\rho$ is expensive (which makes \radiuss{} impractical), and we estimate the value using 100 samples.
As mentioned, this indicates how much potential parallelism we can get on the original graph (i.e., without shortcuts). We show the value of $k_\rho$ when $\rho$ is $\log n$, $\sqrt{n}$, $n/\log n$, $n/10$, and $n$, respectively. Note that $k_n$ is the shortest path tree depth.

On all \sfgraph{s}, we observe that $k_n$ is $O(\log n)$ (about $2\log n$). Interestingly, all \sfgraph{s} are $(\log n,\sqrt{n})$-graph, which means that almost all vertices can reach their $\sqrt{n}$ nearest vertices by $\log n$ hops.
This is not surprising for \sfgraph{s}, which have some ``hubs'' that are well connected to other vertices.
These vertices are easy to be reached by any source in a few steps.
Once reached, they can reach a lot of other vertices, which quickly accumulate $\sqrt{n}$ nearest neighbors for any source.
The road graphs have a different $\rho$-$k_\rho$ pattern.
On GE and USA, it takes more than $100$ hops to reach $\sqrt{n}$ nearest vertices, and $O(\sqrt{n})$ steps to reach $O(n)$ nearest vertices.
We believe the result is reasonable since road graphs are (almost) planar graphs.
\fi

\myparagraph{Discussions for Road Graphs.}
Road graphs are (almost) planar and have different $k$-$\rho$ patterns than other graphs, and the shortest-path trees are deep and slim.
Hence, without the special optimizations (e.g., in Ligra and Julienne), the performance is slow.
As mentioned in \cref{sec:implementation}, our optimization expands multiple levels in the shortest-path tree in one step.
This makes the performance of our implementations competitive or better than GAPBS and Galois.

On road graphs, \ourdelta{} is the fastest.
This somehow indicates that expanding with distance may be a good strategy for road graphs.
One possible reason is that they are planar graphs with Euclidean distance.
Hence, setting fixed-width ``annuli'' seems a reasonable work-parallelism tradeoff, when using a proper $\Delta$.
Since the frontier on road graphs is small, \ourrho{} has insufficient frontier size in each step for enough parallelism.
Hence, it is hard for \ourrho{} to control the number of vertices visited precisely, and the performance is slightly slower than \ourdelta.
However, \ourrho{} has more stable performance than \ourdelta{} in the parameter space (\cref{fig:intro:delta,fig:rho-time}).

\myparagraph{\deltastepping and $\Delta$.} We test all \deltas{} algorithms with varying $\Delta$ on all graphs.
For each test case, we normalize the running time to the best time across all $\Delta$ values.
For page limit, we present four graphs in \cref{fig:intro:delta}, and the full results in the full paper~\cite{ssspfull}.

On the same graph, the best choice of $\Delta$ varies a lot
for different systems.
On TW, Julienne's best $\Delta$ is $2^{12}\times$ larger than Galois's.
The best $\Delta$ for one system can make another system up to $4\times$ slower.
The selection of $\Delta$ in one system does not generalize to other systems.
Secondly, even though all \sfgraph{s} have the same edge weight distribution, for the same implementation, the best choice of $\Delta$ varies a lot on different graphs.
Therefore, the selection of $\Delta$ on one graph does not generalize to other graphs.
On the same graph, the performance is sensitive to the value of $\Delta$.
Usually, $2$--$4\times$ off may lead to a 20\% slowdown, and $4$--$8\times$ off may lead to a 50\% slowdown.
A badly-chosen $\Delta$ can largely affect the performance.
As a result, for every graph-implementation combination, we have to search the best parameter $\Delta$. Fortunately, we find out that different sources show relatively stable performance for the same implementation-graph pair. This is also the conventional way of tuning $\Delta$ (we also did so).
\ifx\fullversion\undefined
We present the results in the full paper \cite{ssspfull}.
\else
We present the results in \cref{fig:differentsource}.
\fi

\ifconference
\else
Also, for almost all systems, the performance on road graphs is more sensitive to the value of $\Delta$ than on \sfgraph{s}.
\fi

\myparagraph{\SSSPalgo{} and $\rho$.} We test \ssspours{} with varying $\rho$.
When $\rho$ is small, the running time increases significantly.
This is also due to the lack of parallelism (similar to when \deltas{} uses small $\Delta$).
When $\rho$ gets large, the performance drops by no more than 20\%.
The best choices of $\rho$ are very consistent on different graphs. This is because the choice of $\rho$ in practice depends on the right level of parallelism we want to achieve, instead of the graph structure or edge weight distribution.
As discussed, \ssspours{} distributes work more evenly to each step.
The goal of setting $\rho$ is to enable enough work to exploit full parallelism in each step, but without introducing more redundant work.
The performance on road graphs are less sensitive, probably because the frontier size seldom reaches $\rho$ in road graphs.
We also tested \ssspours{} on various machines. We observe that the best choice of $\rho$ is still relatively consistent among different settings. We will present more results in the full paper.

Generally speaking, using large $\Delta$ or $\rho$ gives better (and more stable) performance than small $\Delta$ or $\rho$ values.
This is not surprising because when these parameters are large, \deltas{} and \ssspours{} degenerate to Bellman-Ford that still has reasonable performance on social networks.
When the parameters are small, \deltas{} and \ssspours{} both degenerate to Dijkstra and loses parallelism.

\myparagraph{Summary. }
\label{sec:exp-sum}
In summary, our \ourrho{} generally achieves the best performance on the five \sfgraph{s}. On average of the five graphs, \ourrho{} is 1.41-1.93$\times$ faster better than existing systems.
On the two road graphs, \ourdelta{} always has the best performance, which is at least 14\% better than existing systems.
The good performance of \ourrho{} on \sfgraph{s} comes from three aspects. The first is scalability, indicated by the good self-speedup. 
Secondly, it visits fewer vertices and edges on large \sfgraph{s}, which indicates less overall work.
Lastly, the work is more evenly distributed to each step, such that each step can exploit sufficient parallelism, and also avoid performing ``ineffective'' work to relax the neighbors of unsettled vertices.
This also indicates that on \sfgraph{s} with uniformly distributed edge weights, controlling the number of vertices visited per step is a good strategy. 
On road graphs with Euclidean distance, \deltas{} shows better performance. 

Our \ourrho{} implementation generally shows stable performance across $\rho$ values on all tested graphs.
A fixed $\rho$ almost always gives performance within 5\% off the performance with the best $\rho$. 

Finally, on all tested graphs, \ourrho{} and \ourdelta{} are faster than all existing SSSP implementations (except for RoadUSA, \ourrho{} is 0.01s slower than Galois).
\ourbf{} is faster than Ligra on all graphs.
This indicates the efficiency of the stepping algorithm framework on implementing and optimizing parallel SSSP algorithms.

\section{Related Work on Parallel SSSP}\label{sec:related}

\myparagraph{Practical parallel SSSP implementations}.
There have been dozens of practical implementations of parallel SSSP.
In this paper, we compared to a few of them.
Galois~\cite{nguyen2013lightweight} uses an approximate priority queue ordered by integer metric with NUMA-optimization
to improve the performance of SSSP.
GraphIt~\cite{zhang2018graphit, zhang2020optimizing} proposed a priority queue abstraction and a new optimization, bucket fusion, to reduce the synchronization overhead of \deltass{}.
The optimizations are later adopted by GAPBS~\cite{beamer2015gap}, which is the one we compared to.
Julienne~\cite{dhulipala2017} proposed and used the bucketing data structure to order the vertices for
\deltas{} based on semisorting~\cite{gu2015top}.
Ligra~\cite{ShunB2013} includes one of the most efficient Bellman-Ford implementations.

There has also been a significant amount of work on other implementations, include those on the distributed setting~\cite{malewicz2010pregel,zhu2016gemini,besta2017push}, GPUs~\cite{davidson2014work,wang2016gunrock}, among many others.
Our reported running time in this paper is much faster than in these papers on the same graphs, and comparing the superiorities of different settings on parallel SSSP is out of the scope of this paper.
Parallel SSSP based on parallel priority queues are reviewed in \cref{sec:bdpq}.

\myparagraph{Theoretical work on parallel SSSP}.
There has been a rich literature of theoretical parallel SSSP algorithms.
Among them, many algorithms~\cite{ullman1991high,klein1997randomized,cohen1997using,Shi1999,cohen2000polylog,spencer1997time} achieve very similar bounds to \radiuss~\cite{blelloch2016parallel} we discussed in this paper, but require adding shortcut edges.
Basically the product of work and span is $\tilde{\Theta}(nm)$ (referred to as the transitive closure bottleneck~\cite{KarpR90}).
Some algorithms are analyzed based on edge weights~\cite{meyer2002buckets,meyer2001heaps}, and many others are on approximate shortest-paths~\cite{elkin2019hopsets,miller2015improved,cao2020efficient,li2020faster,andoni2020parallel} and other models~\cite{augustine2020shortest,kuhn2020computing,ghaffari2018improved}.
While these algorithms are insightful, to the best of our knowledge, none of them have implementations.




\section*{Acknowledgement}
This work is in partial supported by NSF grant CCF-2103483.
\balance
\bibliographystyle{abbrv}

\ifx\fullversion\undefined
\else
\appendix

\section{Lower Bound for Batch-Update Work}\label{app:lower-bound}

We now show that in order to preserve the total ordering of all records, we need at least $O(\log (n+b))$ work per update, where $n$ is the number of records in the priority queue.
For a batch update of $b$ unordered records, there are ${b+n\choose b} \cdot b!$ total possible input cases.
In the comparison model, the lower bound for work of the batch update is:
 \begin{align*}
     \log_2\left({{b+n\choose b} \cdot b!}\right)=&\log_2\left({(b+n)!\over b!n!}\cdot b!\right)\\
     =&\log_2\left({\prod_{i=1}^{b}{(n+i)}}\right)\\=&\Omega(b\log (n+b))
 \end{align*}

Hence, to achieve better bounds, like in \cref{thm:main-tour-tree}, we cannot maintain the total ordering of all records.

\section{Finding the $\rho$-th Element in a List}\label{app:cpam}

We now discuss how to find the $\rho$-th element in a list, which is used in \ssspours.
One solution is to pick $s=c(n/\rho+\log n)$ random samples where $c>1$ is a constant.
We can then sort the samples and pick the $(\rho s/n)$-th one.
By setting up the parameters correctly and using Chernoff Bound, we can show that with high probability, the selected output is within $(1-\epsilon)\rho$-th and $(1+\epsilon)\rho$-th element for some constant $0<\epsilon<1$.
We can check if the output is in this range by calling \extract of the \BDPQ{}, and if not, we can just retry until succeed.
We note that instead of finding the exact $\rho$-th element, using an approximate $\rho'$-th element within a constant factor of $\rho$ will not affect any bounds for \ssspours.
In practice, we set $c=10$ and run the entire sampling algorithm sequentially, and this simple approach always gives a satisfactory output.
However, the work and span bounds are high probability bounds, instead of deterministic bounds.
Also, it only applies to the case when $\rho=\Omega(\sqrt{n})$.
Otherwise, the sorting cost will dominate.

We are aware of a data structure~\cite{cpam} (unpublished work) that can find the $\rho$-th element in a list efficiently.
More accurately, it takes $O(\rho)$ work and $O(\log n)$ span to find an approximate $\rho$-th element, and $O(\log (n/b))$ work per element in a batch insertion or a batch deletion of $b$ elements.
Basically, this is a blocked linked list data structure that looks like a binary search tree, but each leaf contains a block of $[\rho,3\rho]$ unsorted elements.
When the leaf grows too big or too small, we split or merge it with an amortized constant work per update and $O(\log n)$ span.
Hence, the cost per update is $O(\log (n/b))$ (to find the leaf node) plus a constant (amortized work for future leaf split or merge).
Since the first leaf contains the smallest $\rho$ to $3\rho$ \record{s}, we can just pick the largest key among this leaf.
This gives us a key ranked within $\rho$ and $3\rho$, with the exception if we have fewer than $\rho$ records in total (in this case the largest is returned).

Using this data structure requires some changes since we have to explicitly generate the batch of updates (relaxations).
Theoretically, this is achievable as described below, but in practice the overhead can be significant so we do not use it.
This step is essentially Ligra's approach to generate the next frontier~\cite{ShunB2013}.
For each step, we allocate an array, and the size is the number of total neighbors of the vertices processed in this step.
Then we mark a flag for each successful relaxation, and after all relaxations are completed, we pack this array, and this is the batch of updates that applies to the next frontier (the priority queue in our case).
This step takes linear work and logarithmic span, so the cost is asymptotically bounded by the relaxation cost.


\section{Applying \BDPQ{} to Shi-Spencer}
\label{app:shi-spencer}

Shi-Spencer algorithm~\cite{Shi1999} is a parallel SSSP algorithm with theoretical guarantee.
This algorithm is complicated, which harms its practicability.
We note that if we use our tournament-tree-based \BDPQ{s} to replace the parallel search trees in Shi-Spencer algorithm, we can improve the work bound by up to a logarithmic factor. This algorithm needs preprocessing to shortcut each vertex to its $\rho$ nearest vertices\footnote{The notation used in the original paper is $k$-nearest vertices, we use $\rho$ here to be consistent with the results in our paper.}.

The main bottleneck in Shi-Spencer algorithm is to maintain a binary search tree for each vertex that keeps track of its $\rho$ nearest neighbors.
Once a vertex is settled, it will be removed from all the lists of their neighbors.
Such removals can happen for $O(m+n\rho)$ times in total.
The algorithm requires querying and extracting the $\rho$ nearest vertices of a vertex, which is used for $O(n)$ times in total.
In their original algorithm, they use parallel 2-3 trees to maintain the closest neighbor sets, which incur a logarithmic cost per update.

We note that the functionalities can be implemented by our \tourt-based \BDPQ{s}.
More importantly, similar to the stepping algorithms, since there are more updates than extractions, we can apply the updates lazily for better work efficiency.
In this case, the worst case is all the updates distribute evenly---each extract needs to first lazily update $O(m/n+\rho)$ previous updates and then apply the $O(n)$ extractions.
By the distribution lemma (\cref{lem:distribute}), the total work is:
$$O\left((m+n\rho)\log\frac{n\cdot n}{m+n\rho}\right)=O\left((m+n\rho)\log\frac{n^2}{m+n\rho}\right)$$
Hence, for dense graphs or if the algorithm picks a large $\rho$, the new work can be improved by up to a logarithm factor.
We note that a few other search trees and data structures are maintained in Shi-Spencer, but the total cost is  asymptotically bounded by the cost discussed in the above paragraph.

\section{Fully Dynamic \BDPQ}\label{sec:fully-dynamic}

In \cref{sec:bdpq} we discuss and analyze the \BDPQ's implementations assuming the universe of the records has a fixed size $n$, which is $|V|$ for SSSP.
This is good enough to derive the bounds in \cref{tab:bounds}.
However, we note that it is not hard to extend our algorithms to deal with the case without this assumption.

To do so, we need an explicit batch of updates, which can include insertions, deletions, and/or just key updates.
Our algorithm will in turn process the same type of operations.
For insertions (say $k$ in total), we can grow the tree size from $2n-1$ to $2(n+k)-1$, copy the leaf nodes from the range of $n$ to $n+k-1$ to the range of $2n$ to $2n+k-1$, insert the new keys to the range of $2n+k$ to $2(n+k)-1$ ($k$ in total), and update the corresponding tree paths.
The cost for $k$ insertions is $O(k+\log n)$.
For deletions (again, say $k$ in total), we can use the last tree leaves to fill in the holes for the deletions, and update the tree paths.
The cost for $k$ deletions is $O(k\log (n/k))$.
For updates, we can directly apply \cref{alg:tour-tree}.
All of them can be parallelized and have $O(\log n)$ span by a divide-and-conquer approach to update tree paths.

\section{Implementation Details for Resizable Hash Tables}\label{app:resize-ht}

As mentioned in \cref{sec:implementation}, we use a resizable hash table \variablename{next\_frontier} to maintain the frontier in our sparse version of stepping algorithm implementations.
Whenever a vertex $v$ is to be added to the next frontier, we pick a slot $i$ uniformly at random, and attempt to write $v$ to \variablename{next\_frontier[i]} using \cas{}, and use linear probing if there is collision or conflict due to concurrency. During the process, we keep an estimation of the size of \variablename{next\_frontier}, and resize accordingly. For frontier size $s$, the space used by the hash table is always $O(s)$.
We use two hash tables, alternating them to be the current and next frontiers.

It is worth noting that, even in resizing we do not explicitly allocate memory or move data. We will start with allocating an array of size $n$, and adjust tail pointer of the array to control the current hash table size. When the load factor of the current hash table reaches a pre-defined constant, we enlarge the hash table. We do this by moving the tail pointer to be twice the current size, and let all future scatters go to the last half of the hash table.
For each step, we always start from a pre-defined \texttt{MIN\_SIZE} of the hash table. Since we only insert to the hash table, we do not need to shrink the hash table. We show a pseudocode below to illustrate the process.


\begin{verbatim}
hash_table {
  vertex* table;
  int offset, tail, est_size;
  double SAMPLE_RATE;
  int MIN_SIZE;

  hash_table(int n) {
    table = new vertex[n];
    init(); }

  increment_size() {
    if (flip_coin(SAMPLE_RATE)) fetch_and_add(&est_size);}

  init() {offset = 0; tail = MIN_SIZE;}

  insert(vertex u) {
    slot = random_slot(u) + offset;
    while(!CAS(&table[slot], empty, u))
      slot = next_slot(slot);
    increment_size();
    if (size > est_size) {
      offset = tail;
      tail = tail*2; }
  }
};

hash_table* current_frontier, next_frontier;

parallel_for (v in current_frontier) {
  parallel_for (u is v's out-neighbor) {
    if (v relaxes u) {
      next_frontier.insert(u);
    } } }
swap(current_frontier, next_frontier);
next_frontier->init();
\end{verbatim}

\hide{
\section{$k$-$\rho$ property for different graphs}
\label{sec:krho}

To better understand the properties of different graphs, we show the $\rho$-$k_\rho$ curves in \cref{fig:k-rho}. We note that computing the exact value of $k_\rho$ is expensive (which makes \radiuss{} impractical), and we estimate the value using 100 samples.
As mentioned, this indicates how much potential parallelism we can get on the original graph (i.e., without shortcuts). We show the value of $k_\rho$ when $\rho$ is $\log n$, $\sqrt{n}$, $n/\log n$, $n/10$, and $n$, respectively. Note that $k_n$ is the shortest path tree depth.

On all \sfgraph{s}, we observe that $k_n$ is $O(\log n)$ (about $2\log n$). Interestingly, all \sfgraph{s} are $(\log n,\sqrt{n})$-graph, which means that almost all vertices can reach their $\sqrt{n}$ nearest vertices by $\log n$ hops.
This is not surprising for \sfgraph{s}, which have some ``hubs'' that are well connected to other vertices.
These vertices are easy to be reached by any source in a few steps.
Once reached, they can reach a lot of other vertices, which quickly accumulate $\sqrt{n}$ nearest neighbors for any source.
The road graphs have a different $\rho$-$k_\rho$ pattern.
On GE and USA, it takes more than $100$ hops to reach $\sqrt{n}$ nearest vertices, and $O(\sqrt{n})$ steps to reach $O(n)$ nearest vertices.
We believe the result is reasonable since road graphs are (almost) planar graphs.
}

\section{Experiment on Different Implementations for \BDPQ{}}
\label{sec:exp:winningtree}
We also use a simple experiment to test the relative performance of array-based and \tourt-based \BDPQ{}.
We note that for \lazyinsert{}, our array-based implementation uses a hash-table-like data structure, which is only $O(1)$ work per insertion, while \tourt{} needs up to logarithmic work.
For \extract{}, the array-based implementation needs $O(n)$ time, while \tourt{} needs $O(\rho\log(n/\rho))$ time to extract the top $\rho$ \record{s}.
To estimate which one will give better performance in practice, we simulate the \extract{} function and compare the performance for the two data structures.

\begin{figure}[th]
  \centering
  \includegraphics[width=\columnwidth]{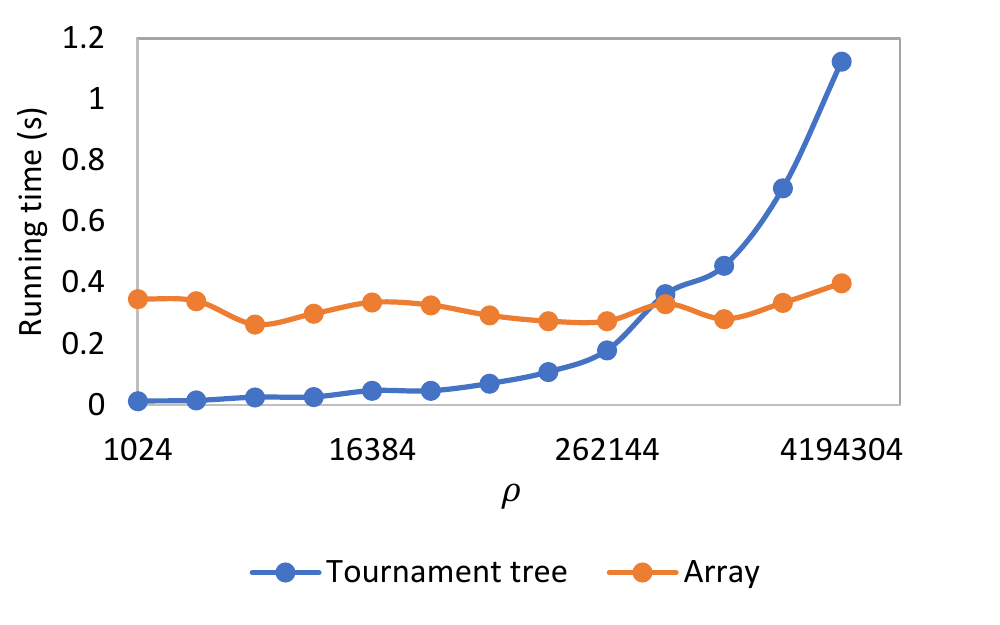}
  \caption{\textbf{Running time of extracting $\rho$ elements from \tourt{}- and array- based implementation of \BDPQ{}, respectively.}}\label{fig:tournament}
\end{figure}

To test the performance of \extract{}, we vary the value of $\rho$. We first initialize the priority queue with $n=10^8$ \record{s}, which is about the same order of magnitude as the number of vertices in the graphs we tested.  Our \tourt{} is also implemented in flat arrays, avoiding expensive pointer-based structure.
We accumulate 10 runs of \extract{}. The result is shown in \cref{fig:tournament}. The running time of array-based implementation is stable and is always around 0.35s. This is because the cost of array-based \BDPQ{} is $O(n)$, which is not affected by the value of $\rho$. The cost of the \tourt{} increases with the value of $\rho$.
When $\rho$ is larger than $2^{19}$, the array-based implementation achieves better performance.
We note that the setting is advantageous to \tourt{}-based implementation since we do not consider the update cost ($O(1)$ for arrays and $O(\log n)$ for \tourt{s}), and we always assume the frontier size to be large ($10^8$).
For the social networks we tested, the value of $\rho$ is usually larger than $2^{19}$. Therefore we always use array-based implementation of \BDPQ{}. For the road graphs, the value of $\rho$ is slightly below $2^{19}$. As mentioned, this experiment does not consider update cost. Also, on road networks, the frontier size is also smaller (much smaller than $10^8$), which is favorable to array-based implementation. Based on these reasons, we choose to use the array-based data structures to implement our stepping algorithms.
From \cref{fig:visted-per-step}, for \deltas{}, in the dense rounds (which dominate total running time), the number of vertices to be extracted is also reasonably large.
It is an interesting future work to see if \tourt{}-based implementation can be more efficient on certain graphs that prefer a small $\rho$ in \ssspours{}.

\section{Experiment on different source vertices}

In previous experiments, we show the average running time on multiple source vertices on each graph.
Now we show if different algorithms perform differently from the sources on the same graph.
To do so, we randomly pick 100 source vertices in Friendster and Twitter, and test the performance from each of them.
The results in shown in \cref{fig:source}.
Although the running time on these vertices and graphs differs, we can see that \ssspours{} is consistently faster than the other two implementations, except for four vertices in Twitter.
We conclude that starting from fringe vertices and hub vertices does not make a significant difference in the relative performance between Bellman-Ford, $\Delta^*$-stepping, and \ssspours{}.

\begin{figure}[ht]
  \centering
  \begin{tabular}{c}
    \includegraphics[width=\columnwidth]{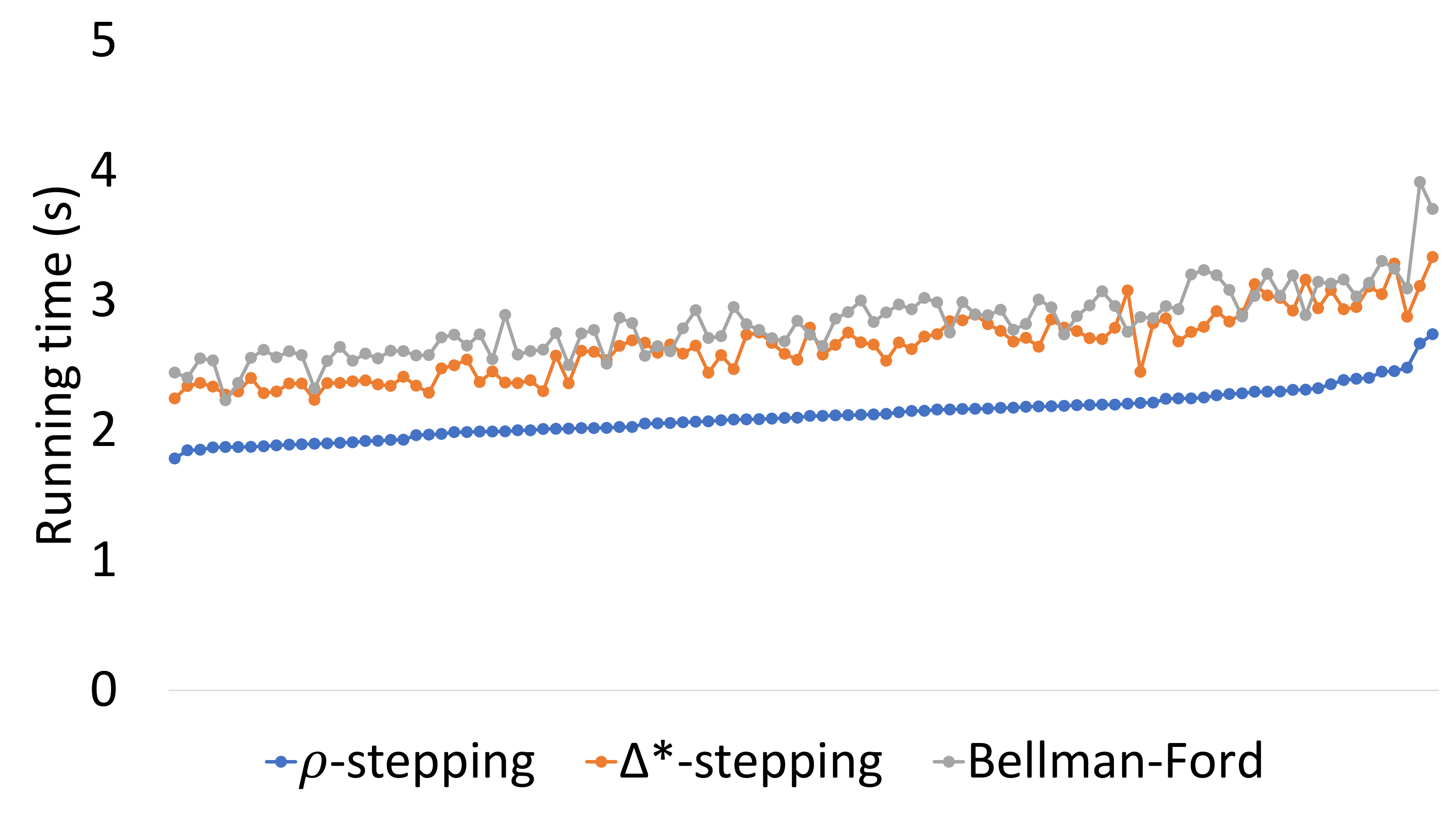} \\
    \bf (a). FT  \\
    \includegraphics[width=\columnwidth]{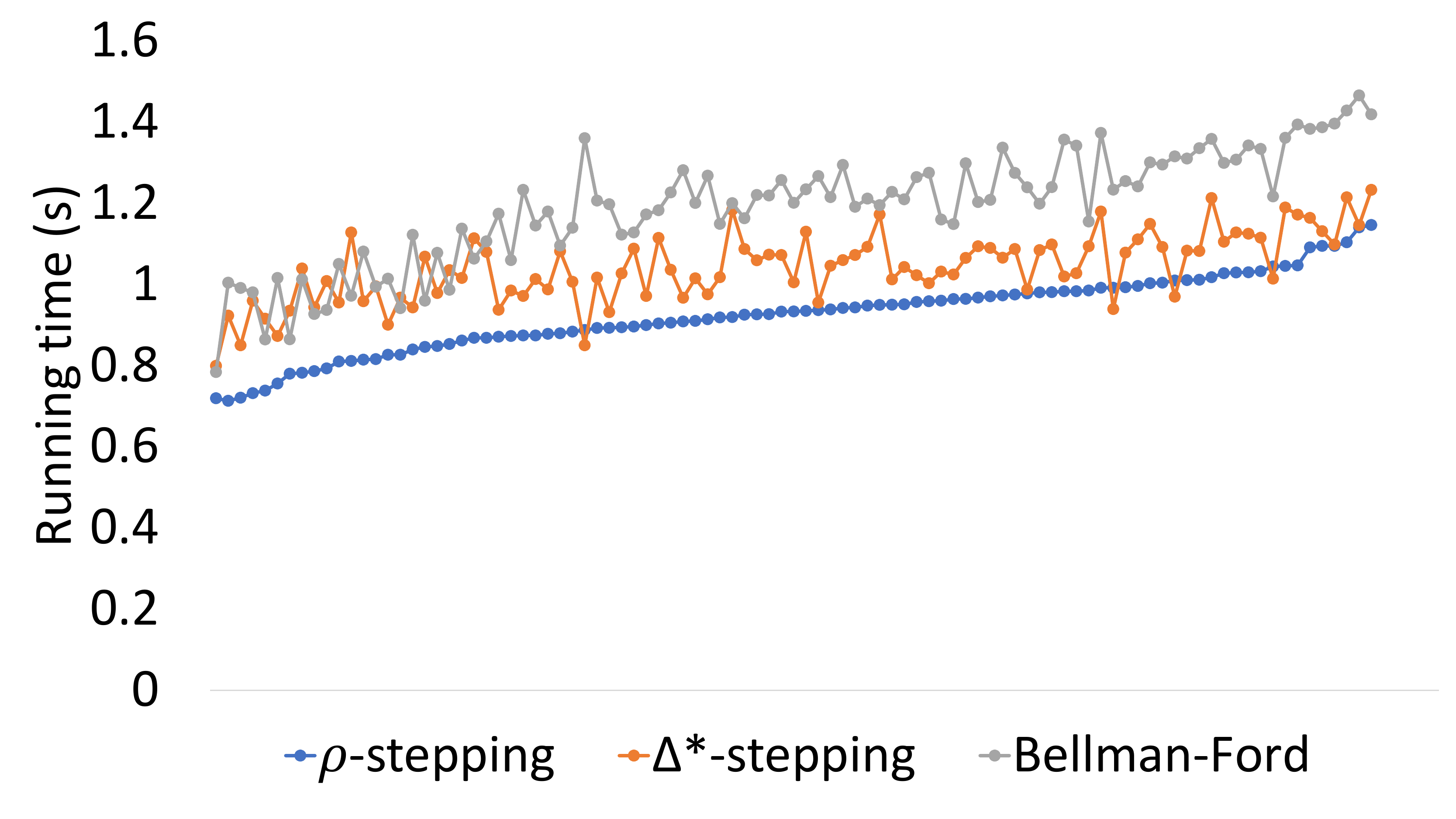} \\
    \bf (b). TW \\
  \end{tabular}
  \caption{\textbf{Running time on difference source vertices.}}\label{fig:source}
\end{figure}

\section{Experiment on a different machine}

We also experiment on whether different machines affect the best $\rho$ choice. We launched a virtual machine in Amazon Web Services (AWS).
This single-socket machine is equipped with Intel Xeon Platinum 8175M CPUs with a total of 24 cores (48 hyperthreads). It has 192GB of main memory and 33MB L3 cache.
Other setting is the same as the previous experiments as mentioned in \cref{sec:exp:setup}. We present the result in \cref{fig:AWS-rho-time}.
Surprisingly, we found out that taking $\rho=2^{21}$ (or any value between $2^{20}$ to $2^{22}$) is still a good choice for all graphs for this different machine, although it has
much fewer cores.
Comparing \cref{fig:AWS-rho-time,fig:rho-time}, the machine with fewer cores shows a less stable performance especially when $\rho$ is small.
We note that there are also other parameters in the implementation that may affect the pattern of the relative performance of different $\rho$ values (e.g., the sparse-dense threshold).
In general, we do note that different hardware settings (e.g., cache size, number of sockets, number of threads) can affect the relative performance with different values of $\rho$, which is more significant for the smaller values of $\rho$.

\begin{figure}[ht]
  \centering
  \includegraphics[width=\columnwidth]{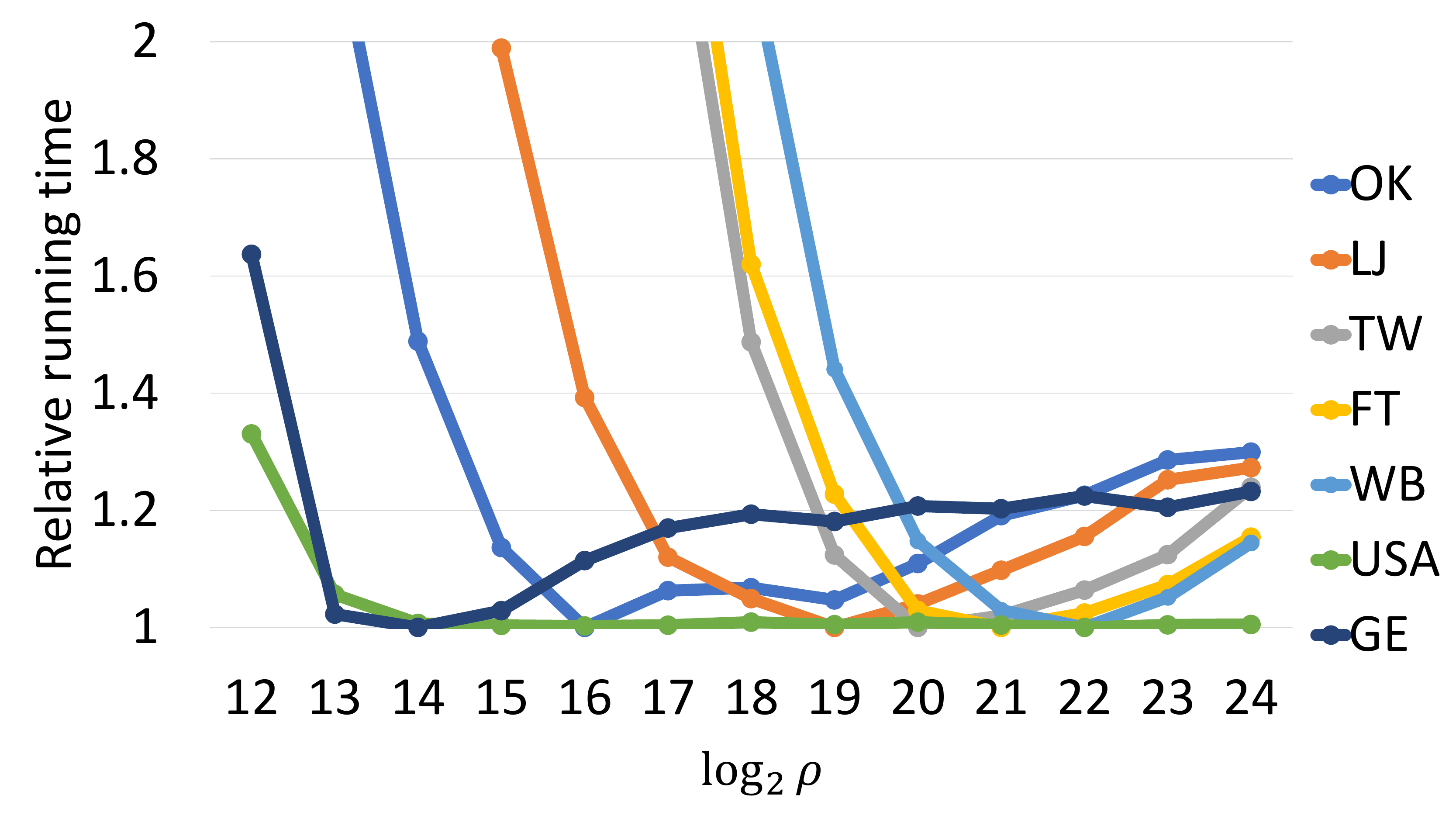}
  \caption{\textbf{Relative running time of \ssspours{} with varied $\rho$ on the AWS machine.}}\label{fig:AWS-rho-time}
\end{figure}

\begin{figure*}[t]
  \centering
  \begin{tabular}{cccc}
    \multicolumn{4}{c}{\bf Friendster (FT):}\\
    \bf GAPBS &\bf Galois&\bf Julinne& \bf Ours (\ourdelta)\\
    \includegraphics[width=0.5\columnwidth]{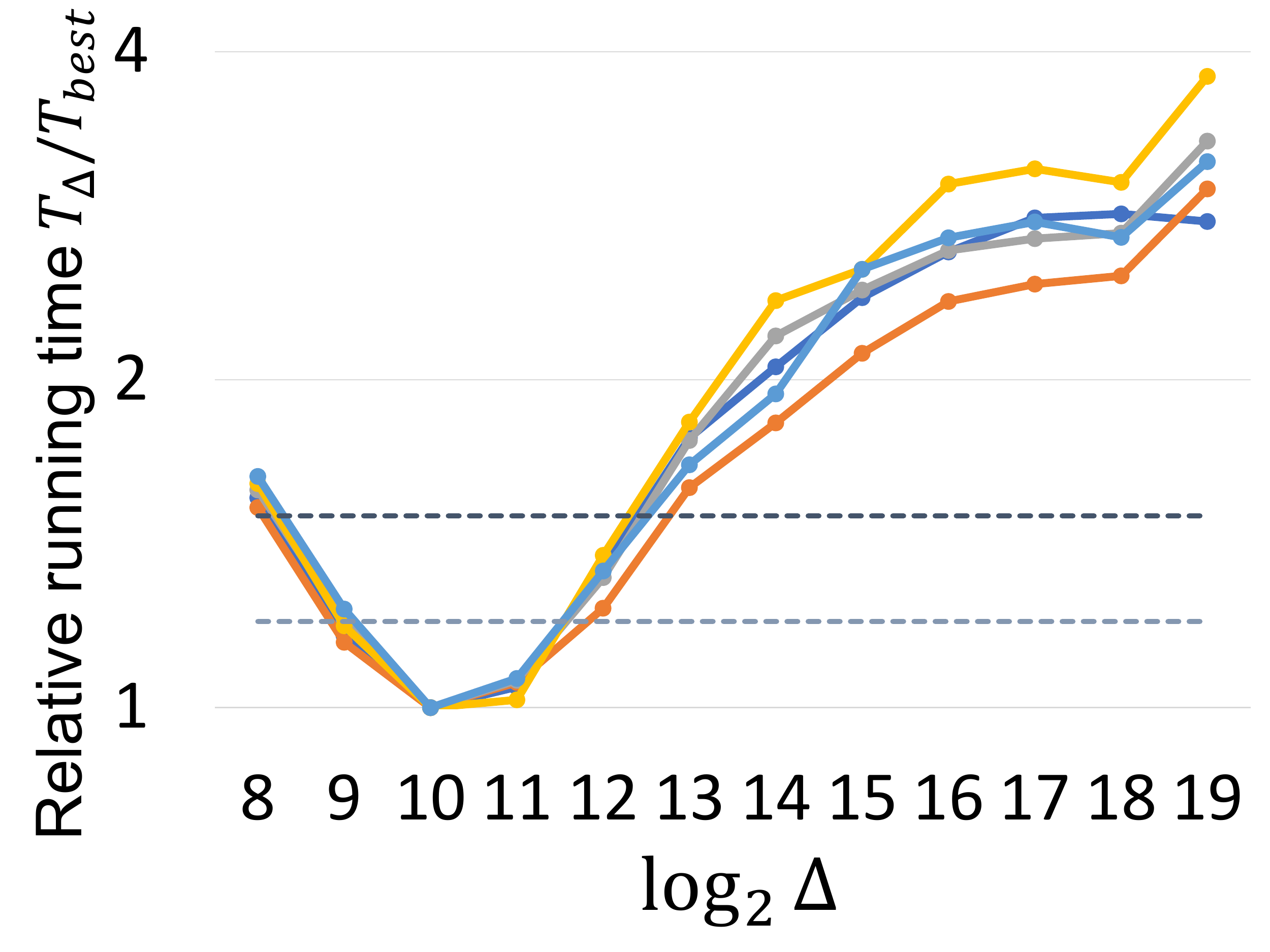} &
    \includegraphics[width=0.5\columnwidth]{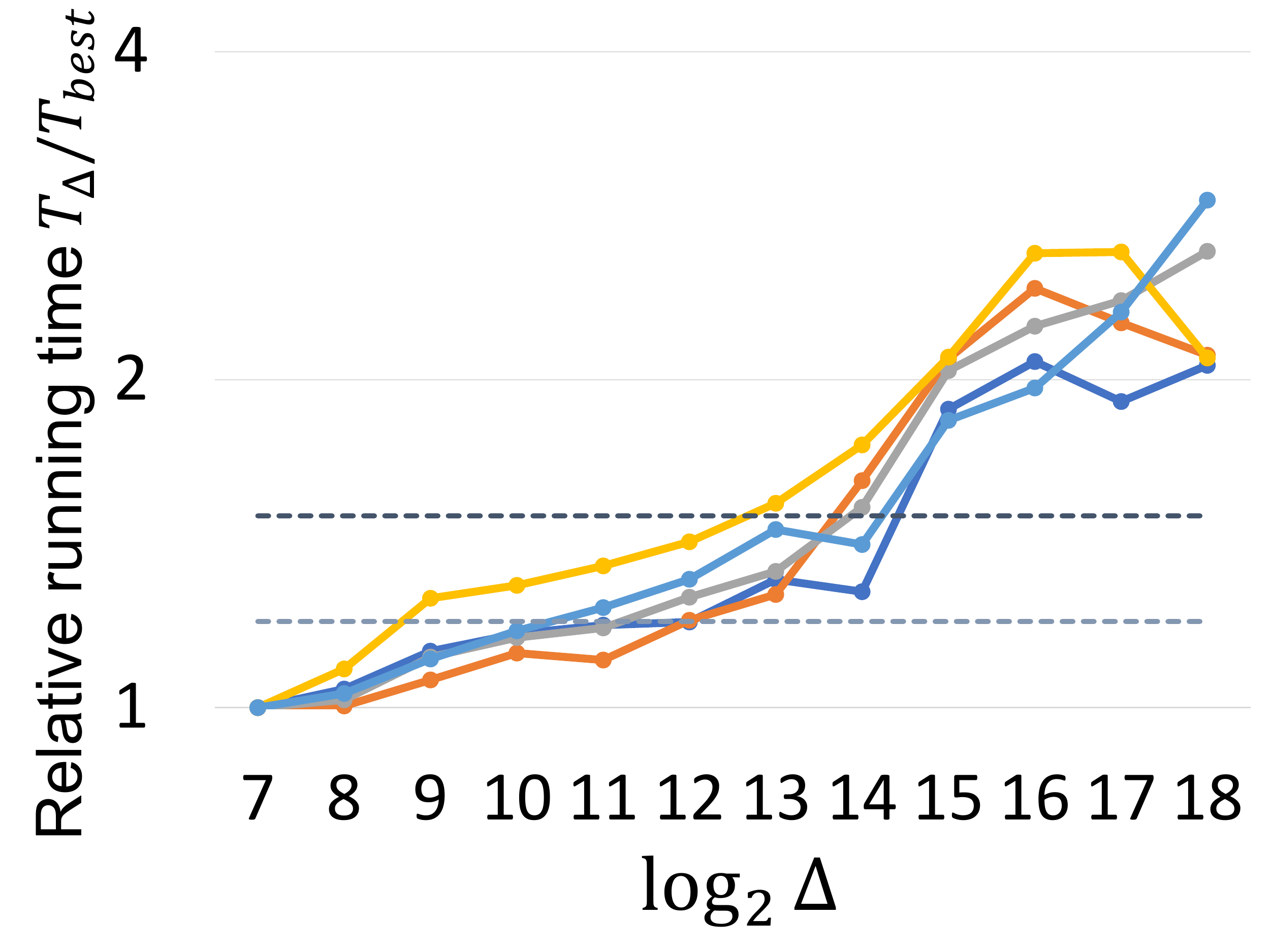} &
    \includegraphics[width=0.5\columnwidth]{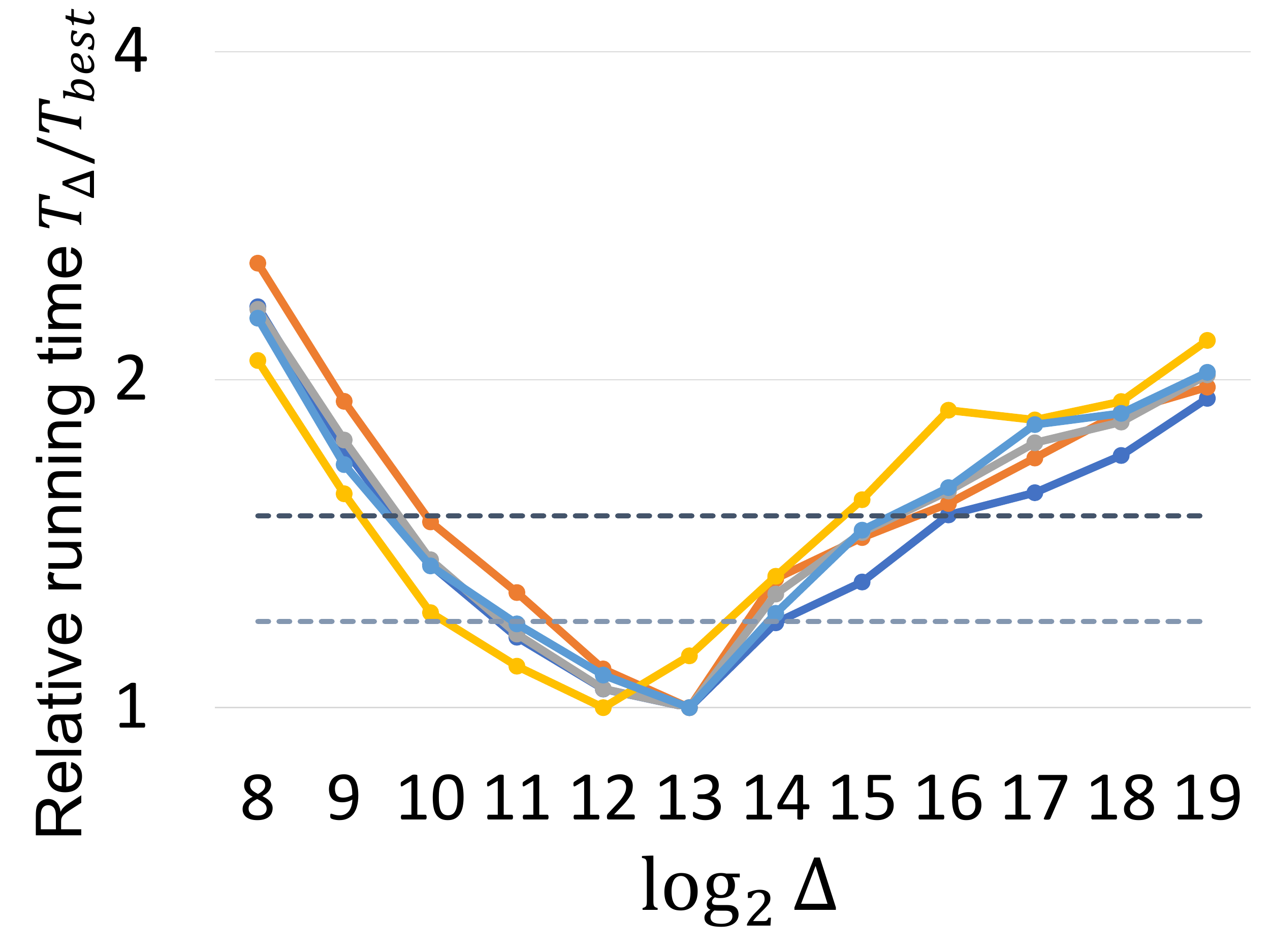} &
    \includegraphics[width=0.5\columnwidth]{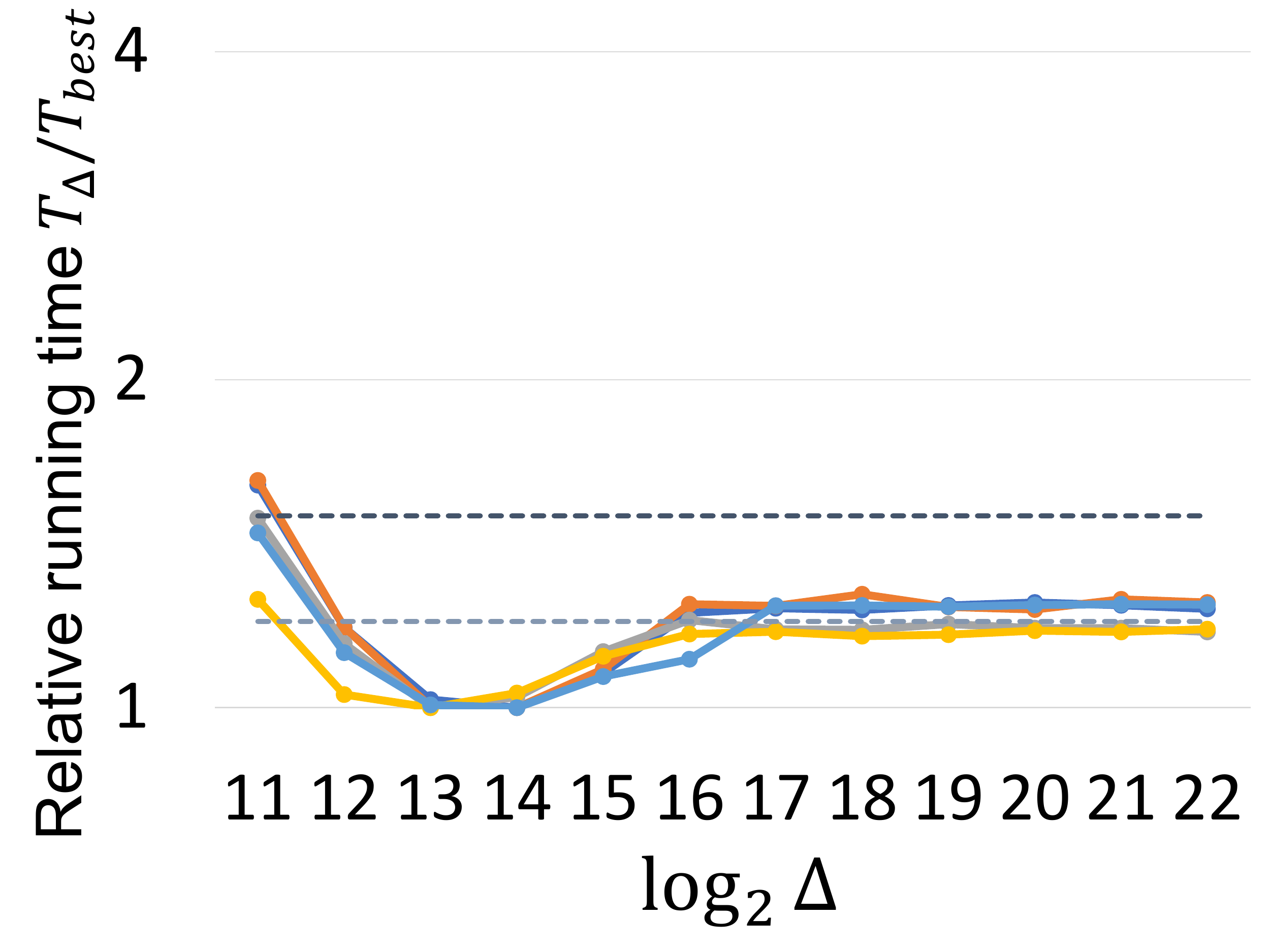}
      \\
    \multicolumn{4}{c}{\bf WebGraph (WB):}\\
    \bf GAPBS &\bf Galois&\bf Julinne& \bf Ours (\ourdelta)\\
    \includegraphics[width=0.5\columnwidth]{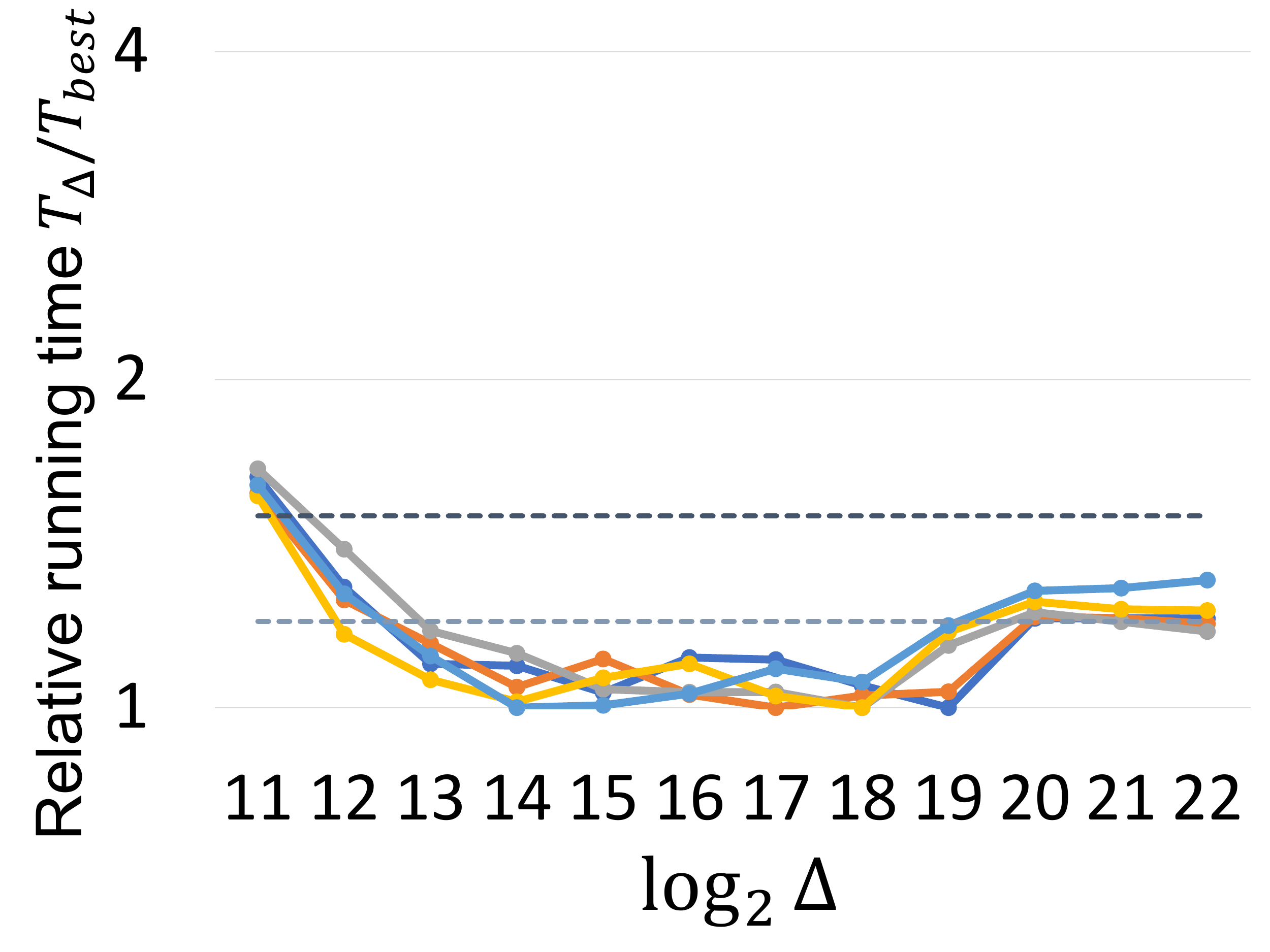} &
    \includegraphics[width=0.5\columnwidth]{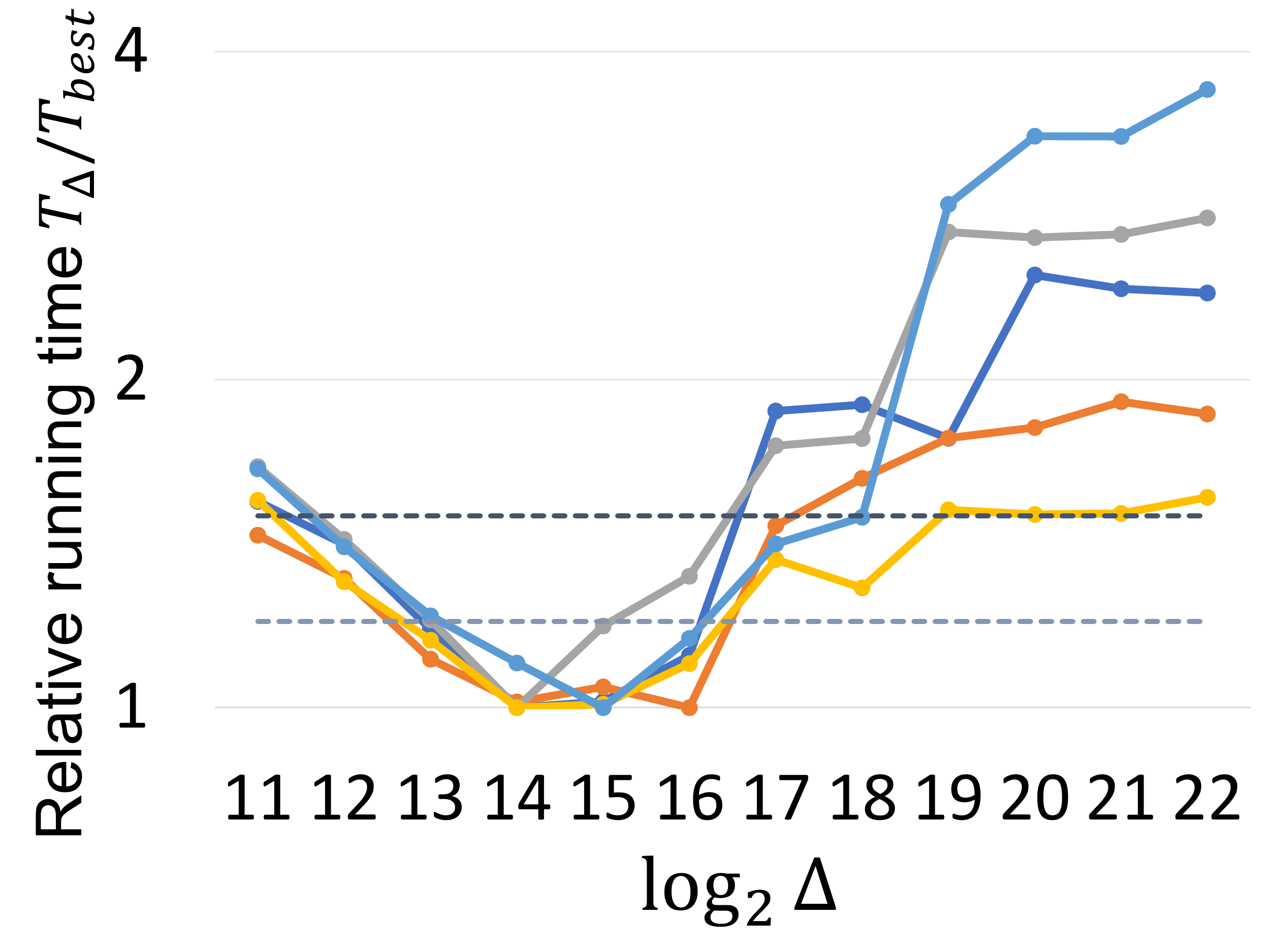} &
    \includegraphics[width=0.5\columnwidth]{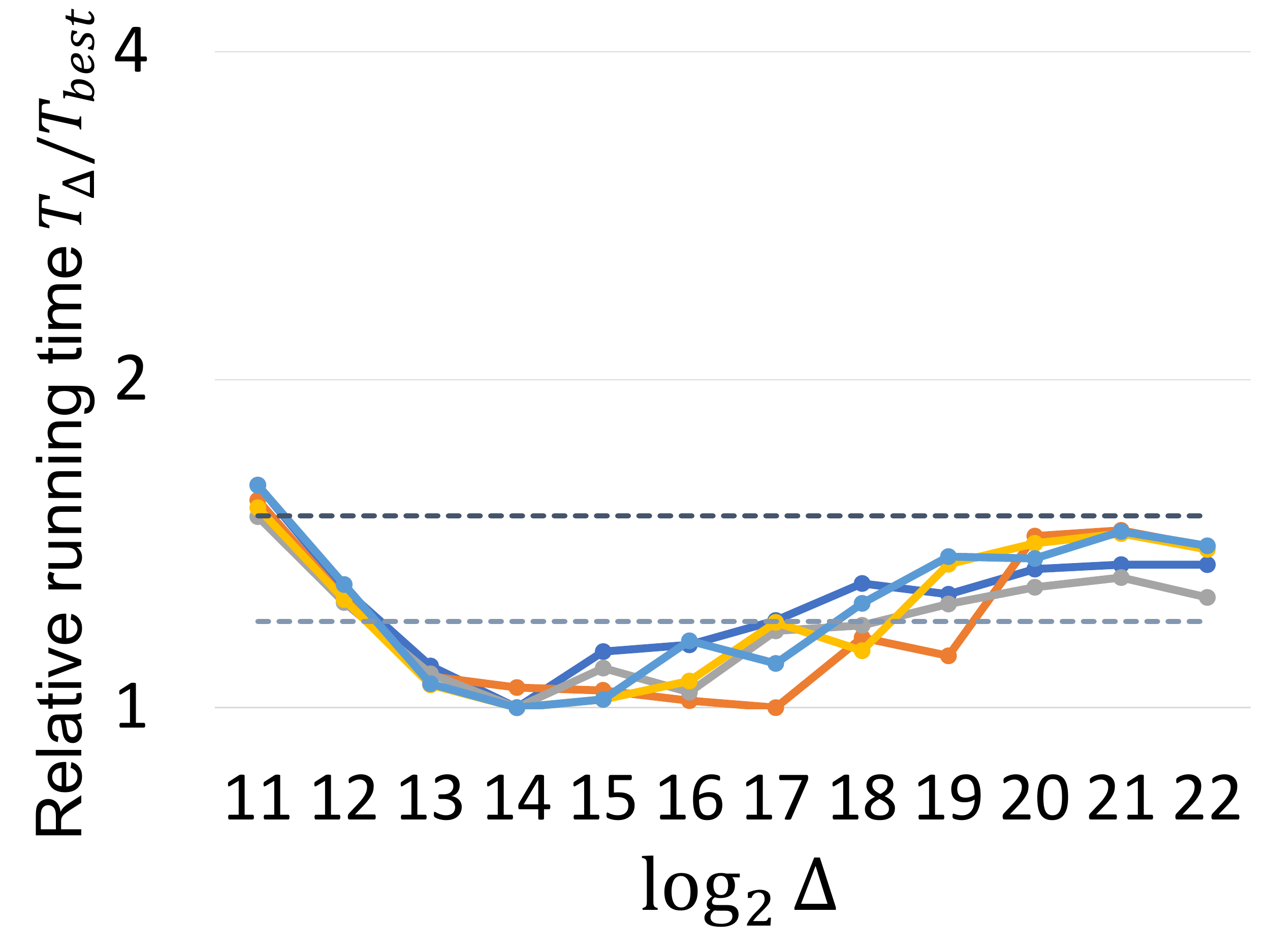} &
    \includegraphics[width=0.5\columnwidth]{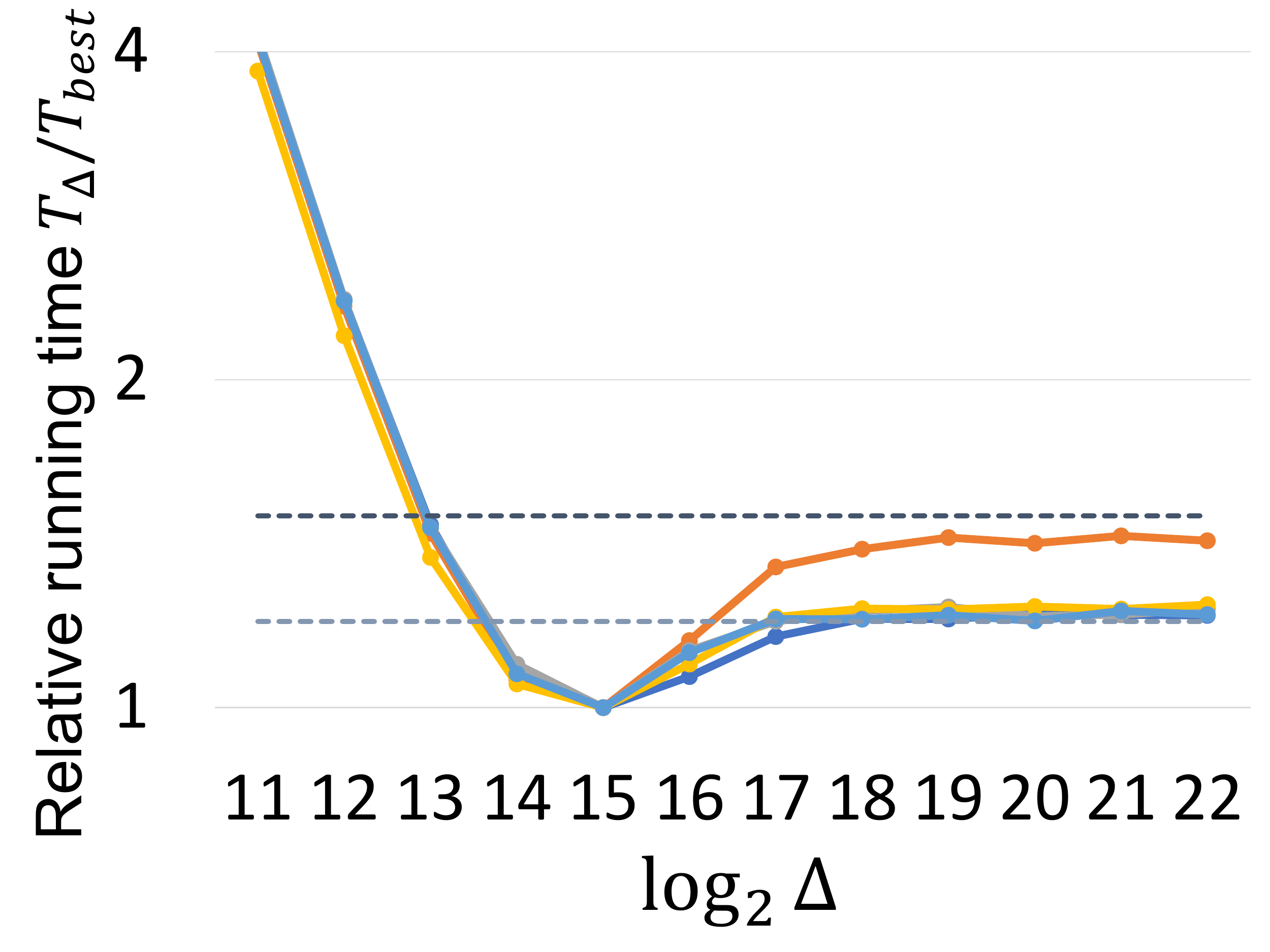}
      \\
  \end{tabular}
  \caption{\textbf{Relative running time with varying $\Delta$ of multiple \deltas{} implementations on different sources.} The five lines with different colors represent five sources. For each source, we normalize the running time to the corresponding fastest time. For each of the implementations, we use a window of $2^{11}$ around the best value of $\Delta$.  The best value of $\Delta$ is relatively stable for all implementations, except for a few instances (e.g., GAPBS and Julienne on WebGraph).
  Generally speaking, the best $\Delta$ for one source node makes another node at most 20\% slower.
  Also, it seems the result on directed graphs (WB) is more unstable than undirected graphs (FT).
  On both WB and FT, our \ourdelta{} based on \deltass{} shows very stable performance.
  }\label{fig:differentsource}
\end{figure*}

\begin{figure*}[t]
  \begin{tabular}{cccc}
    \multicolumn{4}{c}{\includegraphics[width=1.3\columnwidth]{figures-new/delta/caption.pdf}}\\
    \includegraphics[width=0.5\columnwidth]{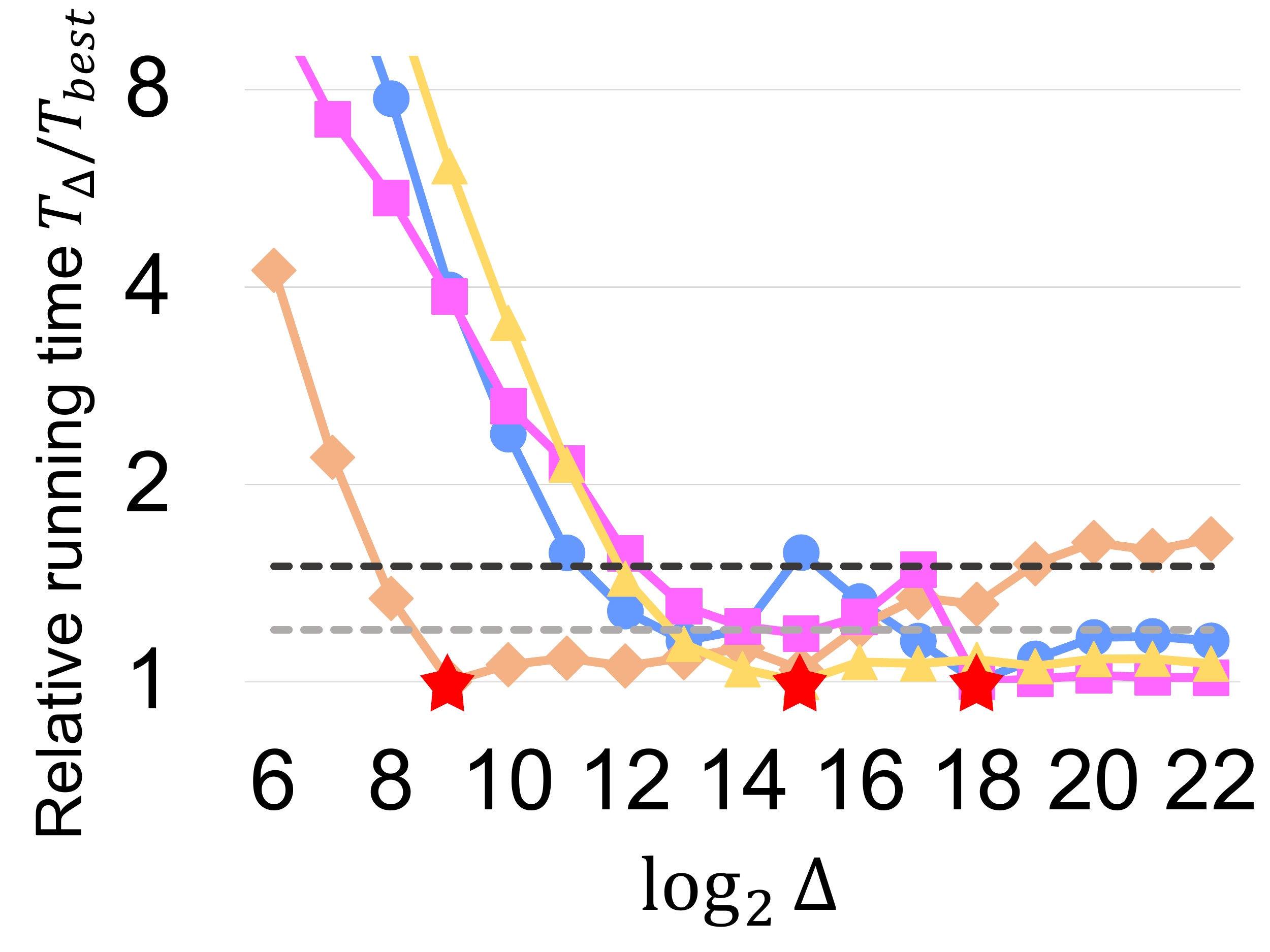} & \includegraphics[width=0.5\columnwidth]{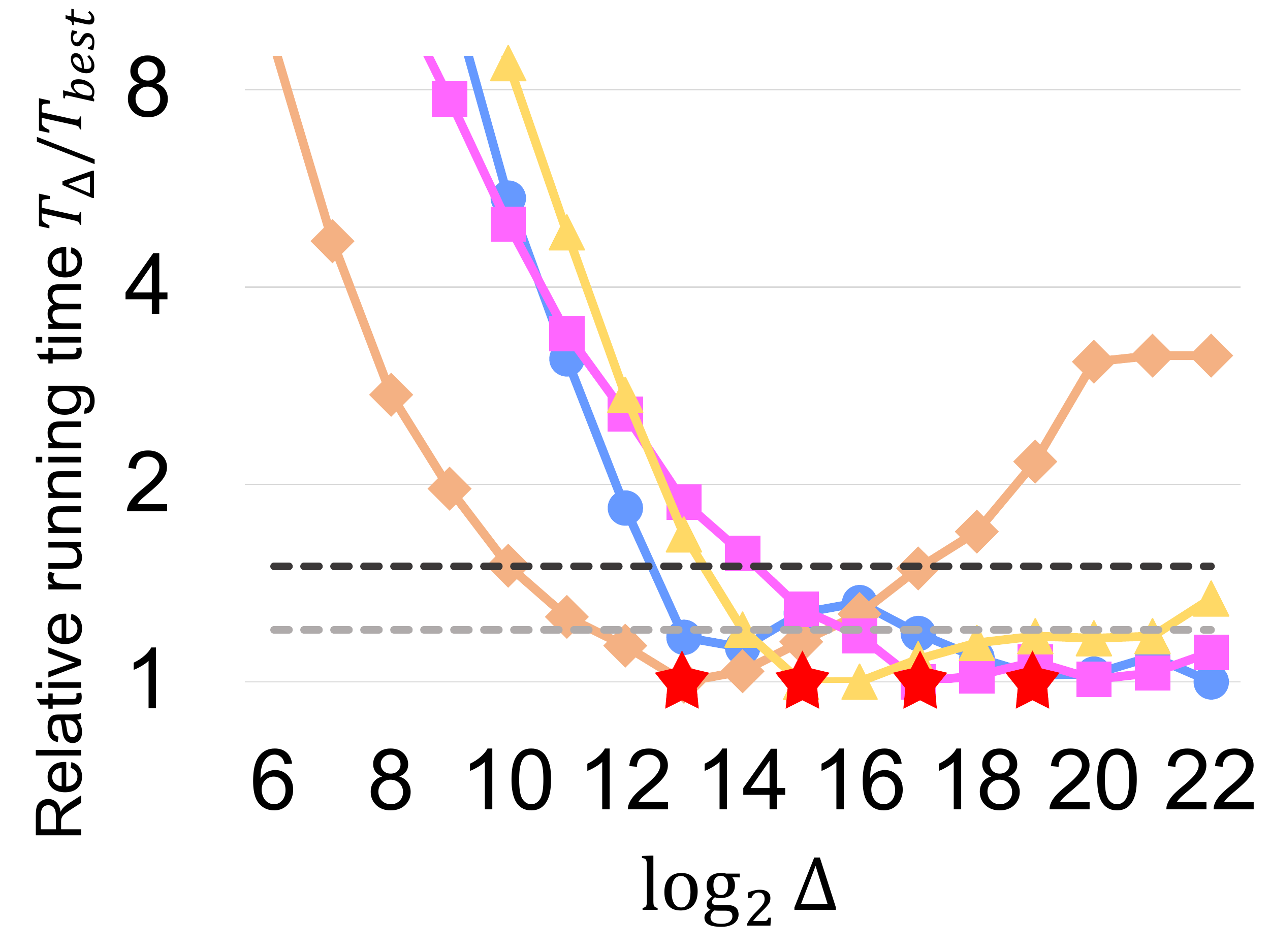} & \includegraphics[width=0.5\columnwidth]{figures-new/delta/delta_twitter.pdf} &
    \includegraphics[width=0.5\columnwidth]{figures-new/delta/delta_friendster.pdf} \\
\bf (a). OK& \bf (b). LJ & \bf (c). TW & \bf (d). FT \\
  \end{tabular}
  \begin{tabular}{ccc}
\includegraphics[width=0.5\columnwidth]{figures-new/delta/delta_web.pdf} & \includegraphics[width=0.5\columnwidth]{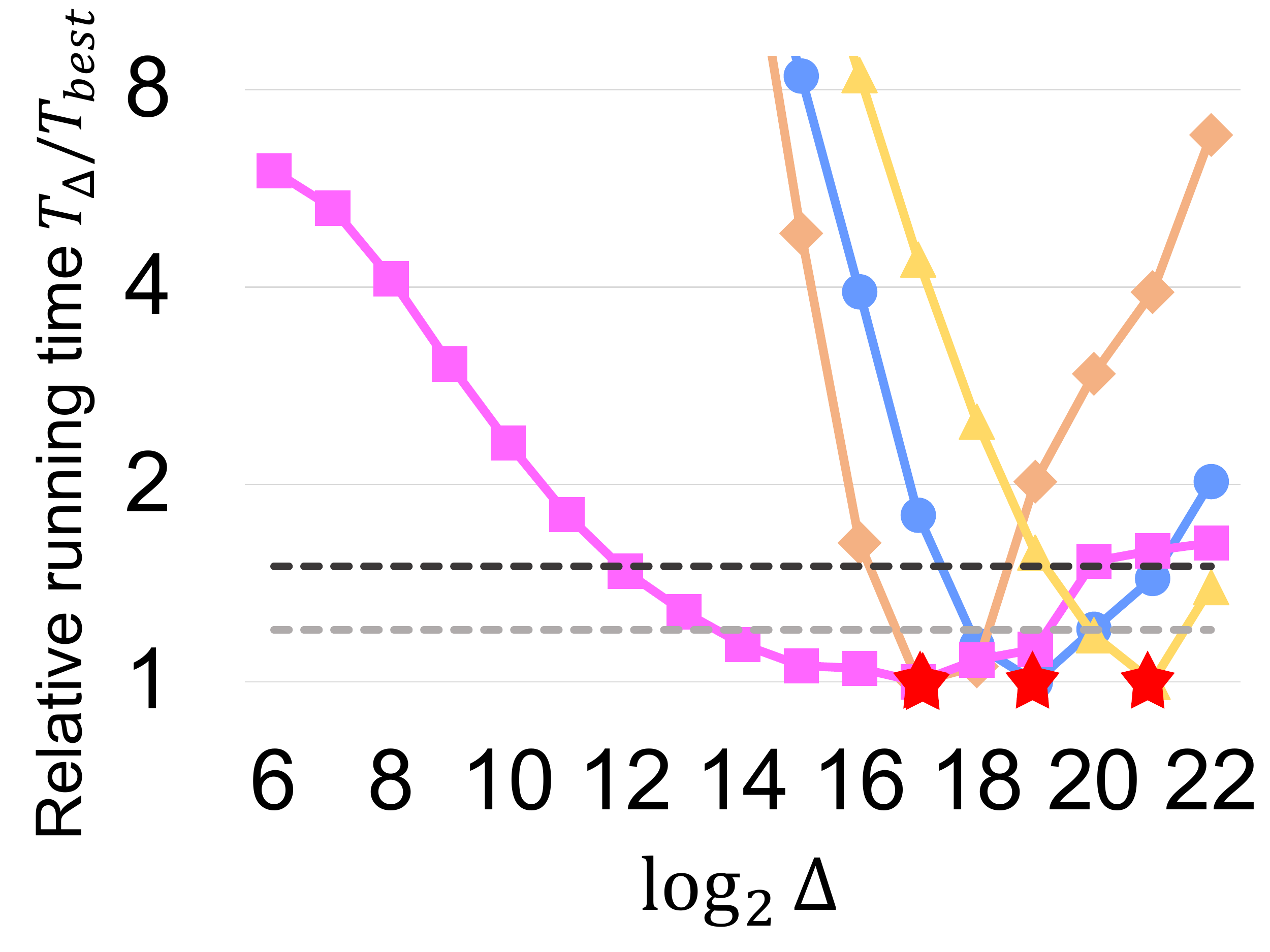} &
\includegraphics[width=0.5\columnwidth]{figures-new/delta/delta_US.pdf}\\
\bf (e). WB&\bf (f). GE&\bf (g). USA\\
  \end{tabular}
  \caption{\textbf{$\Delta$-stepping relative running time with varying $\Delta$.} We use 96 cores (192 hyperthreads).}\label{fig:alldelta}

\end{figure*}

\begin{figure*}
\begin{tabular}{cccc}
  \includegraphics[width=0.5\columnwidth]{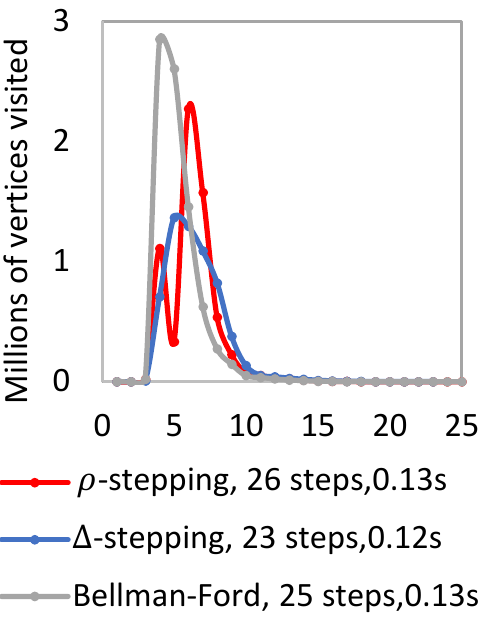}&
  \includegraphics[width=0.5\columnwidth]{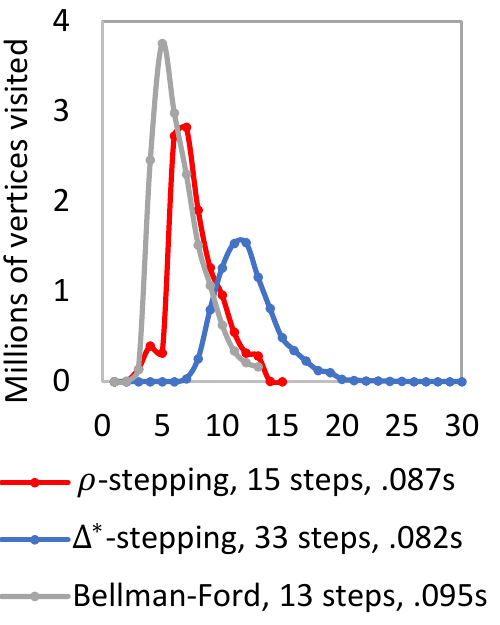}&
  \includegraphics[width=0.5\columnwidth]{figures-new/per-round/TW.pdf}&
  \includegraphics[width=0.5\columnwidth]{figures-new/per-round/FT.pdf}\\
  \includegraphics[width=0.5\columnwidth]{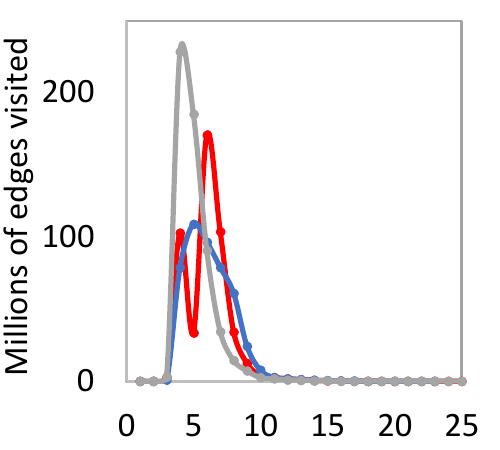}&
  \includegraphics[width=0.5\columnwidth]{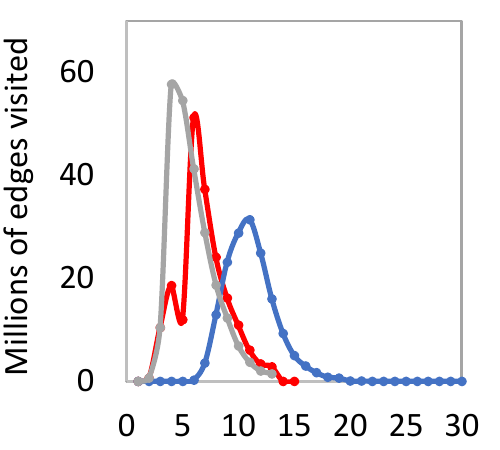}&
  \includegraphics[width=0.5\columnwidth]{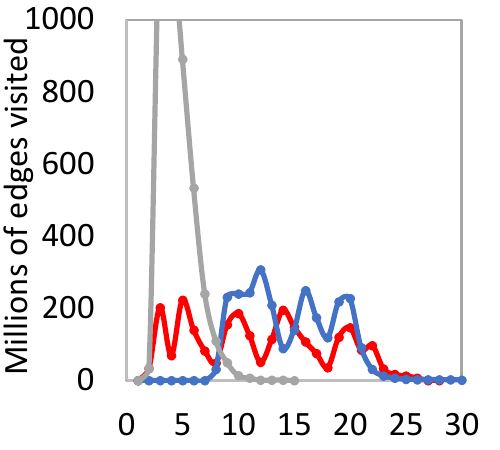}&
  \includegraphics[width=0.5\columnwidth]{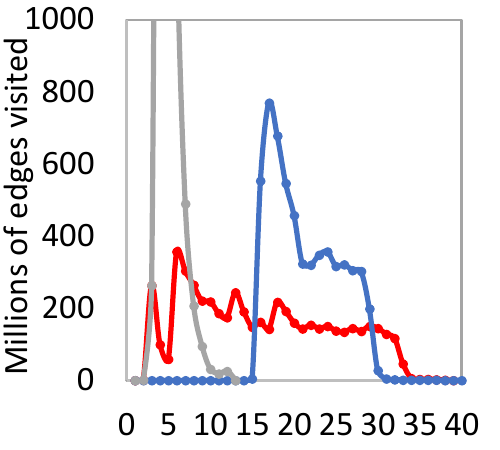}\\
\bf (a). OK& \bf (b). LJ & \bf (c). TW& \bf (d). FT \\
\end{tabular}
\begin{tabular}{ccc}
  \includegraphics[width=0.5\columnwidth]{figures-new/per-round/WB.pdf}&
  \includegraphics[width=0.5\columnwidth]{figures-new/per-round/USA.pdf}&
  \includegraphics[width=0.5\columnwidth]{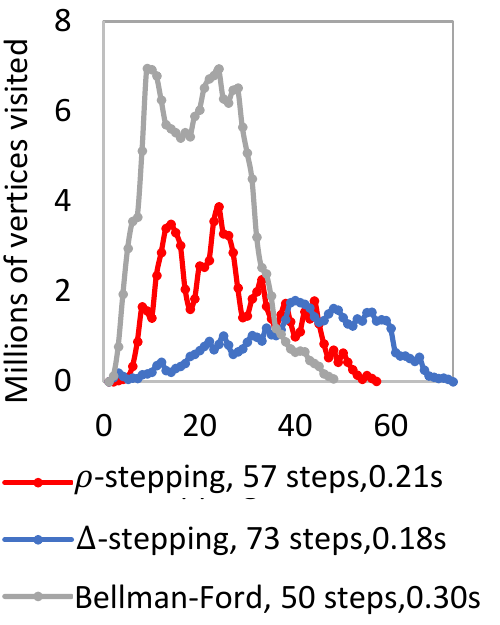}\\
  \includegraphics[width=0.5\columnwidth]{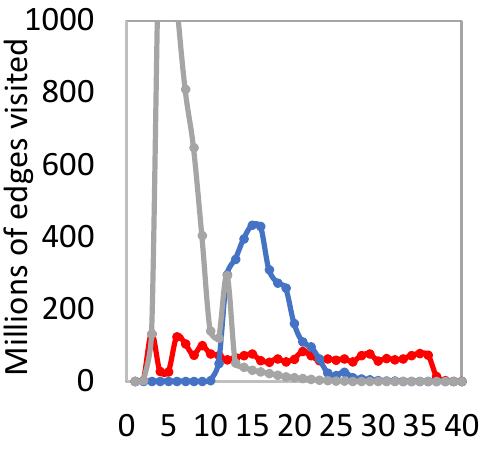}&
  \includegraphics[width=0.5\columnwidth]{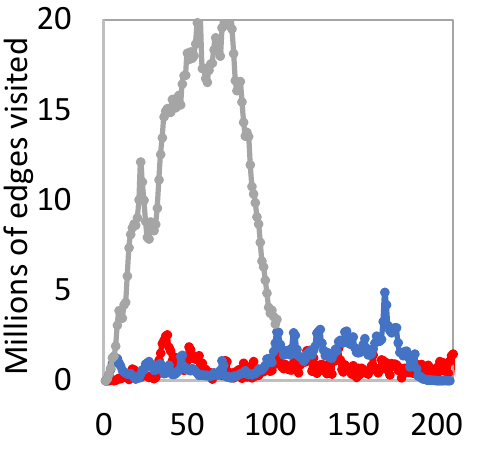}&
  \includegraphics[width=0.5\columnwidth]{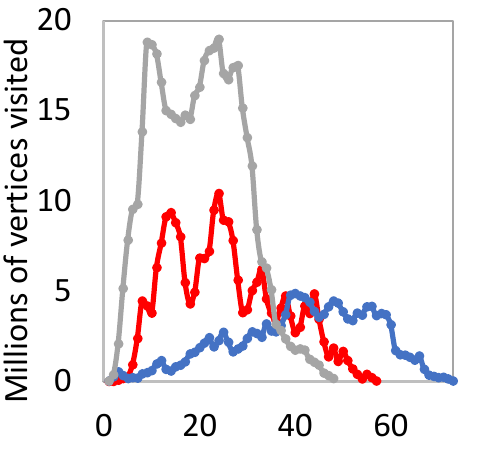}\\
\bf (e). WB& \bf (f). USA & \bf (g). GE\\
\end{tabular}
\caption{\textbf{Number of visited vertices in each step in \ourrho, \ourdelta and \ourbf.} \mdseries Here we only run on one source vertex, since it has unclear meaning to compute the average of multiple runs on each step.  Hence, the runtimes can be different from Table \ref{tab:alltime} (average on 100 runs from 10 source vertices), and some curves are bumpy. We use 96 cores (192 hyperthreads).\label{fig:visted-per-step-full}}
\end{figure*}

\fi

\end{document}